\documentclass[11pt]{article}

\usepackage[margin=1in]{geometry}
\usepackage{fullpage}
\usepackage{tgtermes}
\usepackage[T1]{fontenc}
\usepackage[colorlinks,citecolor=blue,linkcolor=blue,urlcolor=red]{hyperref}
\usepackage{graphicx}
\usepackage{amsfonts,amsmath,amsthm,amssymb,dsfont,mathtools}
\mathtoolsset{centercolon}
\usepackage{xfrac,nicefrac}
\usepackage{mathdots}
\usepackage{bm,bbm}
\usepackage{url}
\usepackage{caption}
\usepackage{paralist}
\usepackage{enumerate}
\usepackage[normalem]{ulem}
\usepackage{enumitem}
\usepackage{xspace}
\xspaceaddexceptions{]\}}
\usepackage[capitalise]{cleveref}
\crefname{algocf}{Algorithm}{Algorithms}
\Crefname{algocf}{Algorithm}{Algorithms}
\usepackage{tabu}
\usepackage{framed}
\usepackage{float,wrapfig}
\usepackage{subfigure}
\usepackage{soul}
\usepackage[usenames,dvipsnames]{pstricks}
\usepackage{thmtools, thm-restate}
\usepackage[linesnumbered,boxed,ruled,vlined]{algorithm2e}
\usepackage{algpseudocode}
\usepackage{tikz}
\usetikzlibrary{decorations.pathreplacing}
\usetikzlibrary{calc}
\usepackage{regexpatch}
\usetikzlibrary{positioning}
\usetikzlibrary{arrows.meta}

\theoremstyle{plain}

\newtheorem{theorem}{Theorem}[section]
\newtheorem{lemma}[theorem]{Lemma}

\newtheorem{prop}[theorem]{Proposition}

\newtheorem{claim}[theorem]{Claim}
\theoremstyle{definition}
\newtheorem{definition}[theorem]{Definition}
\newtheorem{remark}[theorem]{Remark}
\newtheorem{property}[theorem]{Property}

\DeclarePairedDelimiter{\tvdbasic}{\lVert}{\rVert}
\makeatletter
\newcommand{\@tvdstar}[2]{\tvdbasic*{#1 - #2}_{\mathrm{TV}}}
\newcommand{\@tvdnostar}[3][]{\tvdbasic[#1]{#2 - #3}_{\mathrm{TV}}}
\newcommand{\tvd}{\@ifstar\@tvdstar\@tvdnostar}
\makeatother

\usepackage[breakable]{tcolorbox}

\renewcommand{\epsilon}{\varepsilon}

\newcommand{\ceil}[1]{\left\lceil #1 \right\rceil}

\newcommand{\norm}[1]{\left\lVert #1 \right\rVert}
\newcommand{\midnorm}[1]{\lVert #1 \rVert}

\newcommand{\bk}[1]{\left( #1 \right)}

\newcommand{\bigbk}[1]{\bigl( #1 \bigr)}

\newcommand{\Bk}[1]{\left[ #1 \right]}

\newcommand{\bigBk}[1]{\bigl[ #1 \bigr]}

\newcommand{\BK}[1]{\left\{ #1 \right\}}
\newcommand{\midBK}[1]{\{ #1 \}}
\newcommand{\bigBK}[1]{\bigl\{ #1 \bigr\}}
\newcommand{\BigBK}[1]{\Bigl\{ #1 \Bigr\}}

\newcommand{\angbk}[1]{\left\langle #1 \right\rangle}

\newcommand{\abs}[1]{\left\lvert #1 \right\rvert}

\newcommand{\bigabs}[1]{\bigl\lvert #1 \bigr\rvert}

\DeclareMathOperator*{\E}{\mathbb{E}}

\let\Pr\myPr
\DeclareMathOperator{\poly}{poly}
\DeclareMathOperator*{\argmax}{arg\,max}

\newcommand{\F}{\mathbb{F}}

\renewcommand{\tilde}{\widetilde}

\newcommand{\defeq}{\coloneqq}
\newcommand{\eps}{\varepsilon}
\newcommand{\T}{\mathcal{T}}
\newcommand{\N}{\mathbb{N}}
\newcommand{\R}{\mathbb{R}}
\newcommand{\Z}{\mathbb{Z}}
\renewcommand{\l}{\ell}

\newcommand{\numberthis}{\addtocounter{equation}{1}\tag{\theequation}}

\newcommand{\tall}{\vphantom{\sum}}

\usepackage{mleftright}
\let\left\mleft
\let\right\mright

\newcommand{\CW}{\mathrm{CW}}

\newcommand{\degen}{\unrhd}

\renewcommand{\split}{\textsf{\textup{split}}}
\newcommand{\Split}{\textsf{\textup{split}}}

\newcommand{\calG}{\mathcal{G}}
\newcommand{\calS}{\mathcal{S}}

\makeatletter
\let\save@mathaccent\mathaccent
\newcommand*\if@single[3]{%
  \setbox0\hbox{${\mathaccent"0362{#1}}^H$}%
  \setbox2\hbox{${\mathaccent"0362{\kern0pt#1}}^H$}%
  \ifdim\ht0=\ht2 #3\else #2\fi
  }
\newcommand*\rel@kern[1]{\kern#1\dimexpr\macc@kerna}
\newcommand*\widebar[1]{\@ifnextchar^{{\wide@bar{#1}{0}}}{\wide@bar{#1}{1}}}
\newcommand*\wide@bar[2]{\if@single{#1}{\wide@bar@{#1}{#2}{1}}{\wide@bar@{#1}{#2}{2}}}
\newcommand*\wide@bar@[3]{%
  \begingroup
  \def\mathaccent##1##2{%
    \let\mathaccent\save@mathaccent
    \if#32 \let\macc@nucleus\first@char \fi
    \setbox\z@\hbox{$\macc@style{\macc@nucleus}_{}$}%
    \setbox\tw@\hbox{$\macc@style{\macc@nucleus}{}_{}$}%
    \dimen@\wd\tw@
    \advance\dimen@-\wd\z@
    \divide\dimen@ 3
    \@tempdima\wd\tw@
    \advance\@tempdima-\scriptspace
    \divide\@tempdima 10
    \advance\dimen@-\@tempdima
    \ifdim\dimen@>\z@ \dimen@0pt\fi
    \rel@kern{0.6}\kern-\dimen@
    \if#31
      \overline{\rel@kern{-0.6}\kern\dimen@\macc@nucleus\rel@kern{0.4}\kern\dimen@}%
      \advance\dimen@0.4\dimexpr\macc@kerna
      \let\final@kern#2%
      \ifdim\dimen@<\z@ \let\final@kern1\fi
      \if\final@kern1 \kern-\dimen@\fi
    \else
      \overline{\rel@kern{-0.6}\kern\dimen@#1}%
    \fi
  }%
  \macc@depth\@ne
  \let\math@bgroup\@empty \let\math@egroup\macc@set@skewchar
  \mathsurround\z@ \frozen@everymath{\mathgroup\macc@group\relax}%
  \macc@set@skewchar\relax
  \let\mathaccentV\macc@nested@a
  \if#31
    \macc@nested@a\relax111{#1}%
  \else
    \def\gobble@till@marker##1\endmarker{}%
    \futurelet\first@char\gobble@till@marker#1\endmarker
    \ifcat\noexpand\first@char A\else
      \def\first@char{}%
    \fi
    \macc@nested@a\relax111{\first@char}%
  \fi
  \endgroup
}
\makeatother

\makeatletter
\xpatchcmd\thmt@restatable{%
\csname #2\@xa\endcsname\ifx\@nx#1\@nx\else[{#1}]\fi
}{%
\ifthmt@thisistheone
\csname #2\@xa\endcsname\ifx\@nx#1\@nx\else[{#1}]\fi
\else
\csname #2\@xa\endcsname[{Restated}]
\fi}{}{}
\makeatother

\renewcommand{\tilde}{\widetilde}
\renewcommand{\bar}{\widebar}

\newcommand{\Span}{\textup{span}}

\newcommand{\labs}[1]{\left\lvert #1 \right\rvert}
\newcommand{\lpr}[1]{\left( #1 \right)}
\newcommand{\lcr}[1]{\left\{ #1 \right\}}

\newcommand{\alphx}[1][\alpha]{{#1}_{\textit{X}}}
\newcommand{\alphy}[1][\alpha]{{#1}_{\textit{Y}}}
\newcommand{\alphz}[1][\alpha]{{#1}_{\textit{Z}}}
\let\alphax\alphx
\let\alphay\alphy
\let\alphaz\alphz

\newcommand{\lvl}{\ensuremath{\l}}  %

\newcommand{\numxblock}{N_{\textup{BX}}}
\newcommand{\numyblock}{N_{\textup{BY}}}
\newcommand{\numzblock}{N_{\textup{BZ}}}

\newcommand{\numtriple}{N_{\alphax, \alphay, \alphaz}}

\newcommand{\numalpha}[1][\alpha]{N_{#1}}

\newcommand{\pcomp}{p_{\textup{comp}}}

\newcommand{\splres}{\beta}  %
\newcommand{\Thole}{T_{\mathrm{hole}}}

\newcommand{\splresX}{\splres_{\textit{X}}}
\newcommand{\splresY}{\splres_{\textit{Y}}}
\newcommand{\splresZ}{\splres_{\textit{Z}}}

\newcommand{\splresavg}{\bar{\beta}}

\newcommand{\splresXt}[1][t]{\splres_{\textit{X}, #1}}
\newcommand{\splresYt}[1][t]{\splres_{\textit{Y}, #1}}
\newcommand{\splresZt}[1][t]{\splres_{\textit{Z}, #1}}
\newcommand{\splresWt}[1][t]{\splres_{\textit{W}, #1}}
\newcommand{\splonelevelX}{\gamma_{\textit{X}}}
\newcommand{\splonelevelY}{\gamma_{\textit{Y}}}
\newcommand{\splonelevelZ}{\gamma_{\textit{Z}}}
\newcommand{\splonelevelXt}[1][t]{\gamma_{\textit{X}, #1}}
\newcommand{\splonelevelYt}[1][t]{\gamma_{\textit{Y}, #1}}
\newcommand{\splonelevelZt}[1][t]{\gamma_{\textit{Z}, #1}}

\newcommand{\xiXt}[1][t]{\xi_{\textit{X}, #1}}
\newcommand{\xiYt}[1][t]{\xi_{\textit{Y}, #1}}
\newcommand{\xiZt}[1][t]{\xi_{\textit{Z}, #1}}
\newcommand{\xiWt}[1][t]{\xi_{\textit{W}, #1}}

\newcommand{\K}{\kappa}
\newcommand{\oeps}{o_{1/\varepsilon}}

\newcommand{\ang}[1]{\langle #1 \rangle}

\newcommand{\splavg}[1]{\splresavg_{#1, *, *, *}}

\newcommand{\+}{\textup{+}}
\newcommand{\<}{\textup{<}}
\newcommand{\calP}{\mathcal{P}}
\newcommand{\calH}{\mathcal{H}}
\newcommand{\hole}{\textup{hole}}
\newcommand{\TPrime}{\T_{\textup{hash}}}
\newcommand{\TDoublePrime}{\T_{\textup{comp}}}
\newcommand{\TTriplePrime}{\T_{\textup{useful}}}

\newcommand{\mmid}{\;\middle|\;}

\begin{document}

\author{
    Virginia Vassilevska Williams\thanks{Massachusetts Institute of Technology. \texttt{virgi@mit.edu}. Supported by NSF Grants CCF-2129139 and CCF-2330048 and BSF Grant 2020356.}
    \and
    Yinzhan Xu\thanks{Massachusetts Institute of Technology. \texttt{xyzhan@mit.edu}. Partially supported by NSF Grants CCF-2129139 and CCF-2330048 and BSF Grant 2020356.}
    \and
    Zixuan Xu\thanks{Massachusetts Institute of Technology. \texttt{zixuanxu@mit.edu}.}
    \and
    Renfei Zhou\thanks{Institute for Interdisciplinary Information Sciences, Tsinghua University. \texttt{zhourf20@mails.tsinghua.edu.cn}.}
}
\title{New Bounds for Matrix Multiplication: from Alpha to Omega}
\date{}
\pagenumbering{gobble} 
\maketitle

\begin{abstract}
The main contribution of this paper is a new improved variant of the laser method for designing matrix multiplication algorithms. Building upon the recent techniques of [Duan, Wu, Zhou, FOCS 2023], the new method introduces several new ingredients that not only yield an improved bound on the matrix multiplication exponent $\omega$, but also improve the known bounds on rectangular matrix multiplication by [Le Gall and Urrutia, SODA 2018].

In particular, the new bound on $\omega$ is 
\begin{center}
$\omega\le 2.371552$ (improved from $\omega\le 2.371866$).\end{center}
For the dual matrix multiplication exponent $\alpha$ defined as the largest $\alpha$ for which $\omega(1,\alpha,1)=2$, we obtain the improvement
\begin{center}
$\alpha \ge 0.321334$ (improved from $\alpha \ge 0.31389$).
\end{center}
Similar improvements are obtained for various other exponents for multiplying rectangular matrices.
\end{abstract}

\newpage
\pagenumbering{arabic}

\section{Introduction}
Matrix multiplication is arguably the most basic linear algebraic operation, with plentiful applications throughout computer science and beyond. Its algorithmic complexity has been studied for many decades. In 1969 a breakthrough result by Strassen \cite{strassen} showed that $n\times n$ matrices can be multiplied faster than the naive cubic time algorithm. Since then there has been an explosion of results obtaining lower and lower bounds on the exponent $\omega$ defined as the smallest constant such that for all $\eps>0$, $n\times n$ matrices can be multiplied using $O(n^{\omega+\eps})$ arithmetic operations (additions, subtractions, multiplications and divisions; this is called the arithmetic circuit model of computation).
In recent years, the bound $\omega<2.373$ has been obtained~\cite{virgi12,stothers,LeGall32power,AlmanW21}. A new paper by Duan, Wu and Zhou~\cite{duan2023} shows that $\omega<2.3719$.

The dream bound would be $\omega=2$, implying a near-linear time algorithm for multiplying matrices. Unfortunately, a series of works \cite{ambainis,almanitcs,journals/toc/ChristandlVZ21,journals/toc/Alman21,aw2,blasiak2017cap,AlonSU13-sunflower-matrixmult,blasiak2017groups} has shown that the known techniques for multiplying matrices cannot achieve $\omega=2$.

All work on matrix multiplication since 1986 \cite{laser,virgi12,stothers,LeGall32power,AlmanW21,duan2023} has used various variants of the so-called {\em laser method}. The strongest limitation result known for the laser method and its variants \cite{ambainis} is that such techniques cannot show that $\omega<2.3078$. 

The limitation results could mean that radically new methods need to be produced to make big strides. Yet, even if one stays within the laser method framework, it is still an intriguing question: {\em how close can we get to the $2.3078$ barrier bound?}

In many applications of matrix multiplication, one needs to multiply {\em rectangular} matrices: $n^a\times n^b$ by $n^b\times n^c$, where $a,b,c$ are different. Here one defines $\omega(a,b,c)$ to be the exponent for which matrix products of such dimensions can be multiplied in $O(n^{\omega(a,b,c)+\eps})$ time for all $\eps>0$, in the arithmetic circuit model of computation.

For instance, in the study of All-Pairs Shortest Paths (APSP) in unweighted directed graphs \cite{zwickbridge}, the complexity of APSP depends on the value $\mu$ which is defined as the real number satisfying the equation $\omega(1,\mu,1)=1+2\mu$. The same value is needed for the best known algorithms for computing minimum witnesses of Boolean Matrix Multiplication \cite{CzumajKL07}, for All-Pairs Bottleneck Paths in node-weighted graphs~\cite{ShapiraYZ11} and other problems.

In the work on $k$-clique detection, the value of $\omega(1,2,1)$ is important, as it is known \cite{EG04} that $4$-cliques in $n$-node graphs can be detected in $O(n^{\omega(1,2,1)+\eps})$ time for any $\eps > 0$. Moreover, if $\omega(1,2,1)<3.16$, this would improve the known algorithms for $k$-clique detection for all $k\geq 8$ \cite{NS85}. 

A final value of interest is $\alpha$, the largest constant so that $\omega(1,\alpha,1)=2$, first studied by Coppersmith~\cite{Coppersmith82,coppersmith1997rectangular}. If $\omega=2$, then $\alpha=1$. So one can view the goal of increasing $\alpha$ as another way to attempt to prove that $\omega=2$. 

The best bounds on rectangular matrix multiplication to date are given by Le Gall and Urrutia \cite{legallrect2}, which improved upon \cite{Coppersmith82,coppersmith1997rectangular, HUANG1998257, ke2008fast,legallrect}. For the values listed above, the bounds obtained by \cite{legallrect2} are as follows: $\mu<0.5286$, $\omega(1,2,1)<3.25164$ and $\alpha<0.31389$.

The goal of this paper is to obtain better bounds on $\omega,\alpha,\mu$ and rectangular matrix multiplication in general.

\subsection{Our Results}
The main result of this paper is a new improved variant of the laser method for designing matrix multiplication algorithms. Applying the new method, we obtain improved bounds for both square and rectangular matrix multiplication.

In particular, we show that $\alpha >0.321334$ (improving upon the previous bound $0.31389$), $\mu<0.527661$ (improving upon the previous bound $0.5286$) and $\omega(1,2,1)<3.250385$ (improving upon $3.25164$).

As a consequence, Zwick's algorithm for APSP in directed unweighted graphs (and several other algorithms, e.g., minimum witnesses for Boolean Matrix Multiplication \cite{CzumajKL07} and All-Pairs Bottleneck Paths in node-weighted graphs \cite{ShapiraYZ11}) runs in $O(n^{2.527661})$ time and $4$-cliques can be found in $O(n^{3.250385})$ time.

For many other bounds on rectangular matrix multiplication, see Table \ref{table:result}.
\newcommand{\colwidth}{2.5cm}
\newcolumntype{M}[1]{>{\centering\arraybackslash}m{#1}}

\newcommand{\columnone}{%
\begin{tabular}{|c|M{\colwidth}|M{\colwidth}|}
  \hline
  $\kappa$ & upper bound on $\omega(1, \kappa, 1)$ & previous bound on $\omega(1, \kappa, 1)$ \\
  \hline
  \textbf{0.321334} & 2 & N/A \\
  0.33 & 2.000100 & 2.000448 \\
  0.34 & 2.000600 & 2.001118 \\
  0.35 & 2.001363 & 2.001957 \\
  0.40 & 2.009541 & 2.010314 \\
  0.45 & 2.023788 & 2.024801 \\
  0.50 & 2.042994 & 2.044183 \\
  \textbf{0.527661} & 2.055322 & N/A \\
  0.55 & 2.066134 & 2.067488 \\
  0.60 & 2.092631 & 2.093981 \\
  0.65 & 2.121734 & 2.123097 \\
  0.70 & 2.153048 & 2.154399 \\
  \hline
\end{tabular}
}

\newcommand{\columntwo}{%
\begin{tabular}{|c|M{\colwidth}|M{\colwidth}|}
  \hline
  $\kappa$ & upper bound on $\omega(1, \kappa, 1)$ & previous bound on $\omega(1, \kappa, 1)$ \\
  \hline
  0.75 & 2.186210 & 2.187543 \\
  0.80 & 2.220929 & 2.222256 \\
  0.85 & 2.256984 & 2.258317 \\
  0.90 & 2.294209 & 2.295544 \\
  0.95 & 2.332440 & 2.333789 \\
  1.00 & \textbf{2.371552} & 2.371866 \\
  1.10 & 2.452056 & 2.453481 \\
  1.20 & 2.535063 & 2.536550 \\
  1.50 & 2.794941 & 2.796537 \\
  2.00 & 3.250385 & 3.251640 \\
  2.50 & 3.720468 & 3.721503 \\
  3.00 & 4.198809 & 4.199712 \\
  \hline
\end{tabular}
}

\begin{table}[ht]
\centering
\begin{subfigure}
\columnone
\end{subfigure}
\begin{subfigure}
\columntwo
\end{subfigure}
\caption{Our bounds on $\omega(1, \kappa, 1)$ by analyzing the fourth power of the CW tensor compared to the best previous bounds. The previous bound for $\kappa = 1$ comes from \cite{duan2023}'s eighth-power analysis, while all other entries come from \cite{legallrect2}.}\label{table:result}
\end{table}

\paragraph{Independent Work.} Independently, Le Gall \cite{LeGall24} also obtained bounds on rectangular matrix multiplication, improving over \cite{legallrect2}. His method generalizes the approach of \cite{duan2023} to rectangular matrices. For technical reasons, the bound on $\omega$ produced by his method does not match the bound in \cite{duan2023}. In comparison, our method is not only a generalization of \cite{duan2023} to rectangular matrices, but also an improvement. As a result, our bounds are better than the bounds in \cite{LeGall24}.

\section{Technical Overview}

\subsection{Overview of Previous Work}

For positive integers $a,b,c$, the $a\times b \times c$ matrix multiplication tensor $\ang{a,b,c}$ is a tensor over the variable sets $\{x_{ij}\}_{i\in [a], j\in [b]}, \{y_{jk}\}_{j\in [b],k\in[c]}, \{z_{ki}\}_{k\in [c],i\in [a]}$  defined as the tensor computing the $a\times c$ product matrix $\{z_{ki}\}_{k\in [c], i\in [a]}$ of an $a\times b$ matrix $\{x_{ij}\}_{i\in [a], j\in [b]}$ and a $b\times c$ matrix $\{y_{jk}\}_{j\in [b], k\in [c]}$.\footnote{For integer $n \ge 0$, the notation $[n]$ denotes $\{1,\dots, n\}$.} Specifically, $\ang{a,b,c}$ can be written as the following trilinear form
\[\ang{a,b,c} = \sum_{i\in [a]}\sum_{j\in [b]}\sum_{k\in [c]} x_{ij} y_{jk} z_{ki}.\]
It is not hard to check that $\ang{a,b,c}\otimes \ang{d,e,f} = \ang{ad, be, cf}$. For a tensor $T$, let $R(T)$ denote the tensor rank of $T$ and the matrix multiplication exponent $\omega$ is defined as 
\[\omega \defeq \inf_{q\in \N, \, q\ge 2}\log_q R(\ang{q,q,q}).\]
It is hard to directly bound the tensor rank of $\ang{q,q,q}$ in general, so current approaches to bounding $\omega$ utilize Sch{\"o}nhage's asymptotic sum inequality \cite{Schonhage81}, which states that if one can bound the asymptotic rank of a direct sum of matrix multiplication tensors, where the asymptotic rank $\widetilde{R}(T)$ of a tensor $T$ is defined as 
$$\widetilde{R}(T) \defeq \lim_{n \rightarrow \infty} R(T^{\otimes n})^{1/n}, $$
then one can get a bound on $\omega$. More specifically, we recall the asymptotic sum inequality as follows.

\begin{theorem}[Asymptotic Sum Inequality \cite{Schonhage81}]
For positive integers $r > m$ and $a_i, b_i,c_i$ for $i\in [m]$, if 
\[\widetilde{R}\bk{\bigoplus_{i = 1}^m \ang{a_i,b_i,c_i}}\le r,\]
then $\omega\le 3\tau$ where $\tau\in [2/3,1]$ is the solution to the equation
\[\sum_{i = 1}^m (a_i\cdot b_i\cdot c_i)^\tau = r.\]
\end{theorem}

Sch{\"o}nhage's asymptotic sum inequality gave rise to the following approach to bounding $\omega$: start with a tensor $T$ whose asymptotic rank $\widetilde{R}(T)$ is easy to bound. Consider $T^{\otimes n}$ for some $n$ sufficiently large and we want to transform $T^{\otimes n}$ into a direct sum of matrix multiplication tensors whose asymptotic rank is upper bounded by the asymptotic rank of $\widetilde{R}(T^{\otimes n}) = \widetilde{R}(T)^{n}$. The common ways of doing such transformation is via zeroing-out, i.e., setting some variables in $T^{\otimes n}$ to zero, or the more general degeneration, whose definition is deferred to \cref{sec:prelim}.
Then we can apply the asymptotic sum inequality to get a bound on $\omega$. 
Observe that if $T^{\otimes n}$ can be degenerated into $\bigoplus_{i = 1}^m \ang{a_i,b_i,c_i}$, then for a fixed $\tau$, we want to maximize the value of $\sum_{i = 1}^m (a_i\cdot b_i\cdot c_i)^\tau$. This gives a notion of the ``matrix multiplication value'' of a tensor $T$ that we want to maximize. Then notice that a lower bound on the value of $T^{\otimes n}$ would directly imply an upper bound on $\omega$ via the asymptotic sum inequality. It still remains unknown how to get the best possible bound on $\omega$ via a tensor power $T^{\otimes n}$, but the laser method provides one way to give a nontrivial bound.

\paragraph{Laser method.} Let $T$ be a tensor over three sets of variables $X, Y, Z$. For positive integers $s_X, s_Y, s_Z$, let $X = \bigsqcup_{i \in [s_X]} X_i, Y = \bigsqcup_{j \in [s_Y]} Y_j$ and $Z = \bigsqcup_{k \in [s_Z]} Z_k$ be partitions of the $X$-, $Y$-, $Z$-variable sets into $s_X, s_Y, s_Z$ parts respectively. Then $T$ can be written as a sum of subtensors $\sum_{i, j, k} T_{i,j,k}$, where $T_{i,j,k}$ denotes the subtensor of $T$ restricted to variables $X_i, Y_j, Z_k$. 

Suppose for now that each subtensor $T_{i, j, k}$ is a matrix multiplication tensor. If $T$ is a direct sum of matrix multiplication tensors, then we can apply Sch{\"o}nhage's asymptotic sum inequality \cite{Schonhage81} to obtain a bound on $\omega$. However, $T$ is a sum of $T_{i, j, k}$, not necessarily a direct sum.

The laser method \cite{laser} is devised to overcome this issue. First, we take the $n$-th  tensor power of $T$ for some large $n$, which is a tensor over variables $X^n, Y^n, Z^n$: 
$$T^{\otimes n} = \sum_{I \in [s_X]^n} \sum_{J \in [s_Y]^n} \sum_{K \in [s_Z]^n} T_{I,J,K}, $$
where 
\[
  T_{I, J, K} = T_{I_1, J_1, K_1} \otimes T_{I_2, J_2, K_2} \otimes \cdots \otimes T_{I_n, J_n, K_n}.
\]
We will refer to these three sets of variables as $X$-variables, $Y$-variables and $Z$-variables respectively. 
Because each $T_{i, j, k}$ is a matrix multiplication tensor and the tensor products of several $T_{i,j,k}$'s will still be matrix multiplication tensors, $T_{I, J, K}$ is a matrix multiplication tensor for any $I\in [s_X]^n, J\in [s_Y]^n, K\in [s_Z]^n$.
For any $I \in [s_X]^n$, let $X_I$ denote $X_{I_1} \times X_{I_2} \times \cdots \times X_{I_n}$, which is a subset of $X^n$. Similarly we define $Y_J$ and $Z_K$. It is not difficult to see that $T_{I, J, K}$ is exactly the subtensor of $T^{\otimes n}$ when restricted to $X_I, Y_J, Z_K$. We call such subsets $X_I, Y_J, Z_K$ \emph{variable blocks}.

The goal of the laser method is to select some of the variable blocks $X_I, Y_J$ or $Z_K$ and zero out all of the variables in these blocks, i.e. ``zero out the blocks'', so that the remaining tensor is a direct sum of $T_{I,J,K}$'s.

The laser method specifies a distribution $\alpha$ over triples $(i, j, k)$ where $i\in [s_X], j\in [s_Y], k\in [s_Z]$, so that for each $T_{I,J,K}$ that we want to keep in the direct sum, we require that
\begin{equation}\label{eq:consistent-alpha}
    \big| \{t \in [n]\mid (I_t, J_t, K_t) = (i, j, k)\}\big| = \alpha(i, j, k) \cdot n.
\end{equation}
If a subtensor $T_{I,J,K}$ satisfies \eqref{eq:consistent-alpha}, we say that it is {\em consistent} with the distribution $\alpha$. 

The distribution $\alpha$ induces the marginal distributions $\alpha_X, \alpha_Y, \alpha_Z$ on the $X$-, $Y$-, $Z$-variables over the indices $[s_X], [s_Y], [s_Z]$ respectively as follows. Let $\alpha_X$, $\alpha_Y$, $\alpha_Z$ be the marginal distributions of $\alpha$ on the three dimensions respectively, i.e., 
\begin{align*}
    \alpha_X(i) &= \sum_{j \in [s_Y], k \in [s_Z]} \alpha(i, j, k) \quad \forall i\in [s_X],\\
    \alpha_Y(j) &= \sum_{i \in [s_X], k \in [s_Z]} \alpha(i, j, k) \quad \forall j\in [s_Y],\\
    \alpha_Z(k) &= \sum_{i \in [s_X], j \in [s_Y]} \alpha(i, j, k) \quad \forall k\in [s_Z].
\end{align*}
In the laser method, we zero out all $X$-variable blocks $X_I$ that are not consistent with $\alpha_X$ ($X_I$ is consistent with $\alpha_X$ if $|\midBK{t \in [n]: I_t = i}| = \alpha_X(i) \cdot n$ for every $i \in [s_X]$).
We similarly zero out all $Y$-variable blocks $Y_J$ that are not consistent with $\alpha_Y$ and $Z$-variable blocks $Z_K$ that are not consistent with $\alpha_Z$.

At this stage, a subtensor $T_{I, J, K}$ remains if $X_I, Y_J$ and $Z_K$ all remain. Thus, all remaining $T_{I,J,K}$'s are consistent with some distribution $\alpha'$ that induces the same marginal distributions $\alpha_X, \alpha_Y, \alpha_Z$, though $\alpha'$ might be different from $\alpha$. The final stages of the laser method aim to keep a collection of independent subtensors $T_{I, J, K}$ and zero out the subtensors $T_{I,J,K}$ that are consistent with a distribution $\alpha' \ne \alpha$, using techniques such as hashing and greedy procedures. Eventually, the laser method obtains multiple independent copies of the tensor isomorphic to:
\[ \T \defeq \bigotimes_{i, j, k} T_{i,j,k}^{\otimes \alpha(i, j, k) \cdot n}. \]

\paragraph{The Coppersmith-Winograd tensor $\CW_q$.}  Prior works \cite{cw90,stothers,virgi12,LeGall32power,AlmanW21,duan2023} that obtained the best bounds on $\omega$  used the Coppersmith-Winograd tensor $\CW_q$ and its powers as the starting tensor $T$ in the laser method. The Coppersmith-Winograd tensor $\CW_q$ for a nonnegative integer $q$ is defined as
\[\CW_q \defeq x_0y_0z_{q+1} + x_0y_{q+1}z_0 + x_{q+1}y_0z_0 + \sum_{i = 1}^q \lpr{x_0y_iz_i + x_iy_0z_i + x_iy_iz_0}.\]
Observe that 
\[\sum_{i = 1}^q \lpr{x_0y_iz_i + x_iy_0z_i + x_iy_iz_0} \equiv \ang{1,1,q} + \ang{q,1,1} + \ang{1,q,1},\]
so $\CW_q$ is the sum of six matrix multiplication tensors where the other three are copies of $\ang{1,1,1}$. 
Next, we describe the leveled partitions of $\CW_q$ and $\CW_q^{\otimes 2^\lvl}$ that are crucial to our analysis. For simplicity, we denote $T^{(\lvl)} \defeq \CW_q^{\otimes 2^{\lvl-1}}$.

For $T^{(1)} = \CW_q$, its variable sets are partitioned into three parts
\begin{align*}
    X^{(1)} &= X^{(1)}_0 \sqcup X^{(1)}_1 \sqcup X^{(1)}_2 = \{x_0\}\sqcup \{x_1,\dots, x_q\}\sqcup \{x_{q+1}\},\\
    Y^{(1)} &= Y^{(1)}_0 \sqcup Y^{(1)}_1 \sqcup Y^{(1)}_2 = \{y_0\}\sqcup \{y_1,\dots, y_q\}\sqcup \{y_{q+1}\},\\
    Z^{(1)} &= Z^{(1)}_0 \sqcup Z^{(1)}_1 \sqcup Z^{(1)}_2 =  \{z_0\}\sqcup \{z_1,\dots, z_q\}\sqcup \{z_{q+1}\}.
\end{align*}
Notice that under this partition, a constituent tensor $T^{(1)}_{i, j, k}$ is nonzero if and only if $i+j+k = 2$.

For $T^{(\lvl)}=\CW_q^{\otimes 2^{\lvl-1}}$ with variable sets $X^{(\lvl)}, Y^{(\lvl)}, Z^{(\lvl)}$, the above partition on $T^{(1)}$ directly induces a partition on the variable sets $X^{(\lvl)}, Y^{(\lvl)}, Z^{(\lvl)}$ where each part of the partition is indexed by a $\{0, 1, 2\}$-sequence of length $2^{\lvl-1}$. Specifically, this gives the partition
\[X^{(\lvl)} = \bigsqcup_{(\hat{i}_1, \hat{i}_2, \ldots, \hat{i}_{2^{\lvl-1}}) \in \{0, 1, 2\}^{2^{\lvl-1}}}X^{(1)}_{\hat{i}_1} \otimes  X^{(1)}_{\hat{i}_2} \otimes \cdots \otimes X^{(1)}_{\hat{i}_{2^{\lvl-1}}}\]
for $X$-variables and analogous partitions for $Y$- and $Z$-variables. 

One can use the laser method on these partitions. However, this would not yield an improved bound on $\omega$ from what one would get just by analyzing $T^{(1)}$.
The reason behind the improvement obtained by analyzing higher powers of $\CW_q$ comes from the fact that we can consider the following coarsening of the above partition where the parts corresponding to sequences with the same sum are ``merged'' into a single part. More specifically, we have
\[
  X^{(\lvl)} = \bigsqcup_{i = 0}^{2^\lvl} X_i^{(\lvl)},
  \qquad \textup{where} \quad
  X_i^{(\lvl)} \defeq
  \bigsqcup_{\substack{(\hat{i}_1, \hat{i}_2, \ldots, \hat{i}_{2^{\lvl-1}}) \in \{0, 1, 2\}^{2^{\lvl-1}}  \\\sum_t \hat{i}_t = i}}X^{(1)}_{\hat{i}_1} \otimes  X^{(1)}_{\hat{i}_2} \otimes \cdots \otimes X^{(1)}_{\hat{i}_{2^{\lvl-1}}}.
\]
 We refer to this specific partition of $T^{(\lvl)}$ as the \emph{level-$\lvl$ partition}.
 Note that we can also view this partition as obtained from coarsening the level-($\lvl-1$) partition, i.e.,
 \[X_i^{(\lvl)} = \bigsqcup_{\substack{0 \le i' \le i \\ 0 \le i', i - i' \le 2^{\ell}}}X^{(\lvl-1)}_{i'} \otimes X^{(\lvl-1)}_{i-i'}.\]
 We can partition the variable sets $Y^{(\lvl)}$ and $Z^{(\lvl)}$ similarly. Then we use $T^{(\lvl)}_{i, j, k}$ to denote the subtensor of $T^{(\lvl)}$ restricted to the variable subsets $X_i^{(\lvl)}, Y_j^{(\lvl)}, Z_k^{(\lvl)}$ and note that $T^{(\lvl)}_{i, j, k}$ is nonzero if and only if $i+j+k = 2^{\lvl}$. We call $T^{(\lvl)}_{i, j, k}$ a level-$\lvl$ \emph{constituent tensor}, $X^{(\lvl)}_i, Y^{(\lvl)}_j, Z^{(\lvl)}_k$ level-$\lvl$ variable blocks, and we omit the superscript $(\lvl)$ when $\lvl$ is clear from context.

For $\lvl > 1$, some level-$\lvl$ constituent tensors $T^{(\lvl)}_{i,j,k}$ are no longer matrix multiplication tensors, so each independent copy of $\T=\bigotimes_{i, j, k} \bigbk{T_{i,j,k}^{(\lvl)}}^{\otimes \alpha(i,j,k) \cdot n}$ may also no longer be a matrix multiplication tensor. To resolve this issue, prior works \cite{cw90,stothers,virgi12,LeGall32power,AlmanW21} use the laser method recursively to analyze  $T_{i,j,k}$'s that are not matrix multiplication tensors.

\paragraph{The work of \cite{duan2023}.}  Consider the analysis on the tensor $T^{(\lvl)}$ of the laser method. In previous approaches prior to the work of Duan, Wu and Zhou~\cite{duan2023}, one would first apply the laser method on $T^{(\lvl)}$ to obtain multiple copies of $\T=\bigotimes_{i, j, k} \bigbk{T_{i,j,k}^{(\lvl)}}^{\otimes \alpha(i,j,k) \cdot n}$ which consists of level-$\lvl$ constituent tensors $T_{i,j,k}^{(\ell)}$ and do not share level-$\lvl$ variable blocks. Then for each $T_{i,j,k}^{(\ell)}$ that is not a matrix multiplication tensor, one would recursively apply the laser method to obtain multiple copies of some other tensors that are independent over level-$(\lvl-1)$ variable blocks.

Recall that for a level-$\lvl$ constituent tensor $T_{i, j, k}^{(\lvl)}$, we can partition its variable set $X_i^{(\lvl)}, Y_j^{(\lvl)}, Z_k^{(\lvl)}$ recursively into $\bigsqcup_{i'} X_{i'}^{(\lvl-1)} \otimes X_{i-i'}^{(\lvl-1)}$, $\bigsqcup_{j'} Y_{j'}^{(\lvl-1)} \otimes Y_{j-j'}^{(\lvl-1)}$ and $\bigsqcup_{k'} Z_{k'}^{(\lvl-1)} \otimes Z_{k-k'}^{(\lvl-1)}$ respectively.
In the first recursive step, when applying the laser method on $T_{i, j, k}^{(\lvl)}$, we take the $n'$-th tensor power $\bigbk{T_{i,j,k}^{(\lvl)}}^{\otimes n'}$ of $T_{i,j,k}^{(\lvl)}$ for some large $n'$ and specify a distribution $\beta$ over triples $((i', i-i'), \, (j', j-j'), \, (k', k-k'))$ where $0\le i'\le i$, $0\le j'\le j$, $0\le k'\le k$, and zero out all variables blocks that are not consistent with the marginal distributions induced by $\beta$. Therefore, in $T_{i, j, k}^{\otimes n'}$, only a subset of the level-$(\ell-1)$ variable blocks survive the above zeroing-out.

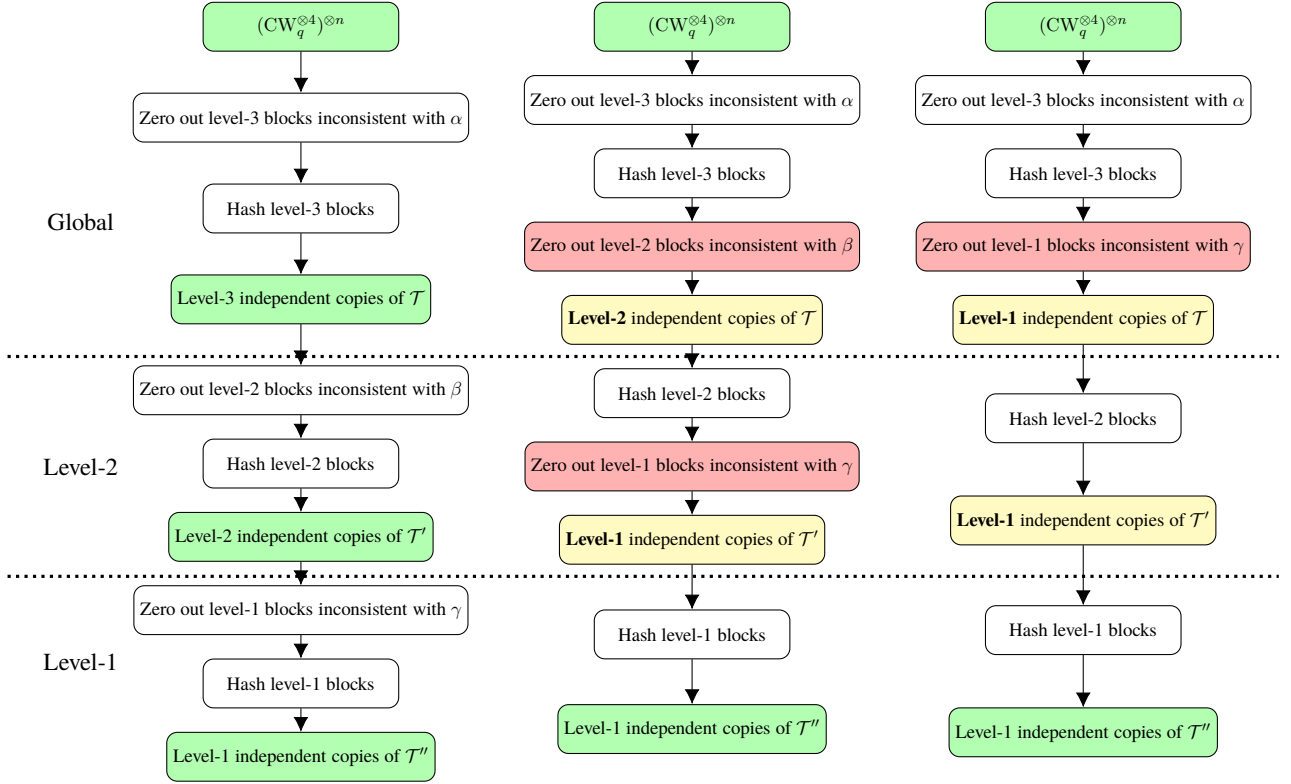
\begin{figure}[ht]
  \centering
  \tikzset{%
  >={Latex[width=2mm,length=2mm]},
    base/.style = {rectangle, rounded corners, draw=black, minimum width=4cm, minimum height=1cm, text centered},
    regularblock/.style = {base},
    highlightblock/.style = {base, fill=red!30},
    inoutblock/.style = {base, fill=green!30},
    changeblock/.style = {base, fill=yellow!30}
}

\begin{tikzpicture}[node distance=1.5cm,
    every node/.style={fill=white}, align=center,scale=0.65, every node/.style={scale=0.65}]
  \node (CW4) [inoutblock] {$(\CW_q^{\otimes 4})^{\otimes n}$};
  \node (zerooutalpha)  [regularblock, below=0.55cm of CW4] {Zero out level-3 blocks inconsistent with $\alpha$};
  \node (hashlevel3) [regularblock, below=0.55cm of zerooutalpha] {Hash level-3 blocks};
  \node (level3ind) [inoutblock, below=0.55cm of hashlevel3] {Level-3 independent copies of $\T$};

  \node (zerooutbeta)  [regularblock, below=0.55cm of level3ind]          {Zero out level-2 blocks inconsistent with $\beta$};
  \node (hashlevel2) [regularblock, below of=zerooutbeta] {Hash level-2 blocks};
  \node (level2ind) [inoutblock, below of=hashlevel2] {Level-2 independent copies of $\T'$};

  \node (zerooutgamma)  [regularblock, below of=level2ind]          {Zero out level-1 blocks inconsistent with $\gamma$};
  \node (hashlevel1) [regularblock, below of=zerooutgamma] {Hash level-1 blocks};
  \node (level1ind) [inoutblock, below of=hashlevel1] {Level-1 independent copies of $\T''$};

  \draw[->] (CW4) -- (zerooutalpha);
  \draw[->] (zerooutalpha) -- (hashlevel3);
  \draw[->] (hashlevel3) -- (level3ind);
  \draw[->] (level3ind) -- (zerooutbeta);
  \draw[->] (zerooutbeta) -- (hashlevel2);
  \draw[->] (hashlevel2) -- (level2ind);
  \draw[->] (level2ind) -- (zerooutgamma);
  \draw[->] (zerooutgamma) -- (hashlevel1);
  \draw[->] (hashlevel1) -- (level1ind);

    \coordinate (shift1) at (8,0);
  \begin{scope}[shift=(shift1)]
     \node (CW4) [inoutblock] {$(\CW_q^{\otimes 4})^{\otimes n}$};
  \node (zerooutalpha)  [regularblock, below of=CW4]          {Zero out level-3 blocks inconsistent with $\alpha$};
  
  \node (hashlevel3) [regularblock, below of=zerooutalpha] {Hash level-3 blocks};
    \node (zerooutbeta)  [highlightblock, below of=hashlevel3]          {Zero out level-2 blocks inconsistent with $\beta$};
  \node (level3ind) [changeblock, below of=zerooutbeta] { \textbf{Level-2} independent copies of $\T$};

  \node (hashlevel2) [regularblock, below of=level3ind] {Hash level-2 blocks};

   \node (zerooutgamma)  [highlightblock, below of=hashlevel2]          {Zero out level-1 blocks inconsistent with $\gamma$};
  
  \node (level2ind) [changeblock, below of=zerooutgamma] {\textbf{Level-1} independent copies of $\T'$};

  \node (hashlevel1) [regularblock, below=0.6cm of level2ind] {Hash level-1 blocks};
  \node (level1ind) [inoutblock, below=0.6cm of hashlevel1] {Level-1 independent copies of $\T''$};

  \draw[->] (CW4) -- (zerooutalpha);
  \draw[->] (zerooutalpha) -- (hashlevel3);
  \draw[->] (hashlevel3) -- (zerooutbeta);
  \draw[->] (zerooutbeta) -- (level3ind);
  \draw[->] (level3ind) -- (hashlevel2);
  \draw[->] (hashlevel2) -- (zerooutgamma);
  \draw[->] (zerooutgamma) -- (level2ind);
  \draw[->] (level2ind) -- (hashlevel1);
  \draw[->] (hashlevel1) -- (level1ind);
  \end{scope}

    \coordinate (shift2) at (16,0);
    \begin{scope}[shift=(shift2)]
        \node (CW4) [inoutblock] {$(\CW_q^{\otimes 4})^{\otimes n}$};
  \node (zerooutalpha)  [regularblock, below of=CW4]          {Zero out level-3 blocks inconsistent with $\alpha$};
  
  \node (hashlevel3) [regularblock, below of=zerooutalpha] {Hash level-3 blocks};
  
    \node (zerooutgamma)  [highlightblock, below of=hashlevel3]          {Zero out level-1 blocks inconsistent with $\gamma$};
    
  \node (level3ind) [changeblock, below of=zerooutgamma] {\textbf{Level-1} independent copies of $\T$};

  \node (hashlevel2) [regularblock, below=0.65cm of level3ind] {Hash level-2 blocks};
  \node (level2ind) [changeblock, below=0.7cm of hashlevel2] {\textbf{Level-1} independent copies of $\T'$};

  \node (hashlevel1) [regularblock, below=0.8cm of level2ind] {Hash level-1 blocks};
  \node (level1ind) [inoutblock, below=0.7cm of hashlevel1] {Level-1 independent copies of $\T''$};

  \draw[->] (CW4) -- (zerooutalpha);
  \draw[->] (zerooutalpha) -- (hashlevel3);
  \draw[->] (hashlevel3) -- (zerooutgamma);
  \draw[->] (zerooutgamma) -- (level3ind);
  \draw[->] (level3ind) -- (hashlevel2);
  \draw[->] (hashlevel2) -- (level2ind);
  \draw[->] (level2ind) -- (hashlevel1);
  \draw[->] (hashlevel1) -- (level1ind);
  
    \end{scope}

  \draw[dotted,very thick] (-6,-6.75) -- (20,-6.75);
  \node[scale=1.3] at (-4.5,-4) {Global};
  \node[scale=1.3] at (-4.5,-9) {Level-2};
  \node[scale=1.3] at (-4.5,-13) {Level-1};

  \draw[dotted,very thick] (-6,-11.25) -- (20,-11.25);
  
  \node[scale=1.5] at (0,-17) {(a) \cite{virgi12,stothers,LeGall32power,AlmanW21}};
  \node[scale=1.5] at (8,-17) {(b) \cite{duan2023}};
  \node[scale=1.5] at (16,-17) {(c) This work};

  \end{tikzpicture}
  \caption{High-level comparison between this work and prior works on $(\CW_q^{\otimes 4})^{\otimes n}$. Here, $\alpha$ is a distribution over level-$3$ constituent tensors, $\beta$ is a collection of distributions over level-$2$ constituent tensors, and $\gamma$ is a collection of distributions over level-$1$ constituent tensors.}
  \label{fig:compare}
\end{figure}

Now suppose we can move the above zeroing-out step earlier, say before we have independent copies of $\T$ when we first apply the laser method on $T^{(\lvl)}$, then instead of keeping independent copies of $\T$, we only need to keep a subtensor $\T'$ of it, where $\T'$ is $\T$ after applying the above zeroing-out step. 
This leads to one of the key observations in~\cite{duan2023}: we do not need to have copies of $\T'$ that are fully independent over the level-$\lvl$ variable blocks. Instead, any two copies can share the same level-$\lvl$ variable block as long as they do not share the same level-$(\lvl-1)$ variable blocks that would survive the first zeroing-out in the recursive application of the laser method on the level-$\lvl$ constituent tensors. As a result, we can potentially keep more independent copies of  $\T'$, because of the relaxed constraints, and each copy $\T'$ would still be essentially as good as $\T$ for the purpose of the analysis because we are merely moving a later zeroing-out earlier. Because we are keeping more copies, by the asymptotic sum inequality, we will achieve a better bound for $\omega$. 
 
As illustrated in \cref{fig:compare}, consider $(\CW_q^{\otimes 4})^{\otimes n}$ and suppose $\alpha, \beta, \gamma$ are (collections of) distributions over level-$3$, level-$2$, level-$1$ constituent tensors respectively. In subfigure (a), works prior to \cite{duan2023} including~\cite{virgi12, AlmanW21, LeGall32power} zero out level-$3$ blocks according to $\alpha$ and obtain level-$3$-independent\footnote{We say several subtensors of $\CW_q^{\otimes N}$ are \emph{level-$\lvl$-independent} if they do not share any level-$\lvl$ variable block, and thus they are also \emph{independent}.} copies of $\T'$ before zeroing out level-$2$ blocks. As shown in subfigure (b), Duan et al.~\cite{duan2023} moved the step of zeroing out level-$2$ blocks according to $\beta$ earlier and only obtained level-$2$-independence as opposed to level-$3$-independence.

It is not obvious how one can accomplish the above modification. Duan et al.~\cite{duan2023} considered the notion of \emph{split distributions}, which roughly measures how a level-$\lvl$ block ``splits'' into level-($\lvl-1$) blocks with respect to the recursive leveled partition. By observing the split distribution of a level-$\lvl$ block, one gains some partial information about the level-$(\lvl-1)$ blocks that allows the modification of zeroing out level-$(\lvl - 1)$ blocks inconsistent with $\beta$ earlier. Ideally, one would hope to achieve this modification symmetrically over the $X$-, $Y$-, and $Z$-variables, i.e., allow the sharing of level-$\lvl$ variable blocks in all three dimensions, but the method in \cite{duan2023} did not achieve that. Instead, their technique works when the multiple copies of $\T$ only share the same level-$\lvl$ $Z$-variable block while each $X$- and $Y$-variable block needs to be contained in a unique level-$\lvl$ subtensor. (More generally, their technique works when level-$\lvl$ variable blocks are shared in exactly one of $X$-, $Y$-, $Z$-variables). In order to set up the tensor satisfying the required constraints, they need to zero out the $Z$-variable blocks asymmetrically with respect to the $X$- and $Y$-variables. It still remains an open question whether the techniques in \cite{duan2023} can be symmetrized over the three dimensions.

Another technical detail is that the obtained independent copies of tensors in \cite{duan2023} are not all necessarily full copies of $\T'$. 
That is, some variables of the independent tensors are zeroed out. This creates independent copies of $\T'$ but with some ``holes''.
 Because of the asymmetry of their method, such holes can only appear in $Z$-variables. In order to overcome this issue, they showed that, as long as the fraction of holes is small, and all holes are in $Z$-variables, one can degenerate a small number of independent copies of $\T'$ with holes to a full copy of $\T$. Prior to their work, Sch{\"{o}}nhage \cite{Schonhage81} also studied this problem of degenerating multiple independent copies of a tensor with holes to a full copy of the tensor. Sch{\"{o}}nhage's method applied to the case when two of the three dimensions can have holes, but it focuses only on matrix multiplication tensors. 

\subsection{Our Improvements}

\paragraph{Complete split distribution. } We take the observation of \cite{duan2023} one step further. The high-level idea is the following: instead of keeping copies of $\T$ that are independent over level-$(\lvl-1)$ variable blocks, we keep copies of it that are independent over level-$1$ variable blocks. For $\lvl > 1$, this should give more degrees of freedom and enable us to keep more copies of $\T$. As illustrated in \cref{fig:compare}~(c), we directly move the step of zeroing out level-$1$ blocks according to $\gamma$ earlier and obtain level-$1$ independence as opposed in level-$2$ independence in \cite{duan2023}.

To implement the above idea, we utilize the notion of \emph{complete split distributions}, which can be viewed as an extension of the notion of split distributions used in \cite{duan2023}. Recall that in \cite{duan2023}, a level-$\lvl$ split distribution measures how a level-$\lvl$ variable block splits into level-$(\lvl-1)$ blocks. A level-$\lvl$ complete split distribution measures how a level-$\lvl$ block splits into level-$1$ variable blocks. Specifically, a level-$1$ block sequence of length $2^{\lvl-1} \cdot n$ in $T^{(\lvl)}$ can be viewed as $n$ consecutive chunks of $\{0, 1, 2\}$-sequences each of length $2^{\lvl-1}$, and we consider the proportion of each of these $3^{2^{\lvl-1}}$ possible types of chunks in the $n$ chunks. A level-$\lvl$ complete split distribution is a distribution on these $3^{2^{\lvl-1}}$ types of chunks, and a level-$1$ block sequence (and its corresponding level-$1$ variable block) is said to be consistent with a level-$\lvl$ complete split distribution if the proportion of each type of chunks matches the corresponding probability specified in the complete split distribution.

Let $\splres_X, \splres_Y, \splres_Z$ be three level-$\lvl$ complete split distributions, and let $T_{i,j,k}$ be a level-$\lvl$ constituent tensor. We will consider the tensor $T_{i,j,k}^{\otimes n}[\splres_X, \splres_Y, \splres_Z]$, which is obtained from $T_{i,j,k}^{\otimes n}$ by zeroing out all level-$1$ $X$-, $Y$-, $Z$-variable blocks that are not consistent with $\splres_X, \splres_Y, \splres_Z$ respectively. We call this ``enforcing the complete split distributions''. In our recursive steps, we will analyze $T_{i,j,k}^{\otimes n}[\splres_X, \splres_Y, \splres_Z]$ instead of $T_{i,j,k}^{\otimes n}$.

\paragraph{Enforcing split distributions in all three dimensions.} Dual et al.~\cite{duan2023} only enforce their split distribution in one of the dimensions (the $Z$ variables). In our method, we need to enforce complete split distributions in all three dimensions. Here we explain why.

First of all, when analyzing a level-$\lvl$ constituent tensor $T_{i, j, k}^{\otimes n}$, \cite{duan2023} only consider split distributions, instead of complete split distributions. Every level-$(\lvl-1)$ block sequence in $T_{i, j, k}^{\otimes n}$ can be viewed as a length-$(2n)$ sequence on $\{0, 1, \ldots, 2^{\lvl-1}\}$. If we split the sequence to chunks of length $2$, we obtain a length-$n$ sequence of pairs in $\{0, 1, \ldots, 2^{\lvl-1}\}^2$. The split distribution used in \cite{duan2023} essentially specifies the proportion of each type of pairs, and they zero out all level-$(\lvl-1)$ variable blocks that are not consistent with the specified proportions.

Similar to what we discussed earlier, when enforcing the split distribution on the tensor $T_{i',j',k'}^{\otimes n}$ (or $T_{i-i',j-j',k-k'}^{\otimes n}$), the constraint becomes a constraint that enforces the proportion of each level-$(\lvl-1)$ variable block in the level-$(\lvl-1)$ variable blocks in $T_{i',j',k'}^{\otimes n}$. Since there is only one level-$(\lvl-1)$ block in $T_{i',j',k'}^{\otimes n}$, either the whole tensor $T_{i',j',k'}^{\otimes n}$ satisfies the constraints, or it does not. Thus, the constraints of the split distribution do not carry over to further recursion levels.

When analyzing each constituent tensor $T_{i, j, k}^{\otimes n}$, Duan et al.~\cite{duan2023} aim to obtain some ``symmetrized value'' of $T_{i, j, k}$, similar to previous works \cite{cw90,stothers,virgi12,LeGall32power,AlmanW21}. As a result, when analyzing $T_{i, j, k}^{\otimes n}$, they apply their method multiple times to enforce a split distribution on each of the three possible dimensions, i.e., they can choose to share $X$-, $Y$-, or $Z$-variables depending on which application of their method it is. Still, the constraints of the split distribution do not carry over to the next recursion level as discussed in the previous paragraph. Thus, in their analysis, holes only appear in one of the dimensions. 

However, when enforcing a complete split distribution, the constraints carry over to further recursion levels: say in the analysis for $T_{i, j, k}$ in some application of the method in the current level, we choose to enforce a complete split distribution on $Z$-variables. This constraint still has an effect on the next level. However, in the analysis at the next level, we can choose to enforce a complete split distribution on $Y$-variables instead. This creates constraints on the complete split distribution in two dimensions. In general, these constraints can appear in all three dimensions, and therefore, we need to handle holes in all three dimensions. 

\paragraph{A technical issue.} A technical issue arises if we enforce complete split distributions in three dimensions. We consider a simplified scenario where the support of the distribution $\beta$ has size $1$ to explain the issue. In other words, we aim to zero out $T_{i,j,k}^{\otimes n}$ into independent copies of $(T_{i',j',k'} \otimes T_{i-i', j-j', k-k'})^{\otimes n}$ for some $i', j', k'$. In this simplified scenario, if we do not enforce complete split distributions, we could rewrite $(T_{i',j',k'} \otimes T_{i-i', j-j', k-k'})^{\otimes n}$ equivalently as $T_{i',j',k'}^{\otimes n} \otimes T_{i-i', j-j', k-k'}^{\otimes n}$ by simply permuting the indices, and then recursively analyze $T_{i',j',k'}^{\otimes n}$ and $T_{i-i', j-j', k-k'}^{\otimes n}$ separately. Now with complete split distribution, this step becomes problematic.
Suppose we are able to obtain independent copies of
\[ 
  \T_1 \defeq (T_{i',j',k'} \otimes T_{i-i', j-j', k-k'})^{\otimes n}\left[\splres_X, \splres_Y, \splres_Z\right],
\]
for some $\splres_X, \splres_Y, \splres_Z$.
Then in order to recursively analyze $\T_1$, we instead need a tensor 
\[
  \T_2 \defeq \left( T_{i',j',k'}^{\otimes n}\bigBk{\splres_X^{(L)}, \splres_Y^{(L)}, \splres_Z^{(L)}}\right) \otimes \left(T_{i-i', j-j', k-k'}^{\otimes n}\bigBk{\splres_X^{(R)}, \splres_Y^{(R)}, \splres_Z^{(R)}} \right),
\]
for some level-$(\lvl-1)$ complete split distributions $\splres_X^{(L)}, \splres_Y^{(L)}, \splres_Z^{(L)}, \splres_X^{(R)}, \splres_Y^{(R)}, \splres_Z^{(R)}$. 

Let us discuss how the above level-$(\ell-1)$ complete split distributions are related to $\beta_X,\beta_Y,\beta_Z$.
To give some intuition, in each length-$2^{\lvl}$ chunk of a level-$1$ block sequence in $\T_1$, the first half-chunk belongs to some $T_{i', j', k'}$, and the second half-chunk belongs to some $T_{i - i', j - j', k - k'}$. In $\T_2$, we permute the indices so that all the first half-chunks belonging to some $T_{i', j', k'}$ are put together in the first half of the resulting sequence, and all the second half-chunks belonging to some $T_{i-i', j-j', k-k'}$ are put together in the second half of the resulting sequence. If we enforce a level-$\lvl$ complete split distribution $\splres_X$ on a level-$1$ block sequence $\hat{I}\in \{0,1,2\}^{2^{\lvl-1}}$ in $\T_1$, what would the permuted sequence look like? Let $\sigma_1,\sigma_2\in \{0,1,2\}^{2^{\lvl-2}}$ denote two length-$2^{\lvl - 2}$ chunks and let $\sigma_1 \circ \sigma_2$  denote their concatenation. Since $\hat{I}$ is consistent with $\splres_X$, $\hat{I}$ contains $\splres_X(\sigma_1 \circ \sigma_2) \cdot n$ chunks $\sigma_1 \circ \sigma_2$ for every $\sigma_1, \sigma_2$. For each of these chunks, $\sigma_1$ gets permuted to the first half of the permuted level-$1$ block sequence in $\T_2$, and $\sigma_2$ gets permuted to the second half of the permuted level-$1$ block sequence in $\T_2$. Summing over all $\sigma_1, \sigma_2$, it is not difficult to verify that
\[\splres_X^{(L)}(\sigma_1) = \sum_{\sigma_2} \splres_X(\sigma_1 \circ \sigma_2), \qquad \splres_X^{(R)}(\sigma_2) = \sum_{\sigma_1} \splres_X(\sigma_1 \circ \sigma_2).\]
In this sense, $\splres_X^{(L)}$ and $\splres_X^{(R)}$ can be viewed as two marginal distributions of $\splres_X$. This similarly holds for $Y$ and $Z$. 

One set of constraints we can add to make $\splres_X^{(L)}$ and $\splres_X^{(R)}$ always the two marginal distributions of $\splres_X$ is $\splres_X = \splres_X^{(L)} \times \splres_X^{(R)}$, namely we enforce $\splres_X$ to be the joint distribution of (independently distributed) $\splres_X^{(L)}$ and $\splres_X^{(R)}$. 
Similarly we can add the constraints $\splres_Y = \splres_Y^{(L)} \times \splres_Y^{(R)}$ and $\splres_Z = \splres_Z^{(L)} \times \splres_Z^{(R)}$. 

However, even with these constraints, $\T_1$ might not necessarily be equivalent to $\T_2$. By the above reasoning, every level-$1$ block sequence in $\T_1$ is permuted into a level-$1$ block sequence in $\T_2$, but not all block sequences in $\T_2$ can be obtained this way. Intuitively, this is because joint distributions can determine marginal distributions, which means that, for instance, $\splres_X^{(L)} \times \splres_X^{(R)}$ can determine both  $\splres_X^{(L)}$ and $\splres_X^{(R)}$. The other way is not true, and there could be multiple joint distributions whose marginals satisfy $\splres_X^{(L)}$ and $\splres_X^{(R)}$. 

By a careful calculation, one can still show that the proportion of $X$-, $Y$-, $Z$-variables in $\T_2$ that are not in $\T_1$ is at most a $1-2^{-o(N)}$ fraction of those in $\T_2$. These variables become holes. Unfortunately, the methods in previous works \cite{Schonhage81, duan2023} do not apply, as they are unable to fix holes that are present in all three dimensions ($X$-, $Y$-, $Z$-variables). 

Next, we discuss how we fix the technical issue.

\paragraph{Intuition of the fix.} The first step towards resolving this issue is to decrease the fraction of holes in all three dimensions, from $1-2^{-o(N)}$ all the way down to $2^{-\Omega(N)}$. Then we describe a generic method adapted from \cite{duanpersonal} for fixing holes in all three dimensions as long as the fractions of holes are small. 

For the first step, we slightly relax the condition for zeroing out variables in $\T_1$ and $\T_2$. Let $\eps > 0$ be an arbitrary constant. For any $T_{i,j,k}$, we use $T_{i,j,k}^{\otimes n}[\splres_X, \splres_Y, \splres_Z, \eps]$ to denote  $T_{i,j,k}^{\otimes n}$ but we zero out all level-$1$ $X$-, $Y$-, $Z$-variables, where the proportion of each chunk in $\{0, 1, 2\}^{2^{\lvl-1}}$ in their level-$1$ block sequence differs at most $\eps$ from the corresponding probability in  $\splres_X, \splres_Y, \splres_Z$ respectively. That is, we allow some small flexibility when zeroing out variables. Then, let
\[
  \T_1' \defeq \bigbk{T_{i',j',k'} \otimes T_{i-i', j-j', k-k'}}^{\otimes n}\left[\splres_X^{(L)} \times \splres_X^{(R)}, \; \splres_Y^{(L)} \times \splres_Y^{(R)}, \; \splres_Z^{(L)} \times \splres_Z^{(R)}, \; \eps\right],
\]
and recall
\[
  \T_2 = \left( T_{i',j',k'}^{\otimes n}\bigBk{\splres_X^{(L)}, \splres_Y^{(L)}, \splres_Z^{(L)}}\right) \otimes \left(T_{i-i', j-j', k-k'}^{\otimes n}\bigBk{\splres_X^{(R)}, \splres_Y^{(R)}, \splres_Z^{(R)}} \right).
\]

Intuitively, we allow more flexibility in $\T_1$ than that in $\T_2$, so that more variables remain in $\T_1$ compared to $\T_2$, and the fraction of holes should become smaller. The idea for proving this is to use concentration bounds: if we pick a uniformly random level-$1$  $X$-variable block from $T_{i',j',k'}^{\otimes n}\bigBk{\splres_X^{(L)}, \splres_Y^{(L)}, \splres_Z^{(L)}}$ and another uniformly random level-$1$ $X$-variable block from $T_{i-i', j-j', k-k'}^{\otimes n}\bigBk{\splres_X^{(R)}, \splres_Y^{(R)}, \splres_Z^{(R)}}$, then with very high probability ($1-2^{-\Omega(n)}$), the combination (interleaving the length $2^{\lvl-2}$ chunks between their level-$1$  block sequences) of them satisfies $\splres_X^{(L)} \times \splres_X^{(R)}$, up to $\eps$ additive error. Then the fraction of holes is $2^{-\Omega(n)}$. Similar reasons also apply to $Y$- and $Z$-variables.

\paragraph{Fixing the holes in all three dimensions.}
Suppose we have many ``broken'' copies of some tensor $T$, in each of which a small fraction of variables (holes) are missing. The goal of this step is to degenerate these broken tensors into one without holes. Sch{\"o}nhage~\cite{Schonhage81} solved this problem for matrix multiplication tensors with holes in only $X$- and $Y$-variables, but not $Z$, via an elegant linear transformation. 
Duan et al.~\cite{duan2023} introduced another method for so-called standard form tensors, which are quite general and are able to capture tensor products of constituent tensors, but can only deal with holes in a single dimension. Duan~\cite{duanpersonal} developed a  method utilizing an elegant recursive approach for fixing holes in all three dimensions, but only for matrix multiplication tensors. 

We generalize the method of \cite{duanpersonal} so that it can fix holes in all three dimensions simultaneously, while supporting a broad class of tensors similar to \cite{duan2023}. The only additional requirement compared to \cite{duan2023} is that the fraction of holes is below $O(1/\log N)$, where $N$ is the number of variables in the tensor $T$. This requirement is satisfied via the previous step of the fix.

Next, we provide some intuition of the recursive hole-fixing approach. Assume $T$ is supported on variable sets $X, Y, Z$, and the fraction of holes in every copy of $T$ does not exceed $c \ll 1$. We first take one broken  copy of $T$, which we call $T_{\textup{hole}}$, and let $X^{(0)}, Y^{(0)}, Z^{(0)}$ denote the set of holes in $T_{\textup{hole}}$; let $X^{(1)} \defeq X \setminus X^{(0)}$, $Y^{(1)} \defeq Y \setminus Y^{(0)}$, $Z^{(1)} \defeq Z \setminus Z^{(0)}$ represent the set of non-hole variables. We can further divide $T$ into the sum of eight subtensors:
\begin{align*}
  T \,=\, \sum_{a,b,c \in \midBK{0,1}} T \big\vert_{X^{(a)}, Y^{(b)}, Z^{(c)}} = T_{\textup{hole}} + \sum_{\substack{a,b,c \in \midBK{0,1} \\ 1 \in \midBK{a,b,c} }} T \big\vert_{X^{(a)}, Y^{(b)}, Z^{(c)}},
\end{align*}
where $T \vert_{X', Y', Z'}$ denotes the subtensor of $T$ over subsets of variables $X' \subseteq X$, $Y' \subseteq Y$, and $Z' \subseteq Z$. We directly use the broken copy $T_{\textup{hole}}$ for the first term, and recurse into seven subproblems to produce the other terms. In each subproblem, at least one of the variable sets is $X_1$, $Y_1$ or $Z_1$, which is $c$ times the size of $X$, $Y$, or $Z$. As long as $c$ is very small, the number of broken copies of $T$ used in this recursive algorithm is affordable. 

\paragraph{Rectangular matrix multiplication.}

In the analysis for square matrix multiplication, we could lower bound the ``symmetrized value'' of every constituent tensor $T_{i,j,k}$, which captures the asymptotic ability of $T_{i,j,k}^{\otimes n} \otimes T_{j,k,i}^{\otimes n} \otimes T_{k,i,j}^{\otimes n}$ to degenerate into matrix multiplication tensors. The reason why we could symmetrize the constituent tensors is that we want to obtain square matrix multiplication tensors $\angbk{a, a, a}$ for some $a$, which is symmetric about all three dimensions. The situation is different when we consider rectangular matrix multiplications, where we produce matrix multiplication tensors of the form $\angbk{a, a^{\K}, a}$ to bound $\omega(1, \K, 1)$. Thus, we no longer treat the analysis of each constituent tensor $T_{i,j,k}$ as an individual subproblem, because the proportion of $T_{i,j,k}$, $T_{j,k,i}$, and $T_{k,i,j}$ could be different. Hence, it is natural to adopt the framework introduced by Le Gall~\cite{legallrect} (and further used in \cite{legallrect2}) for rectangular matrix multiplication: we directly apply the laser method on a tensor consisting of multiple constituent tensors, e.g., on $\T = \bigotimes_{i,j,k} T_{i,j,k}^{\otimes \alpha(i,j,k) \cdot n}$, rather than doing this for every term $T_{i,j,k}^{\otimes \alpha(i,j,k) \cdot n}$ separately.

\paragraph{Difficulty of applying the refined laser method.}

Another natural attempt would be to combine our techniques with the refined laser method introduced in \cite{AlmanW21}, which aims to reduce the ``penalty term'' that arises when we deal with the block triples inconsistent with the selected distribution $\alpha$ but consistent with the marginals of $\alpha$. Alman and Vassilevska W.~\cite{AlmanW21} pick a collection of disjoint level-$\lvl$ block triples $X_I Y_J Z_K$ consistent with the chosen distribution $\alpha$, which we call the ``wanted'' triples. Then, they zero out a wanted triple with probability $1 - p$ and keep it with probability $p$. Any ``unwanted'' triple $X_{I'} Y_{J'} Z_{K'}$ only remains with probability $p^3$, since three involved variable blocks come from three different wanted triples and are zeroed out independently; in contrast, every wanted triple has probability $p$  to remain. The gap between $p$ and $p^3$ makes it a nontrivial improvement beyond the older method (increasing the modulus of hashing, see, e.g.,~\cite{stothers,virgi12,LeGall32power}), which produces a gap between $p$ and $p^2$.

However, a difficulty arises when the refined laser method is combined with the asymmetric hashing technique in \cite{duan2023} and this paper. Since we allow, e.g., level-$\l$ $Z$-variable blocks to be shared, we can no longer zero out all three blocks $X_I, Y_J, Z_K$ when we decide to give up on this triple, as $Z_K$ might be utilized by other wanted triples. If we only zero out $X_I$ and $Y_J$ simultaneously, the probability of remaining becomes $p$ (for a wanted triple) versus $p^2$ (for an unwanted triple), which results in the same bound as the older approach.

\section{Preliminaries}
\label{sec:prelim}

\subsection{Tensors and Tensor Operations}

\paragraph{Tensors.} A tensor $T$ over variable sets $X = \midBK{x_1, \ldots, x_{|X|}}$, $Y = \midBK{y_1, \ldots, y_{|Y|}}$, $Z = \midBK{z_1, \ldots, z_{|Z|}}$ and field $\F$ is a trilinear form
\[
  T = \sum_{i=1}^{|X|} \sum_{j=1}^{|Y|} \sum_{k=1}^{|Z|} a_{i,j,k} \cdot x_i y_j z_k,
\]
where all $a_{i,j,k}$ are from $\F$. $X, Y, Z$ are also called the \emph{support} of the tensor. If all $a_{i,j,k} \in \midBK{0, 1}$, the tensor $T$ can be considered as over any field $\F$, which is the case for all tensors involved in this paper.

In the following, assume $T$ is a tensor over $X = \midBK{x_1, \ldots, x_{|X|}}$, $Y = \midBK{y_1, \ldots, y_{|Y|}}$, $Z = \midBK{z_1, \ldots, z_{|Z|}}$ and $T'$ is a tensor over $X' = \midBK{x'_1, \ldots, x'_{|X'|}}$, $Y' = \midBK{y'_1, \ldots, y'_{|Y'|}}$, $Z' = \midBK{z'_1, \ldots, z'_{|Z'|}}$, written as
\[
  T = \sum_{i=1}^{|X|} \sum_{j=1}^{|Y|} \sum_{k=1}^{|Z|} a_{i,j,k} \cdot x_i y_j z_k, \qquad
  T' = \sum_{i=1}^{|X'|} \sum_{j=1}^{|Y'|} \sum_{k=1}^{|Z'|} b_{i,j,k} \cdot x'_i y'_j z'_k,
\]
\paragraph{Tensor Operations.} Recall the following tensor operations between two tensors $T$ and $T'$:
\begin{itemize}
\item The \emph{sum} $T + T'$ is only defined when both tensors are supported on the same sets $(X, Y, Z) = (X', Y', Z')$, given by
  \[
    T + T' = \sum_{i=1}^{|X|} \sum_{j=1}^{|Y|} \sum_{k=1}^{|Z|} (a_{i,j,k} + b_{i,j,k}) \cdot x_i y_j z_k.
  \]
\item The \emph{direct sum} $T \oplus T'$ equals the sum $T + T'$ over disjoint unions $X \sqcup X'$, $Y \sqcup Y'$, and $Z \sqcup Z'$, i.e., we first relabel the variables so that $T$ and $T'$ have disjoint supports, and then take their sum. If $T$ and $T'$ are supported on disjoint variable sets, their sum is the same as their direct sum, in which case we say $T$ and $T'$ are \emph{independent}. We write $T^{\oplus n} \defeq \underbrace{T \oplus T \oplus \cdots \oplus T}_{n~\textup{copies}}$ to denote the sum of $n$ independent copies of $T$.
\item The \emph{tensor product}, a.k.a.~the \emph{Kronecker product}, is defined as the tensor
  \[
    T \otimes T' = \sum_{i=1}^{|X|} \sum_{j=1}^{|Y|} \sum_{k=1}^{|Z|} \sum_{i'=1}^{|X'|} \sum_{j'=1}^{|Y'|} \sum_{k'=1}^{|Z'|} a_{i,j,k} \cdot b_{i', j', k'} \cdot (x_i, x'_{i'}) \cdot (y_j, y'_{j'}) \cdot (z_k, z'_{k'})
  \]
  over variable sets $X \times X'$, $Y \times Y'$, and $Z \times Z'$. We write $T^{\otimes n} \defeq \underbrace{T \otimes T \otimes \cdots \otimes T}_{n~\textup{times}}$ to denote the \emph{$n$-th tensor power} of $T$.
\item We say $T$ and $T'$ are \emph{isomorphic}, denoted by $T \equiv T'$, if $|X| = |X'|$, $|Y| = |Y'|$, $|Z| = |Z'|$, and there are permutations $\pi_X, \pi_Y, \pi_Z$ over $[|X|], [|Y|], [|Z|]$ respectively, such that $a_{i,j,k} = b_{\pi_X(i), \pi_Y(j), \pi_Z(k)}$ for all $i, j, k$. In other words, both tensors are equivalent up to a relabeling of the variables. 
\end{itemize}

\subsection{Tensor Rank}

Given a tensor $T$ over $X, Y, Z$, the tensor rank $R(T)$ is defined to be the minimum integer $r\ge 0$ such that $T$ can be written as 
\[T = \sum_{t = 1}^r \lpr{\sum_{i = 1}^{|X|} a_{t,i}\cdot x_i} \lpr{\sum_{j = 1}^{|Y|} b_{t,j}\cdot y_j} \lpr{\sum_{k = 1}^{|Z|} c_{t,k}\cdot z_k},\]
where the above sum is called the \emph{rank decomposition} of $T$.

Given two tensors $T, T'$, the tensor rank satisfies the following property with respect to tensor operations.

\begin{itemize}
    \item $R(T+T')\le R(T) + R(T')$.
    \item $R(T\oplus T')\le R(T) + R(T')$.
    \item $R(T\otimes T')\le R(T)\cdot R(T').$
\end{itemize}

The \emph{asymptotic rank} $\Tilde{R}(T)$ of $T$ is defined as 
\[\Tilde{R}(T) := \lim_{n\to \infty}\lpr{R(T^{\otimes n})}^{1/n}.\]
Due to the third item above and Fekete's lemma, the asymptotic rank is well-defined and upper bounded by $R(T^{\otimes m})^{1/m}$ for any fixed integer $m > 0$.

\subsection{Degenerations, Restrictions, Zero-outs}
Let $T$ be a tensor over $X, Y, Z$ and $T'$ be a tensor over $X', Y', Z'$. Both $T$ and $T'$ are tensors over a field $\F$.

\paragraph{Degeneration.}
Let $\F[\lambda]$ be the ring of polynomials of the formal variable $\lambda$. We say that $T'$ is a degeneration of $T$, written as $T \unrhd T'$, if there exists $\F[\lambda]$-linear maps 
\begin{align*}
    \phi_X &: \Span_{\F[\lambda]}(X)\to \Span_{\F[\lambda]}(X'),\\
    \phi_Y &: \Span_{\F[\lambda]}(Y)\to \Span_{\F[\lambda]}(Y'),\\
    \phi_Z &: \Span_{\F[\lambda]}(Z)\to \Span_{\F[\lambda]}(Z'),
\end{align*}
and $d\in \N$ such that  
\[T' = \lambda^{-d}\lpr{\sum_{i = 1}^{|X|}\sum_{j = 1}^{|Y|}\sum_{k = 1}^{|Z|}a_{i,j,k} \cdot \phi_X(x_i) \cdot \phi_Y(y_j) \cdot \phi_Z(z_k)} + O(\lambda).\]

It is not hard to check that if $T'\unrhd T$, then $\Tilde{R}(T')\le \Tilde{R}(T)$.

\paragraph{Restriction.}
Restriction is a special type of degeneration that considers the case where the maps $\phi_X,\phi_Y, \phi_Z$ are $\F$-linear maps. More specifically, $T'$ is a restriction of $T$ if there exist $\F$-linear maps
\begin{align*}
    \phi_X &: \Span_{\F}(X)\to \Span_{\F}(X'),\\
    \phi_Y &: \Span_{\F}(Y)\to \Span_{\F}(Y'),\\
    \phi_Z &: \Span_{\F}(Z)\to \Span_{\F}(Z'),
\end{align*}
such that
\[T' = \sum_{i = 1}^{|X|}\sum_{j = 1}^{|Y|}\sum_{k = 1}^{|Z|}a_{i,j,k} \cdot \phi_X(x_i) \cdot \phi_Y(y_j) \cdot \phi_Z(z_k).\]
It is not hard to see that since the maps $\phi_X, \phi_Y,\phi_Z$ are linear transformations, we have $R(T')\le R(T)$ and consequently $\Tilde{R}(T')\le \Tilde{R}(T)$.

\paragraph{Zero-out.}

In the laser method, we only consider a limited type of restriction called zero-outs, namely the maps $\phi_X, \phi_Y, \phi_Z$ set some variables to zero. More specifically, we choose subsets $X'\subseteq X$, $Y'\subseteq Y$, $Z'\subseteq Z$ and define the maps as 
\[\phi_X(x_i) = \begin{cases}x_i & \text{If }x_i\in X',\\ 0 & \text{otherwise},\end{cases}\]
and similarly for $\phi_Y, \phi_Z$. The resulting tensor 
\[T' = \sum_{i = 1}^{|X|}\sum_{j = 1}^{|Y|}\sum_{k = 1}^{|Z|}a_{i,j,k} \cdot \phi_X(x_i) \cdot \phi_Y(y_j) \cdot \phi_Z(z_k) = \sum_{x_i\in X'}\sum_{y_j\in Y'}\sum_{z_k\in Z'}a_{i,j,k} \cdot x_iy_jz_k\]
is called a zero-out of $T$. Throughout this paper, we use the notation $T' = T\vert_{X',Y',Z'}$ to denote such a tensor $T'$ obtained as a zero-out of $T$ and we say that the variables in $X\setminus X'$, $Y\setminus Y'$, $Z\setminus Z'$ are zeroed out. In this case, we also call $T'$ the \emph{subtensor} of $T$ over $X', Y', Z'$.

\subsection{Matrix Multiplication Tensors}

For positive integers $a,b,c$, the $a\times b \times c$ matrix multiplication tensor $\ang{a,b,c}$ is a tensor over the variable sets $\{x_{ij}\}_{i\in [a], j\in [b]}, \{y_{jk}\}_{j\in [b],k\in[c]}, \{z_{ki}\}_{i\in [a],k\in [c]}$ defined as the tensor computing the $a\times c$ product matrix
$\{z_{ki}\}_{i\in [a], k\in [c]}$
of an $a\times b$ matrix $\{x_{ij}\}_{i\in [a], j\in [b]}$ and $b\times c$ matrix $\{y_{jk}\}_{j\in [b], k\in [c]}$. Specifically, $\ang{a,b,c}$ can be written as the trilinear form
\[\ang{a,b,c} = \sum_{i\in [a]}\sum_{j\in [b]}\sum_{k\in [c]} x_{ij} y_{jk} z_{ki}.\]
It is not hard to check that $\ang{a,b,c}\otimes \ang{d,e,f} \equiv \ang{ad, be, cf}$.

Following from the recursive approach introduced by Strassen in \cite{strassen}, for any integer $q \ge 2$, if $R(\ang{q,q,q})\le r$, then one can use the rank decomposition of $\ang{q,q,q}$ to design an arithmetic circuit of size $O(n^{\log_q(r)})$ to multiply two $n \times n$ matrices. This motivates the definition of the \emph{matrix multiplication exponent} $\omega$ as follows:
\[\omega := \inf_{q \in \N, \, q \ge 2} \log_q (R(\ang{q,q,q})).\]
Namely, for every $\eps > 0$, there exists an arithmetic circuit of size $O(n^{\omega + \eps})$ that computes the multiplication of two $n\times n$ matrices. Since $\ang{q,q,q}^{\otimes n} \equiv \ang{q^n, q^n, q^n}$, equivalently $\omega$ can be written in terms of the asymptotic rank of $\ang{q,q,q}$ as 
\[\omega = \log_q (\Tilde{R}(\ang{q,q,q})).\]

In this paper, we also consider the arithmetic complexity of multiplying rectangular matrices of sizes $n^a\times n^b$ and $n^b\times n^c$ where $a,b,c \in \R_{\ge 0}$. We define the quantity $\omega(a,b,c)$ similar to $\omega$ as 
\[\omega(a,b,c) = \log_q\lpr{\Tilde{R}(\ang{q^a,q^b,q^c})}\]
where $q \ge 2$ is a positive integer. This means that for any $\eps > 0$, there exists an arithmetic circuit of size $O(n^{\omega(a,b,c) + \eps})$ that computes the multiplication of an $n^a\times n^b$ matrix with an $n^b\times n^c$ matrix. In this paper, we focus on bounds for the values of the form $\omega(1,\kappa,1)$ for $\kappa > 0$. We remark that it is known that $\omega(1,1,\kappa) =\omega(1,\kappa,1) = \omega(\kappa,1,1)$.

\subsection{Sch\"onhage's Asymptotic Sum Inequality}

By the above definition of $\omega$, it is clear that if one can bound the asymptotic rank of matrix multiplication tensors, then one would get an upper bound on $\omega$. In fact, Sch\"onhage showed in \cite{Schonhage81} that one can obtain an upper bound on $\omega$ if one can bound the asymptotic rank of a direct sum of matrix multiplication tensors. Specifically, we recall Sh\"onhage's asymptotic sum inequality as follows.

\begin{theorem}[Asymptotic Sum Inequality \cite{Schonhage81}]\label{thm:schonhage-ineq}
For positive integers $r > m$ and $a_i, b_i,c_i$ for $i\in [m]$, if 
\[\Tilde{R}\bk{\bigoplus_{i = 1}^m \ang{a_i,b_i,c_i}}\le r,\]
then $\omega\le 3\tau$ where $\tau\in [2/3,1]$ is the solution to the equation
\[\sum_{i = 1}^m (a_i\cdot b_i\cdot c_i)^\tau = r.\]
\end{theorem}

Analogously, the asymptotic sum inequality can also be used to obtain bounds on the rectangular matrix multiplication as follows.

\begin{theorem}[Asymptotic Sum Inequality for $\omega(a,b,c)$ \cite{Schonhage81}]\label{thm:schonhage-ineq-rect}
    Let $t, \, q > 0$ be positive integers and $a,b,c \ge 0$ , then
    \[t\cdot q^{\omega(a,b,c)}\le \Tilde{R}\lpr{\bigoplus_{i = 1}^t \ang{q^a, q^b,q^c}}.\]
\end{theorem}

\subsection{The Coppersmith-Winograd Tensor}

For a nonnegative integer $q\ge 0$, the Coppersmith-Winograd tensor $\CW_q$ over the variables $X = \{x_0,\dots, x_{q+1}\}$, $Y = \{y_0,\dots, y_{q+1}\}$, $Z = \{z_0,\dots, z_{q+1}\}$ is defined as
\[\CW_q := x_0y_0z_{q+1} + x_0y_{q+1}z_0 + x_{q+1}y_0z_0 + \sum_{i = 1}^q \lpr{x_0y_iz_i + x_iy_0z_i + x_iy_iz_0}.\]
Observe that 
\[\sum_{i = 1}^q x_0y_iz_i + \sum_{i = 1}^q x_iy_0z_i + \sum_{i = 1}^q x_iy_iz_0 \equiv \ang{1,1,q} + \ang{q,1,1} + \ang{1,q,1},\]
so $\CW_q$ is the sum of six matrix multiplication tensors where the other three are copies of $\ang{1,1,1}$. It is known from Coppersmith and Winograd \cite{cw90} that $\Tilde{R}(\CW_q) \le q+2$.

\subsection{Base Leveled Partition of \texorpdfstring{$\CW_q$}{CWq}}

We will consider the $2^{\lvl-1}$-th tensor power of $\CW_q$ for $\lvl \ge 1$. For convenience, we use the notation $T^{(\lvl)} := \CW_q^{\otimes 2^{\lvl-1}}$. There is a natural partitioning of the variables of $\CW_q$ introduced in \cite{cw90} and consequently used in all following works including \cite{virgi12,AlmanW21,LeGall32power,duan2023}. We now describe the leveled partition of $T^{(\lvl)}$.

\paragraph{Level-$1$ Partition.} For $T^{(1)} = \CW_q$, its variable sets $X^{(1)}, Y^{(1)}, Z^{(1)}$ are partitioned into three parts
\begin{align*}
    X^{(1)} &= X^{(1)}_0 \sqcup X^{(1)}_1 \sqcup X^{(1)}_2 = \{x_0\}\sqcup \{x_1,\dots, x_q\}\sqcup \{x_{q+1}\},\\
    Y^{(1)} &= Y^{(1)}_0 \sqcup Y^{(1)}_1 \sqcup Y^{(1)}_2 = \{y_0\}\sqcup \{y_1,\dots, y_q\}\sqcup \{y_{q+1}\},\\
    Z^{(1)} &= Z^{(1)}_0 \sqcup Z^{(1)}_1 \sqcup Z^{(1)}_2 =  \{z_0\}\sqcup \{z_1,\dots, z_q\}\sqcup \{z_{q+1}\}.
\end{align*}
We use $T_{i,j,k}^{(1)}$ to denote the subtensor $T^{(1)}\vert_{X_i,Y_j,Z_k}$ and we call $T_{i,j,k}^{(1)}$ a \emph{level-$1$ constituent tensor}. Then notice that under the above partition, the constituent tensor $T^{(1)}_{i, j, k}$ is nonzero if and only if $i+j+k = 2$. In particular, we can write $\CW_q$ as a sum of constituent tensors as follows
\[T^{(1)} = \CW_q = \sum_{\substack{i,j,k\ge 0\\ i+j+k = 2}} T_{i,j,k}^{(1)}.\]

\paragraph{Level-$\lvl$ Partition.} For $T^{(\lvl)}=\CW_q^{\otimes 2^{\lvl-1}}$ with variable sets $X^{(\lvl)}, Y^{(\lvl)}, Z^{(\lvl)}$, the above level-$1$ partition on $T^{(1)}$ directly induces a partition on the variable sets $X^{(\lvl)}, Y^{(\lvl)}, Z^{(\lvl)}$ where each part of the partition is indexed by a $\{0, 1, 2\}$-sequence of length $2^{\lvl-1}$. Specifically, this gives the partition
\[X^{(\lvl)} = \bigsqcup_{(\hat{i}_1, \hat{i}_2, \ldots, \hat{i}_{2^{\lvl-1}}) \in \{0, 1, 2\}^{2^{\lvl-1}}}X^{(1)}_{\hat{i}_1} \otimes  X^{(1)}_{\hat{i}_2} \otimes \cdots \otimes X^{(1)}_{\hat{i}_{2^{\lvl-1}}}\]
for $X$-variables and analogous partitions for $Y$- and $Z$-variables. 

In order to obtain an improvement by analyzing higher tensor powers of $\CW_q$, we need to consider the following coarsening of the induced partition where the parts corresponding to sequences with the same sum are ``merged'' into a single part. More specifically, we have
\[
  X^{(\lvl)} = \bigsqcup_{i = 0}^{2^\lvl} X_i^{(\lvl)},
  \qquad \textup{where} \quad
  X_i^{(\lvl)} \defeq
  \bigsqcup_{\substack{(\hat{i}_1, \hat{i}_2, \ldots, \hat{i}_{2^{\lvl-1}}) \in \{0, 1, 2\}^{2^{\lvl-1}}  \\\sum_t \hat{i}_t = i}}X^{(1)}_{\hat{i}_1} \otimes  X^{(1)}_{\hat{i}_2} \otimes \cdots \otimes X^{(1)}_{\hat{i}_{2^{\lvl-1}}}.
\]
 We refer to this above coarsened partition of $T^{(\lvl)}$ as the \emph{level-$\lvl$ partition}.
 Note that we can also view this partition as obtained from coarsening the level-($\lvl-1$) partition, i.e.,
 \[X_i^{(\lvl)} = \bigsqcup_{\substack{0 \le i' \le i \\ 0 \le i', i - i' \le 2^{\lvl}}}X^{(\lvl-1)}_{i'} \otimes X^{(\lvl-1)}_{i-i'}.\]
 We can partition the variable sets $Y^{(\lvl)}$ and $Z^{(\lvl)}$ similarly. 
 
 Under the level-$\lvl$ partition, we use $T^{(\lvl)}_{i, j, k}$ to denote the subtensor $T^{(\lvl)}\vert_{X_i^{(\lvl)}, Y_j^{(\lvl)}, Z_k^{(\lvl)}}$ and note that $T^{(\lvl)}_{i, j, k}$ is nonzero if and only if $i+j+k = 2^{\lvl}$. So we have 
 \[T^{(\lvl)} = \CW_q^{\otimes 2^{\lvl-1}}= \sum_{\substack{i,j,k\ge 0\\ i+j+k = 2^{\lvl}}} T_{i,j,k}^{(\lvl)}.\]
 We call each $T^{(\lvl)}_{i, j, k}$ a \emph{level-$\lvl$ constituent tensor}, $X^{(\lvl)}_i, Y^{(\lvl)}_j, Z^{(\lvl)}_k$ level-$\lvl$ variable blocks, and we omit the superscript $(\lvl)$ when $\lvl$ is clear from context.

\subsection{Leveled Partition for Large Tensor Powers of \texorpdfstring{$\CW_q$}{CWq}}

In the laser method, we consider a large tensor power of $\CW_q$ in the form $(T^{(\lvl)})^{\otimes n} = (\CW_q)^{\otimes n\cdot 2^{\lvl-1}}$. We set $N := n\cdot 2^{\lvl-1}$ and note that the leveled partition of $T^{(\lvl)}$ induces a partition on $(T^{\lvl})^{\otimes n}$. We recall some basic terminology and notations with respect to the leveled-partition of $(T^{\lvl})^{\otimes n}$.

\paragraph{Level-$1$ partition of $(\CW_q)^{\otimes N}$.}
In level-$1$, we view $(\CW_q)^{\otimes N}$ as the tensor $(T^{(1)})^{\otimes N}$ and consider the partition induced by the level-$1$ partition on $T^{(1)}$. Each level-$1$ $X$-variable block $X_{\hat{I}}$ is indexed by a sequence $\hat{I} = (\hat{I}_1,\dots, \hat{I}_N)$ of length $N$ in $\{0,1,2\}^N$. The variable block $X_I$ is defined as 
\[X_{\hat{I}} := X_{\hat{I}_1}^{(1)} \otimes \dots \otimes X_{\hat{I}_N}^{(1)},\]
where $X_{I_t}^{(1)}$ for $t\in [N]$ is the level-$1$ partition of $T^{(1)}$. We call $X_{\hat{I}}$ a \emph{level-$1$ variable block} and $\hat{I}$ its \emph{level-$1$ index sequence}. The level-$1$ $Y$- and $Z$-variable blocks $Y_{\hat{J}}$ and $Z_{\hat{K}}$ are defined similarly for level-$1$ index sequences $\hat{J}, \hat{K}\in \{0,1,2\}^N$. Then notice that $X_{\hat{I}}, Y_{\hat{J}}, Z_{\hat{K}}$ form a nonzero subtensor of $(T^{(1)})^{\otimes N}$ if $\hat{I}_t + \hat{J}_t + \hat{K}_t = 2$ for all $t\in [N]$. So we can write $(T^{(1)})^{\otimes N}$ as a sum of subtensors
\[(T^{(1)})^{\otimes N} = \sum_{\substack{\hat{I}, \hat{J}, \hat{K}\in \{0,1,2\}^{N}\\ \hat{I}_t + \hat{J}_t + \hat{K}_t = 2\,\ \forall t\in [N]}} (T^{(1)})^{\otimes N} \big\vert_{X_{\hat{I}}, Y_{\hat{J}}, Z_{\hat{K}}}.\]
For convenience, we use $X_{\hat{I}}Y_{\hat{J}}Z_{\hat{K}}$ to denote the subtensor $(T^{(1)})^{\otimes N}\vert_{X_{\hat{I}}, Y_{\hat{J}}, Z_{\hat{K}}}$ and we call $X_{\hat{I}}Y_{\hat{J}}Z_{\hat{K}}$ a \emph{level-$1$ triple}.

\paragraph{Level-$\lvl$ partition of $(\CW_q)^{\otimes N}$.}
In level-$\lvl$, we view $(\CW_q)^{\otimes N}$ as the tensor $(T^{(\lvl)})^{\otimes n}$ where $n = N / 2^{\lvl - 1}$ and consider the partition induced by the level-$\lvl$ partition on $T^{(\lvl)}$. Each level-$1$ $X$-variable block $X_{I}$ is indexed by a sequence $I\in \{0,1,\dots, 2^{\lvl}\}^n$ of length $n$. The variable block $X_I$ is defined as 
\[X_I :=  X_{I_1}^{(\lvl)}\otimes \dots \otimes X_{I_n}^{(\lvl)}\]
where $X_{i}^{(\lvl)}$ ($0 \le i \le 2^{\lvl}$) is the $i$-th part in the level-$\lvl$ partition of $T^{(\lvl)}$. We call $X_{I}$ a \emph{level-$\lvl$ variable block} and $I$ its \emph{level-$\lvl$ index sequence}. The level-$\lvl$ $Y$- and $Z$-variable blocks $Y_{J}$ and $Z_{K}$ are defined similarly for level-$\lvl$ index sequences $J,K\in \{0,1,\dots, 2^{\lvl}\}^n$. Similarly, the level-$\lvl$ variable blocks $X_I, Y_J, Z_K$ form a nonzero subtensor of $(T^{(\lvl)})^{\otimes n}$ when $I_t + J_t + K_t = 2^\lvl$ for all $t\in [n]$. So we can write
\[(T^{(\lvl)})^{\otimes n} = \sum_{\substack{\hat{I}, \hat{J}, \hat{K}\in \{0,1,2^\lvl\}^{n}\\ I_t+J_t +K_t = 2^\lvl\,\ \forall t\in [N]}}(T^{(\lvl)})^{\otimes n}\vert_{X_I, Y_J, Z_K}.\]
For convenience, we use the notation $X_IY_JZ_K$ to denote the subtensor $(T^{(\lvl)})^{\otimes n}\vert_{X_I, Y_J, Z_K}$ and we call such $X_IY_JZ_K$ a \emph{level-$\lvl$ triple.}

In addition, note that since the level-$\lvl$ partition of $T^{(\lvl)}$ is a coarsening of the partition induced by the level-$1$ partition of $T^{(1)}$, a level-$1$ variable block $X_{\hat{I}}$ is contained in a level-$\lvl$ variable block $X_I$ if the sequence $I' = (I'_1,\dots, I'_n)$ formed by taking $I'_t = \sum_{i = 1}^{2^{\lvl-1}} \hat{I}_{(t-1)\cdot 2^{\lvl-1} + i}$ satisfies $I'_t = I_t$ for all $t\in [n]$. Namely, if taking the sum of consecutive length-$2^{\lvl-1}$ subsequences in $\hat{I}$ yields the sequence $I$, then $X_{\hat{I}}$ is contained in $X_I$. In this case, we use the notation $\hat{I}\in I$ and $X_{\hat{I}}\in X_I$.

\subsection{Distributions and Entropy}

In this paper, we only consider distributions with a finite support. Let $\alpha$ be a distribution supported on a set $S$, we have $\alpha(s)\ge 0$ for all $s\in S$ and $\sum_{s\in S}\alpha(s) = 1$. The \emph{entropy} of $\alpha$, denoted as $H(\alpha)$, is defined as
\[H(\alpha) \defeq - \sum_{\substack{s\in S\\ \alpha(s) > 0}}\alpha(s)\log \alpha(s), \]
where the $\log$ has base $2$. 
We will frequently use the following well-known combinatorial fact.
\begin{lemma}
  Let $\alpha$ be a distribution over the set $[s] = \{1,\dots, s\}$. Let $N > 0$ be a positive integer, then we have
  \[\binom{N}{\alpha(1)N,\dots, \alpha(s)N}= 2^{N(H(\alpha) \pm o(1))}.\]
\end{lemma}

 For two distributions $\alpha$ and $\beta$ over the sets $S$ and $S'$ respectively, we define the joint distribution $\alpha \times \beta$ as the distribution over $S\times S' = \{(s,s') \mid s\in S, s'\in S'\}$ such that
 \[(\alpha\times \beta) (s,s') = \alpha(s)\cdot \beta(s').\]
  When $S$ and $S'$ are sets of integer sequences, we will instead define $\alpha \times \beta$ as a distribution over all integer sequences that can be obtained by concatenating one sequence in $S$ and one sequence in $S'$, such that 
$$(\alpha \times \beta)(s \circ s') = \alpha(s) \cdot \beta(s'), $$
where $s \circ s'$ denotes the concatenation of $s$ and $s'$. 

\subsection{Complete Split Distributions}

Motivated by the leveled partition of tensor powers of $\CW_q$, we define the notion of complete split distributions to characterize the level-$1$ variable blocks contained in level-$\lvl$ variable blocks.

\begin{definition}[Complete Split Distribution]
    A \emph{complete split distribution} for a level-$\lvl$ constituent tensor $T_{i,j,k}$ with $i+j+k=2^{\lvl}$ is a distribution on all length $2^{\lvl-1}$ sequences $(\hat{i}_1, \hat{i}_2, \ldots, \hat{i}_{2^{\lvl-1}}) \in \{0, 1, 2\}^{2^{\lvl-1}}$. 
\end{definition}

For a level-$1$ index sequence $\hat I \in \{0, 1, 2\}^{2^{\lvl-1} \cdot n}$, we say that it is \emph{consistent} with a complete split distribution $\splres$ if the proportion of any  index sequence $(\hat{i}_1, \hat{i}_2, \ldots, \hat{i}_{2^{\lvl-1}})$ in
\[ \left\{ \bigbk{\hat{I}_{(t-1)\cdot 2^{\lvl - 1}+p}}_{p=1}^{2^{\lvl-1}} \;\middle|\; t \in [n] \right\} \]
equals $\splres(\hat{i}_1, \hat{i}_2, \ldots, \hat{i}_{2^{\lvl-1}})$. Namely, for every $(\hat{i}_1,\dots, \hat{i}_{2^{\lvl-1}})\in \{0,1,2\}^{2^{\lvl-1}}$, we have
\[\labs{\lcr{t \in [n] \mmid \bigbk{\hat{I}_{(t-1)\cdot 2^{\lvl - 1}+p}}_{p=1}^{2^{\lvl - 1}} = (\hat{i}_1,\dots, \hat{i}_{2^{\lvl-1}})}}  = \splres(\hat{i}_1, \hat{i}_2, \ldots, \hat{i}_{2^{\lvl-1}}) \cdot n.\]

Notice that any level-1 index sequence $\hat{I}\in \{0,1,2\}^{2^{\lvl - 1}\cdot n}$ defines a complete split distribution by computing the proportions of each type of length-$2^{\lvl-1}$ consecutive chunks present in $\hat{I}$. More specifically, we have the following definition.

\begin{definition}\label{def:split-hatI}
    Given a level-1 index sequence $\hat{I}\in \{0,1,2\}^{2^{\lvl - 1}\cdot n}$, its complete split distribution over $(\hat{i}_1,\dots, \hat{i}_{2^{\lvl-1}})\in \{0,1,2\}^{2^{\lvl - 1}}$ is defined as
    \[\split\bigbk{\hat{I}}\bigbk{\hat{i}_1,\dots, \hat{i}_{2^{\lvl-1}}} = \frac{1}{n}\cdot \labs{\lcr{t \in [n] \mmid \bigbk{\hat{I}_{(t-1)\cdot 2^{\lvl - 1}+p}}_{p=1}^{2^{\lvl - 1}} = \bigbk{\hat{i}_1,\dots, \hat{i}_{2^{\lvl-1}}}}}.\]
    Given a subset $S\subseteq [n]$, we can define the complete split distribution over $(\hat{i}_1,\dots, \hat{i}_{2^{\lvl-1}})\in \{0,1,2\}^{2^{\lvl - 1}}$ given by $\hat{I}$ restricted to the subset $S$ as 
    \[\split\bigbk{\hat{I}, S}\bigbk{\hat{i}_1,\dots, \hat{i}_{2^{\lvl-1}}} = \frac{1}{|S|}\cdot \labs{\lcr{t \in S \mmid \bigbk{\hat{I}_{(t-1)\cdot 2^{\lvl - 1}+p}}_{p=1}^{2^{\lvl - 1}} = \bigbk{\hat{i}_1,\dots, \hat{i}_{2^{\lvl-1}}}}}.\]
\end{definition}

Given two complete split distributions $\splres_1$ and $\splres_2$ over the length-$2^{\lvl-1}$ index sequences $\{0,1,2\}^{2^{\lvl - 1}}$, the $L_\infty$ distances between $\splres_1$ and $\splres_2$ is defined to be 
\[\norm{\splres_1 - \splres_2}_\infty = \max_{\sigma\in \{0,1,2\}^{2^{\lvl-1}}}|\splres_1(\sigma) - \splres_2(\sigma)|.\]
For any constant $\eps > 0$ and a fixed complete split distribution $\splres$, we say that a level-1 index sequence $\hat{I}\in \{0,1,2\}^{2^{\lvl - 1}\cdot n}$ is consistent with $\splres$ up to $\eps$ error if  $\midnorm{\split(\hat{I}) - \splres}_\infty \le \eps$. When the $\eps$ is clear from context, we say that $\hat{I}$ is approximately consistent with $\splres$ if it is consistent with $\splres$ up to $\eps$ error.

\begin{definition}
    For a level-$\lvl$ constituent tensor $T_{i,j,k}$, an integer exponent $N$, a constant $\eps \ge 0$, and three complete split distributions $\splres_X, \splres_Y, \splres_Z$ for the $X$-, $Y$-, $Z$-variables respectively, we define
    $$T_{i,j,k}^{\otimes N}[\splres_X, \splres_Y, \splres_Z,\eps] := \sum_{\substack{\text{level-}1 \text{ triple } X_{\hat I}Y_{\hat J}Z_{\hat K} \textup{ in } T_{i,j,k}^{\otimes N} \\ \hat I \text{ approximately consistent with } \splres_X \\ \hat J \text{ approximately consistent with } \splres_Y \\ \hat K \text{ approximately consistent with } \splres_Z}} X_{\hat I} Y_{\hat J} Z_{\hat K}. $$
    It is a subtensor of $T_{i,j,k}^{\otimes N}$ over all level-1 $X$-, $Y$-, $Z$-variable blocks that are approximately consistent with $\splresX$, $\splresY$, $\splresZ$, respectively. When $\eps = 0$, we will simplify the notation to $T_{i,j,k}^{\otimes N}[\splres_X, \splres_Y, \splres_Z]$.
\end{definition}

\subsection{Salem-Spencer Sets}
In the hashing step of the laser method, we make use of the existence of a large dense subset of $\Z_M$ that avoids $3$-term arithmetic progressions. We recall the following past result.

\begin{theorem}[\cite{salemspencer,behrend1946sets}]
 For every positive integer $M > 0$, there exists a subset $B\subseteq \Z_M$ of size 
 \[|B|\ge M\cdot e^{-O(\sqrt{\log M})} = M^{1-o(1)}\]
 that contains no nontrivial $3$-term arithmetic progressions. Specifically, any $a,b,c\in B$ satisfy $a+b \equiv 2c \pmod M$ if and only if $a = b = c$.
\end{theorem}

\section{Algorithm Outline}
\label{sec:outline}

In the following, we will use $\K \ge 0$ to denote that we want to obtain an upper bound on $\omega(1, \K, 1)$. 

In this section, we give the outline of our algorithm, which accepts $\CW_q^{\otimes N}$ as its input for a large enough $N$, and degenerates it into a collection of independent matrix multiplication tensors of the same size $\angbk{m, m^\K, m}$. By the asymptotic sum inequality (\cref{thm:schonhage-ineq-rect}), this will give an upper bound on $\omega(1, \K, 1)$.

\subsection{Algorithm Framework}

The following notion of \emph{interface tensor} acts as an interface of our algorithm between different levels. In general, each level of our algorithm takes an interface tensor as input (except the first level, which takes a large tensor power of $\CW_q$), and degenerates it into independent copies of an interface tensor. 

\begin{definition}[Interface Tensor]\label{def:interface-tensor}
  For a positive integer $\lvl\ge 1$ and any constant $0\le \eps \le 1$, a level-$\lvl$ $\eps$-\emph{interface tensor} $\T^*$ with parameter list 
  \[\{(n_t, i_t, j_t, k_t, \splresXt, \splresYt, \splresZt)\}_{t \in [s]}\]
  is defined as
  \[
    \T^* \defeq \bigotimes_{t = 1}^{s} T_{i_t, j_t, k_t}^{\otimes n_t}[\splresXt, \splresYt, \splresZt, \eps],
  \]
  where $i_t+j_t+k_t = 2^{\lvl}$ for every $t\in [s]$ (i.e., $T_{i_t, j_t, k_t}$ is a level-$\lvl$ constituent tensor) and  $\splresXt, \splresYt, \splresZt$ are level-$\lvl$ complete split distributions for $X$-, $Y$-, $Z$-variables respectively. We call each $T_{i_t, j_t, k_t}^{\otimes n_t}[\splresXt, \splresYt, \splresZt, \eps]$ a \emph{term} of $\T^*$. When $\eps = 0$, we will simply call $\T^*$ a level-$\lvl$ interface tensor. 
\end{definition}

Note that the same $(i_t, j_t, k_t)$ can appear multiple times in the parameter list, with potentially different $n_k, \splresXt, \splresYt, \splresZt$. Also note that the tensor product of two level-$\lvl$ $\eps$-interface tensors is also a level-$\lvl$ $\eps$-interface tensor, whose parameter list is the concatenation of the parameter lists of the two level-$\lvl$ $\eps$-interface tensors.

The framework of our algorithm is as follows. First, we apply the \emph{global stage} algorithm described in \cref{sec:global} on input $\bigbk{\CW_q^{\otimes 2^{\lvl^*}}}^{\otimes n}$ to degenerate it into independent copies of a level-$\lvl^*$ $\eps_{\lvl^*}$-interface tensor.
Then we apply the \emph{constituent tensor stage} algorithm described in \cref{sec:constituent} for $\lvl = \lvl^* ,\lvl^*-1,\dots, 2$ to obtain the tensor product between a matrix multiplication tensor and independent copies of a level-$1$ $\eps_1$-interface tensor. More specifically, the constituent tensor stage algorithm takes as input a level-$\lvl$ $\eps_\lvl$-interface tensor and outputs the tensor product between a matrix multiplication tensor and independent copies of a level-$(\lvl-1)$ $\eps_{\lvl-1}$-interface tensor, so we can keep applying the constituent tensor stage algorithm on each level-$(\lvl-1)$ interface tensors that was outputted previously until we get a  tensor product between a matrix multiplication tensor and independent copies of a  level-$1$ $\eps_1$-interface tensor. Finally, we show that each level-$1$ $\eps_1$-interface tensor can be easily degenerated into a matrix multiplication tensor, so we obtain independent copies of matrix multiplication tensors of dimension $\angbk{m, m^\K, m}$.

\subsection{Algorithm Outline}
We first give a high-level outline of each step of the global stage algorithm. The constituent tensor stage algorithm will share similar high-level ideas. 

The algorithm takes in $\bigbk{\CW_q^{\otimes 2^{\lvl-1}}}^{\otimes n}$ as input and outputs level-$1$-independent level-$\lvl$ interface tensors as a degeneration of the input (for simplicity, we consider the $\eps=0$ case in this outline). In the algorithm, we define the notion of compatibility between level-$1$ blocks and level-$\lvl$ triples with respect to some specified complete split distributions, so that if all level-$1$ blocks in the remaining tensor are compatible with \emph{exactly one} level-$\lvl$ triple, then the subtensors over each remaining triple are level-$1$-independent. So the goal of the algorithm is to zero out some level-$\lvl$ and level-$1$ variable blocks such that each remaining level-$1$ block is compatible with a unique level-$\lvl$ triple. The structure of the algorithm is similar to the global stage algorithm in \cite{duan2023} with the main modification being the generalization from split distributions to complete split distributions.

On input $\bigbk{\CW_q^{\otimes 2^{\lvl-1}}}^{\otimes n}$, we first view the tensor as the tensor product of three terms, where each term is called a region, i.e., we write $\bigbk{\CW_q^{\otimes 2^{\lvl-1}}}^{\otimes n}$ as $\bigotimes_{r \in [3]} \bigbk{\CW_q^{\otimes 2^{\lvl-1}}}^{\otimes A_r \cdot n}$ for some $A_1, A_2, A_3\ge 0$  and $A_1+A_2+A_3 = 1$. Recall that we are only able to allow the sharing of level-$\lvl$ variable blocks in one of $X$-, $Y$-, $Z$-dimensions, so each region will allow the sharing of level-$\lvl$ variable blocks in different dimensions and we will perform the subsequent steps on the three regions separately. This step helps balance the number of remaining variable blocks in the three dimensions due to the asymmetric nature of the subsequent procedure.

From now on, we describe the procedure on the first region where we allow the sharing of level-$\lvl$ $Z$-variable blocks. We perform the same procedure up to rotation of the three dimensions on the other two regions separately.
\begin{enumerate}
\item \textbf{Zero out according to $\alpha$.} For a distribution $\alpha$ over the level-$\lvl$ constituent subtensors and its induced marginals $\alphx, \alphy, \alphz$ in the $X$-, $Y$-, $Z$-dimensions, we zero out level-$\lvl$ $X$-, $Y$-, $Z$-variable blocks that are not consistent with $\alphx, \alphy, \alphz$ respectively.
  
\item \textbf{Asymmetric hashing.} We use pairwise independent hash functions that hash level-$\lvl$ index sequences to the set $\{0,\dots, M-1\}$ for some $M$ which partitions the level-$\lvl$ variable blocks into buckets based on its hash value. Within each bucket, we do asymmetric cleanup so that every level-$\lvl$ $X$-variable block $X_I$ or $Y$-variable block $Y_J$ is contained in a unique level-$\lvl$ triple $X_I Y_J Z_K$, while a level-$\lvl$ $Z$-variable block $Z_K$ could be contained in multiple level-$\lvl$ triples.

\item \textbf{Compatibility zero-out I.} 
We define a notion of \emph{compatibility} with respect to the complete split distributions between level-$1$ blocks and level-$\ell$ triples for a set of specified level-$\lvl$ complete split distributions $\left\{\splres_{X, i, j, k}, \splres_{Y, i, j, k}, \splres_{Z, i, j, k}\right\}_{i+j+k = 2^\lvl}$ for the $X$-, $Y$-, $Z$-blocks. We zero out all the level-$1$ $X$- or $Y$-blocks that are not consistent with $\left\{\splres_{X, i, j, k}\right\}_{i + j + k = 2^\lvl}, \left\{\splres_{Y, i, j, k}\right\}_{i+j+k=2^\lvl}$ respectively (we can only do this because every level-$\lvl$ $X$-variable block $X_I$ or $Y$-variable block $Y_J$ is contained in a unique level-$\lvl$ triple). We zero out all the level-$1$ $Z$-blocks that are incompatible with any level-$\lvl$ triples.

\item \textbf{Compatibility zero-out II: unique triple.} After the compatibility zero-out I, every level-$1$ block is compatible with at least $1$ level-$\lvl$ triple and we want every level-$1$ block to be compatible with \emph{exactly one} level-$\lvl$ triple. So in this step, we zero out level-1 $Z$-blocks that are compatible with more than one level-$\lvl$ triples. Note that the level-$1$ blocks zeroed out in this step will become holes.  

\item \textbf{Usefulness zero-out.} Now that each remaining level-$1$ $Z$-block $Z_{\hat{K}}$ is contained in exactly one level-$\ell$ triple $X_IY_JZ_K$, we can define the notion of whether a level-$1$ block is \emph{useful} for the level-$\lvl$ triple containing it as whether it is consistent with the complete split distributions $\{\splres_{Z,i,j,k}\}_{i+j+k=2^\lvl}$. Note that we can only do this now because previously we do not have the property that every level-$1$ $Z$-block is in a unique level-$\lvl$ triple. In this step we zero out the level-$1$ blocks that are not useful for the level-$\ell$ triple containing it.

\item \textbf{Fixing holes.} Now we have obtained level-1-independent level-$\lvl$ interface tensors with holes. We use the following result which will be proved in \cref{sec:fixing_holes} to fix the holes.

  \begin{restatable}[Fixing holes in interface tensors]{cor}{fixinterface}\label{cor:fix-interface}
    Let $T$ be a level-$\ell$ interface tensor with parameter list
    \[\{(n_t, i_t, j_t, k_t, \splresXt, \splresYt, \splresZt)\}_{t \in [s]}.\]
    Let $N = 2^{\ell-1} \cdot \sum_{t\in [s]} n_t$. Suppose $T_1, \dots, T_r$ are broken copies of $T$ where $\le \frac{1}{8N}$ fraction of level-1 $X$-, $Y$- and $Z$-blocks are holes. If $r \ge 2^{C_1 N / \log N}$ for some large enough constant $C_1>0$, the direct sum $\bigoplus_{i = 1}^r T_i$ can degenerate into an unbroken copy of $T$.
  \end{restatable}

\end{enumerate}

\section{Global Stage}\label{sec:global}

In the global stage, we take as input the tensor $\CW_q^{\otimes N}$ for $N = n\cdot 2^{\lvl^*}$ and output independent copies of a level-$\lvl^*$ interface tensor, where the output will be a degeneration of the input. For the rest of this section, we will use $\lvl$ to denote $\lvl^*$ for convenience.

Given $\alpha$, which is a distribution over $\{(i, j, k) \in \mathbb{Z}_{\ge 0}^3 \mid i + j + k = 2^{\lvl}\}$, and $\splres_{X, i, j, k}, \splres_{Y, i, j, k}, \splres_{Z, i, j, k}$, which are level-$\lvl$ complete split distributions, we define the following quantities:
\begin{itemize}
\item $\alphx$ is the marginal distribution of $\alpha$ on the $X$-dimension, i.e., $\alphx(i) = \sum_{j, k} \alpha(i, j, k)$ for any $i$. We also similarly define $\alphy$ and $\alphz$.
\item $D$ is the set of distributions  whose marginal distributions on the three dimensions are $\alphx, \alphy, \alphz$ respectively, and let the penalty term $P_\alpha \defeq \max_{\alpha' \in D} H(\alpha') - H(\alpha) \ge 0$. 
\item For every $k$, 
  $\alpha(\+, \+, k) \defeq \sum_{i > 0, j > 0} \alpha(i, j, k)$; for every $j$, $\alpha(\+, j, \+) \defeq \sum_{i > 0, k > 0} \alpha(i, j, k)$; and for every $i$, 
  $\alpha(i, \+, \+) \defeq \sum_{j > 0, k > 0} \alpha(i, j, k)$. 
\item For every $k$, $\splresavg_{Z, \+, \+, k} \defeq \frac{1}{\alpha(\+, \+, k)} \sum_{i > 0, j > 0} \alpha(i, j, k) \cdot \splres_{Z, i, j, k}$, and $\splresavg_{Y, \+, j, \+}$ and $\splresavg_{X, i, \+, \+}$ are defined similarly. 
\item $\splavg{X} \defeq \sum_{i, j, k} \alpha(i, j, k) \cdot \splres_{X, i, j, k}$ and $\splavg{Y}$ and $\splavg{Z}$ are defined similarly.
\item $\lambda_Z \defeq \sum_{i, j, k: i = 0 \text{ or } j = 0} \alpha(i, j, k) \cdot H(\splres_{Z, i, j, k}) + \sum_k \alpha(\+, \+, k) \cdot H(\splresavg_{Z, \+, \+, k})$, and $\lambda_X$ and $\lambda_Y$ are defined similarly. 
\end{itemize}

In the following proposition, we will have $\alpha^{(r)}, \splres_{X, i, j, k}^{(r)}, \splres_{Y, i, j, k}^{(r)}, \splres_{Z, i, j, k}^{(r)}$ for every $r \in [3]$. For every $r \in [3]$, we use superscript $(r)$ on variables to denote that they are computed using values of $\alpha^{(r)}, \splres_{X, i, j, k}^{(r)}, \splres_{Y, i, j, k}^{(r)}, \splres_{Z, i, j, k}^{(r)}$.

\begin{prop}
\label{prop:global-stage-no-eps}
    $\bigbk{\CW_q^{\otimes 2^{\lvl-1}}}^{\otimes n}$ can be degenerated into
    \[ 2^{(A_1 E_1 + A_2 E_2 + A_3 E_3) n - o(n)} \]
    independent copies of a level-$\lvl$ interface tensor with parameter list 
    \[ \left\{\bk{ n \cdot A_r \cdot \alpha^{(r)}(i, j, k), i, j, k, \splres_{X, i, j, k}^{(r)}, \splres^{(r)}_{Y, i, j, k}, \splres^{(r)}_{Z, i, j, k} }\right\}_{r \in [3], \, i + j + k = 2^{\lvl}} \]
    where
    \begin{itemize}
        \item $0 \le A_1, A_2, A_3 \le 1, A_1 + A_2 + A_3 = 1$;
        \item $\alpha^{(r)}$ for every $r \in [3]$ is a  distribution over $\{(i, j, k) \in \mathbb{Z}_{\ge 0}^3 \mid i + j + k = 2^{\lvl}\}$;
        \item For every $W \in \{X, Y, Z\}$, $\splres_{W, i, j, k}^{(r)}$ for $r \in [3], i + j + k = 2^{\lvl - 1}$ is a level-$\lvl$ complete split distribution;
        \item \qquad\qquad $\displaystyle
        \begin{aligned}[t]
            E_1 &\defeq \min\left\{H(\alphx^{(1)}) - P_\alpha^{(1)}, H(\alphy^{(1)}) - P_\alpha^{(1)},  H(\splavg{Z}^{(1)}) - \lambda_Z^{(1)}\right\}, \\
            E_2 &\defeq \min\left\{H(\alphx^{(2)}) - P_\alpha^{(2)}, H(\alphz^{(2)}) - P_\alpha^{(2)},  H(\splavg{Y}^{(2)}) - \lambda_Y^{(2)}\right\}, \\
            E_3 &\defeq \min\left\{H(\alphy^{(3)}) - P_\alpha^{(3)}, H(\alphz^{(3)}) - P_\alpha^{(3)},  H(\splavg{X}^{(3)}) - \lambda_X^{(3)}\right\}.
        \end{aligned}$
    \end{itemize}
\end{prop}

\begin{remark}
\label{rmk:assumptions_on_complete_split_dist}
Note that without loss of generality, we can assume that, for every $r, i, j, k$, and every $L \in \{0, 1, 2\}^{2^{\lvl-1}}$, 
\[
\splres^{(r)}_{X, i, 0, k}(L) = \splres^{(r)}_{Z, i, 0, k}(\vec{2}-L), \quad   
\splres^{(r)}_{Z, 0, j, k}(L) = \splres^{(r)}_{Y, 0, j, k}(\vec{2}-L), \quad
\splres^{(r)}_{Y, i, j, 0}(L) = \splres^{(r)}_{X, i, j, 0}(\vec{2}-L), 
\]
where $\vec{2}$ denotes the length-$(2^{\lvl-1})$ vector whose coordinates are all $2$, 
and 
\[
\splres^{(r)}_{X, i, j, k}(L) = 0 \text{ if } \sum_{t} L_t \ne i, \quad   
\splres^{(r)}_{Y, i, j, k}(L) = 0 \text{ if } \sum_{t} L_t \ne j, \quad   
\splres^{(r)}_{Z, i, j, k}(L) = 0 \text{ if } \sum_{t} L_t \ne k,   
\]
because otherwise, the level-$\lvl$ interface tensor will be the zero tensor and the lemma will follow trivially. 
\end{remark}

Next, we show \cref{thm:global-stage-with-eps}, which is a corollary of \cref{prop:global-stage-no-eps}. 
\begin{theorem}
\label{thm:global-stage-with-eps}
For any $\eps > 0$, $2^{o(n)}$ independent copies of $(\CW_q^{\otimes 2^{\lvl-1}})^{\otimes n}$ can be degenerated into 
    $$2^{(A_1 E_1 + A_2 E_2 + A_3 E_3 - o_{1/\eps}(1)) n - o(n)}$$
    independent copies of a level-$\lvl$ $\eps$-interface tensor with parameter list 
    \[ \left\{ \bk{ n \cdot A_r \cdot \alpha^{(r)}(i, j, k), i, j, k, \splres_{X, i, j, k}^{(r)}, \splres^{(r)}_{Y, i, j, k}, \splres^{(r)}_{Z, i, j, k} } \right\}_{r \in [3], i + j + k = 2^{\lvl}} \]
    where the constraints are the same as those in \cref{prop:global-stage-no-eps}.\footnote{$\oeps(1)$ denotes a function $f(\eps)$ where $f(\eps) \to 0$ as $\eps \to 0$. We also use $\oeps(n)$ to denote $\oeps(1) \cdot n$.}
\end{theorem}

Here, the differences with \cref{prop:global-stage-no-eps} are the followings:
\begin{itemize}
    \item The input becomes multiple independent copies of $\bigbk{\CW_q^{\otimes 2^{\lvl-1}}}^{\otimes n}$. 
    \item The output tensor becomes independent copies of some level-$\lvl$ $\eps$-interface tensor, instead of level-$\lvl$ interface tensor in \cref{prop:global-stage-no-eps}. 
    \item There is a small $2^{o_{1/\eps}(n)}$ factor loss in the number of independent copies of the level-$\lvl$ $\eps$-interface tensor we can keep. 
\end{itemize}
The high-level idea of the proof is the following: for each copy of $\bigbk{\CW_q^{\otimes 2^{\lvl-1}}}^{\otimes n}$ in the input, we apply \cref{prop:global-stage-no-eps} where the target complete split distributions are slightly different in each application (up to $\eps$ in $L_\infty$ distance with some specified complete split distributions). Finally, we merge  the level-$\lvl$ interface tensors into a level-$\lvl$ $\eps$-interface tensor.

\begin{proof}[Proof of \cref{thm:global-stage-with-eps}]
  Let
  \[\left\{\xi_{W, i, j, k}^{(r)}\right\}_{r \in [3], W \in \{X, Y, Z\}, i + j + k = 2^\lvl}\]
  be a set of level-$\lvl$ complete split distributions whose $L_\infty$ distance with
  \[ \left\{\splres_{W, i, j, k}^{(r)}\right\}_{r \in [3], W \in \{X, Y, Z\}, i + j + k = 2^\lvl} \]
  is at most $\eps$. Furthermore, we require that all entries of $\xi_{W, i, j, k}^{(r)}$ are integral multiples of $\frac{1}{A_r \cdot \alpha(i, j, k) \cdot n}$. Let $\tilde{\mathcal{D}}$ be the collection of such sets of complete split distributions. For every $W, i, j, k$, there are $O(n)$ choices for the value of each entry in $\xi_{W, i, j, k}^{(r)}$, and the total number of entries is $3^{2^{\lvl - 1}} = O(1)$ as $\lvl$ is a constant. Thus, the number of $\xi_{W, i, j, k}^{(r)}$ is bounded by $\poly(n) = 2^{o(n)}$, 
    and consequently the number of $\left\{\splres_{W, i, j, k}^{(r)}\right\}_{r \in [3], W \in \{X, Y, Z\}, i + j + k = 2^\lvl}$ (i.e., the size of $\tilde{\mathcal{D}}$) is also bounded by $2^{o(n)}$.  Also, it is not difficult to verify that the level-$\lvl$ $\eps$-interface tensor with parameter list 
    \begin{equation}\label{eq:param-list-1}
        \left\{ \bk{ n \cdot A_r \cdot \alpha^{(r)}(i, j, k), i, j, k, \splres_{X, i, j, k}^{(r)}, \splres^{(r)}_{Y, i, j, k}, \splres^{(r)}_{Z, i, j, k} } \right\}_{r \in [3], i + j + k = 2^{\lvl}} 
    \end{equation}
    is the sum of all level-$\lvl$ interface tensors with parameter lists
    \begin{equation}\label{eq:param-list-2}
        \left\{ \bk{ n \cdot A_r \cdot \alpha^{(r)}(i, j, k), i, j, k, \xi_{X, i, j, k}^{(r)}, \xi_{Y, i, j, k}^{(r)}, \xi_{Z, i, j, k}^{(r)} } \right\}_{r \in [3], i + j + k = 2^{\lvl}}
    \end{equation}
    over all such $\left\{\xi_{W, i, j, k}^{(r)}\right\} \in \tilde{\mathcal{D}}$. 

    Let $E_1, E_2, E_3$ be defined as in \cref{prop:global-stage-no-eps} applied to complete split distributions $\left\{\splres_{W, i, j, k}^{(r)}\right\}$, and let $E_1', E_2', E_3'$ be defined as in \cref{prop:global-stage-no-eps} but applied to some complete split distributions $\left\{\xi_{W, i, j, k}^{(r)}\right\} \in \tilde{\mathcal{D}}$. By \cref{prop:global-stage-no-eps}, each copy of $\bigbk{\CW_q^{\otimes 2^{\lvl-1}}}^{\otimes n}$ can be degenerated into $2^{(A_1 E_1' + A_2 E_2' + A_3 E_3') n - o(n)}$ independent copies of the level-$\lvl$ interface tensor with parameter list as in~\eqref{eq:param-list-2}.
    
    It is not difficult to see that $A_1 E_1 + A_2 E_2 + A_3 E_3$ is continuous with respect to $\left\{\splres_{W, i, j, k}^{(r)}\right\}$, and because the $L_\infty$ distance between $\left\{\splres_{W, i, j, k}^{(r)}\right\}$ and $\left\{\xi_{W, i, j, k}^{(r)}\right\}$ is at most $\eps$, we get that $A_1 E_1' + A_2 E_2' + A_3 E_3' \ge A_1 E_1 + A_2 E_2 + A_3 E_3 - o_{1/\eps}(1)$.

    Thus, $2^{o(n)}$ independent copies of $\bigbk{\CW_q^{\otimes 2^{\lvl-1}}}^{\otimes n}$ can be degenerated into $2^{(A_1 E_1 + A_2 E_2 + A_3 E_3 - o_{1/\eps}(1)) n - o(n)}$ independent copies of a direct sum of all level-$\lvl$  interface tensor with parameter list
    \[ \left\{ \bk{ n \cdot A_r \cdot \alpha^{(r)}(i, j, k), i, j, k, \xi_{X, i, j, k}^{(r)}, \xi_{Y, i, j, k}^{(r)}, \xi_{Z, i, j, k}^{(r)} } \right\}_{r \in [3], i + j + k = 2^{\lvl}} \]
    over all such $\left\{\xi_{W, i, j, k}^{(r)}\right\} \in \tilde{\mathcal{D}}$, and because a direct sum of some tensors can be degenerated into the sum of these tensors, the theorem follows. 
\end{proof}

The remainder of this section aims to show and analyze an algorithm that proves \cref{prop:global-stage-no-eps}.

\subsection{Dividing into Regions}

Similar to \cite{duan2023}, we consider
\[\bigbk{\CW_q^{\otimes 2^{\lvl-1}}}^{\otimes n} \equiv \bigbk{\CW_q^{\otimes 2^{\lvl-1}}}^{\otimes A_1\cdot n}\otimes \bigbk{\CW_q^{\otimes 2^{\lvl-1}}}^{\otimes A_2\cdot n} \otimes \bigbk{\CW_q^{\otimes 2^{\lvl-1}}}^{\otimes A_3\cdot n}\]
for $A_1,A_2,A_3\ge 0$ and $A_1+A_2+A_3 = 1$. We call each of the three factors of the above tensor product a \emph{region}.
For $r\in [3]$, we denote the $r$-th region as
\[ \T^{(r)} \defeq \bigbk{\CW_q^{\otimes 2^{\lvl - 1}}}^{\otimes A_r \cdot n}. \]

The idea is to apply asymmetric hashing on the three regions separately. We will use asymmetric hashing that shares level-$\lvl$ $Z$-blocks in the first region, $Y$-blocks in the second region, and $X$-blocks in the third region. Each region will be degenerated into independent copies of a level-$\lvl$ interface tensor and the output will be the tensor product of the independent copies of the three level-$\lvl$ interface tensors from the three regions. Thus we can analyze each region independently and we only give the detailed analysis on the first region as the analysis for the other two regions follow by symmetry. 

From now on, we will describe the analysis on $\T^{(1)}$ in which the level-$\lvl$ $Z$-variable blocks are shared and we will omit the superscript $(1)$ on all variables for conciseness. 

\subsection{Asymmetric Hashing}
Recall that $\alpha$ is a distribution on $\{(i, j, k) \in \mathbb{Z}_{\ge 0}^3\mid i + j + k = 2^{\lvl}\}$, i.e., it can be viewed as a distribution on level-$\lvl$ constituent tensors. Recall that $\alpha$ induces marginal distributions $\alphx, \alphy, \alphz$. We first zero out $X$-, $Y$-, $Z$-blocks that are not consistent with the marginals $\alphx, \alphy, \alphz$ respectively. 
Let $\numxblock$ be the number of remaining level-$\lvl$ $X$-blocks, and it is not difficult to see that
\begin{equation}
\label{eq:global:NBX}
    \numxblock = 2^{H(\alphx) \cdot A_1 n \pm o(n)}.
\end{equation}
Similarly, let $\numyblock$ and $\numzblock$ be the number of remaining $Y$- and $Z$-blocks, and we have 
\begin{equation}
\label{eq:global:NBY-NBZ}
    \numyblock = 2^{H(\alphy) \cdot A_1 n \pm o(n)}, \quad \numzblock = 2^{H(\alphz) \cdot A_1 n \pm o(n)}. 
\end{equation}
Let $\numalpha$ be the number of remaining block triples that are consistent with $\alpha$. We have 
\begin{equation}
\label{eq:global:numalpha}
    \numalpha = 2^{H(\alpha) \cdot A_1 n \pm o(n)}.
\end{equation}
Finally, let $\numtriple$ be the number of remaining block triples $X_IY_JZ_K$.
\begin{claim}
  \label{cl:global:numtriple-bound}
  $\numtriple = 2^{(H(\alpha) + P_\alpha) \cdot A_1 n \pm o(n)}$.
\end{claim}
\begin{proof}
Recall that $P_\alpha = \max_{\alpha' \in D} H(\alpha') - H(\alpha)$ where $D$ is the set of distributions whose marginal distributions on the three dimensions are $\alphx, \alphy, \alphz$ respectively.

As we zeroed out $X$-, $Y$-, $Z$-blocks based on $\alphx, \alphy, \alphz$ respectively, all remaining block triples are consistent with one of the distributions $\alpha' \in D$. Additionally, $\alpha'(i, j, k) \cdot A_1 \cdot n$ must be an integer for every $i, j, k$. Let us denote the set of distributions satisfying such properties as $D'$. 

Thus, $\numtriple = \sum_{\alpha' \in D'} 2^{H(\alpha') \cdot A_1 n \pm o(n)}$. As $|D'| = \poly(n)$, we have that
\[ \numtriple = 2^{(\max_{\alpha' \in D'} H(\alpha')) \cdot A_1 n \pm o(n)}. \]
When $n$ approaches $\infty$, the difference between $\max_{\alpha' \in D'} H(\alpha')$ and $\max_{\alpha' \in D} H(\alpha')$ will approach $0$, as the entropy function $H$ is continuous. Thus, 
\[ \numtriple = 2^{(\max_{\alpha' \in D} H(\alpha')) \cdot A_1 n \pm o(n)} = 2^{(H(\alpha) + P_\alpha) \cdot A_1 n \pm o(n)}. \qedhere
\]
\end{proof}

Let $M \in [M_0, 2M_0]$ be a prime number for some integer $M_0$. The value of $M_0$ is yet to be fixed, but we first require that 
\begin{equation}
\label{eq:global:M_0_bounds_1}
M_0 \ge 8\cdot \max\left\{\frac{\numtriple}{\numxblock}, \frac{\numtriple}{\numyblock}\right\}.
\end{equation}
One additional term that lower bounds $M_0$ will be mentioned later. 

We independently pick uniformly random elements $b_0, \{w_t\}_{t=0}^n \in \{0, \ldots, M-1\}$, and define the following hash functions $h_X, h_Y, h_Z : \{0, \ldots, 2^\lvl\}^n \rightarrow \{0, \ldots, M - 1\}$:
\begin{align*}
    h_X(I) &= b_0 + \left(\sum_{t=1}^n w_t \cdot I_t \right) \bmod M,\\
    h_Y(J) &= b_0 + \left(w_0 + \sum_{t=1}^n w_t \cdot J_t \right) \bmod M,\\
    h_Z(K) &= b_0 + \frac{1}{2}\left(w_0+\sum_{t=1}^n w_t \cdot (2^\lvl - K_t) \right) \bmod M.
\end{align*}

Let $B$ be a Salem-Spencer subset of $\{0, \ldots, M - 1\}$ that has size $M^{1-o(1)}$ and  does not contain any nontrivial $3$-term arithmetic progressions (modulo $M$). Then we zero out all blocks $X_I$ with $h_X(I) \notin B$, $Y_J$ with $h_Y(J) \notin B$, and $Z_K$ with $h_Z(K) \notin B$.

For every block triple $X_I Y_J Z_K$ in $\T$, we have that $X_t + Y_t + Z_t = 2^\lvl$ for every $t \in [n]$. Therefore, it is not difficult to verify that $h_X(I) + h_Y(J) \equiv 2h_Z(K) \pmod{M}$. In order for $h_X(I), h_Y(J), h_Z(K) \in B$, we must have $h_X(I) = h_Y(J) = h_Z(K) = b$ for some $b$, because $B$ does not contain any nontrivial $3$-term arithmetic progression (modulo $M$). We say that triples $X_I Y_J Z_K$ with $h_X(I) = h_Y(J) = h_Z(K) = b$ are contained in bucket $b$. 

For every bucket $b$, if it contains two level-$\lvl$ triples $X_I Y_J Z_K$ and $X_I Y_{J'}Z_{K'}$ that share the same $X$-block, then we zero out $X_I$. Similarly, if a bucket contains two level-$\lvl$ triples $X_I Y_J Z_K$ and $X_{I'} Y_{J} Z_{K'}$ that share the same level-$\lvl$ $Y$-block, then we zero out $Y_J$. We repeatedly perform the previous zeroing-outs so that eventually, all remaining triples in the same bucket do not share $X$- or $Y$-blocks. 
As each level-$\lvl$ block triple in $\T$ must belong to some bucket, we get that all remaining triples do not share $X$- or $Y$-blocks, i.e., each level-$\lvl$ block $X_I$ or $Y_J$ is in a unique level-$\lvl$ block triple. For every level-$\lvl$ block  $X_I$ (or $Y_J$), we check whether the unique triple containing it is consistent with the distribution $\alpha$; if not, we zero out $X_I$ (or $Y_J$). We call the tensor after this step $\TPrime$. 

\begin{claim}[Implicit in \cite{cw90}, see also \cite{duan2023}]
\label{cl:global:hash-function-basic-properties}
For a block triple $X_I Y_J Z_K \in \T$, and for every $b \in \{0, \ldots, M - 1\}$,
\[\Pr\Bk{\tall h_X(I) = h_Y(J) = h_Z(K) = b} = \frac{1}{M^2}.\]
Furthermore, for two different block triples $X_I Y_J Z_K,  X_I Y_{J'} Z_{K'} \in \T$ that share the same $X$-block, and for every $b \in \{0, \ldots, M - 1\}$,
\[\Pr\Bk{\tall h_X(I) = h_Y(J') = h_Z(K') = b \;\middle\vert\; h_X(I) = h_Y(J) = h_Z(K) = b} = \frac{1}{M}. \]
This also holds analogously for different block triples that share the same $Y$-block or $Z$-block.
\end{claim}

\begin{claim}
\label{cl:global:prob-IJK-remain}
For every $b \in B$ and for every level-$\lvl$ block triple $X_I Y_J Z_K \in \T$ that is consistent with $\alpha$, the probability that $X_I Y_J Z_K$ remains in $\TPrime$ conditioned on $h_X(I) = h_Y(J) = h_Z(K) = b$ is $\ge \frac{3}{4}$. 
\end{claim}

\begin{proof}
The only way that $X_I Y_J Z_K$ does not remain in $\TPrime$ conditioned on $h_X(I) = h_Y(J) = h_Z(K) = b$ is when some other block triples that share the same $X$-block or the same $Y$-block are hashed to the same bucket $b$. 

Right before the hashing step, 
the total number of block triples remaining is $\numtriple$, and the number of $X$-blocks is $\numxblock$. By symmetry, each $X$-block is in the same number of block triples, which is $\frac{\numtriple}{\numxblock}$. Thus, the total number of block triples that share the same $X$-block as  $X_I Y_J Z_K$ is $\frac{\numtriple}{\numxblock} - 1$. For each of them, the probability that they are hashed to the same bucket $b$ with $X_I Y_J Z_K$ is $\frac{1}{M}$ by \cref{cl:global:hash-function-basic-properties}. Therefore, by union bound, the probability that any of them is hashed to the same bucket with $X_I Y_J Z_K$ is at most
\[
  \frac{\numtriple}{M \cdot \numxblock} \le \frac{\numtriple}{M_0 \cdot \numxblock} \;\overset{\cref{eq:global:M_0_bounds_1}}{\le}\; \frac{1}{8}.
\]

Similarly, the probability that any block triple that shares the same level-$\lvl$ $Y$-block is mapped to the same bucket as $X_I Y_J Z_K$ is at most $\frac{1}{8}$. By union bound, the probability that $X_I Y_J Z_K$ will be zeroed out is $\le \frac{1}{4}$. 
\end{proof}

\begin{claim}
  \label{cl:global:expected-IJK-remain}
  The expected number of level-$\lvl$ block triples in $\TPrime$ is at least 
  \[\numalpha \cdot M_0^{-1-o(1)}.\]
\end{claim}

\begin{proof}
For every level-$\lvl$ block triple $X_I Y_J Z_K \in \T$ that is consistent with $\alpha$, and for every $b \in B$, the probability that $h_X(I) = h_Y(J) = h_Z(K) = b$ is $\frac{1}{M^2}$ by \cref{cl:global:hash-function-basic-properties}. Also, by \cref{cl:global:prob-IJK-remain}, $X_I Y_J Z_K$ will remain in $\TPrime$ with probability $\ge \frac{3}{4} \cdot \frac{1}{M^2}$. 

Summing over all block triples $X_I Y_J Z_K$ and all $b \in B$, we get that the expected number of block triples in $\TPrime$ is at least
\[
\numalpha \cdot |B| \cdot \frac{3}{4} \cdot \frac{1}{M^2} = \numalpha \cdot M_0^{-1-o(1)}.
\qedhere
\]
\end{proof}

\subsection{Compatibility Zero-Out I}

Recall that $\left\{\splres_{X, i, j, k}, \splres_{Y, i, j, k}, \splres_{Z, i, j, k}\right\}_{i+j+k = 2^\lvl}$ are level-$\lvl$ complete split distributions for the $X$-, $Y$-, $Z$-blocks.

Let 
\[ S^{(I, J, K)}_{i, j, k} \defeq \{t \in [n]\mid I_t = i, J_t = j, K_t = k\}, \]
and 
\[S^{(K)}_{*, *, k} \defeq \{t \in [n] \mid K_t = k\}. \]
If clear from the context, we will drop the superscript $(I, J, K)$ or $(K)$. 

Recall that in $\TPrime$, every level-$\lvl$ block $X_I$ is in a unique block triple $X_I Y_J Z_K$. For every level-1 block $X_{\hat{I}} \in X_I$, we will zero out $X_{\hat{I}}$ if $\split(\hat I, S_{i,j,k}) \ne \splres_{X, i, j, k}$ for any $i, j, k$ (recall the definition of $\split(\hat I, S_{i,j,k})$ in \cref{def:split-hatI}). Similarly, every level-$\lvl$ block $Y_J$ is in a unique block triple, and we zero out every $Y_{\hat{J}} \in Y_J$ where $\split(\hat J, S_{i,j,k}) \ne \splres_{Y, i, j, k}$ for any $i, j, k$.

We can not perform the same zeroing out for $Z$-variables, because in $\TPrime$ each level-$\lvl$ $Z$-block is not in a unique block triple and $S_{i, j,k}$ is not well-defined just given the $Z$-block. Instead, for every level-1 block $Z_{\hat K} \in Z_K$, we zero out $Z_{\hat K}$ if $\split(\hat K, S_{*,*,k}) \ne \splresavg_{Z, *,*,k}$ for any $k$,  where
\[
\splresavg_{Z, *,*,k} = \frac{1}{\sum_{i+j = 2^\lvl - k} \alpha(i, j, k)}\sum_{i+j = 2^\lvl - k} \alpha(i, j, k) \cdot \splres_{Z,i,j,k}
\]
is the average complete split distribution for constituent tensors whose third coordinate is $k$.

We call the tensor after the previous zeroing-outs $\TDoublePrime$.

Next, we are ready to define the notion of compatibility. The notion is adapted from \cite{duan2023}, which is a crucial ingredient in their analysis (and ours). 

\begin{definition}[Compatibility]
  \label{def:global:compatibility}
  For some $I, J, K$, a level-1 block $Z_{\hat{K}} \in Z_K$ is compatible with a level-$\lvl$ triple $X_IY_J Z_K$ if 
  \begin{enumerate}
  \item 
    \label{item:global:compatibility1} For every $(i, j, k) \in \mathbb{Z}_{\ge 0}^3$ with $ i + j + k = 2^{\lvl}, i = 0 \text{ or } j = 0$,  $\split(\hat K, S_{i,j,k}) = \splres_{Z,i,j,k}$.
  \item
    \label{item:global:compatibility2}
    For every index $k \in \{0, 1, \ldots, 2^\lvl\}$,  $\split(\hat K, S_{*,*,k}) = \splresavg_{Z, *,*,k}$.
  \end{enumerate}
\end{definition}

\begin{claim}
  \label{cl:global:compatible}
  In $\TDoublePrime$, for every level-1 block triple $X_{\hat I}Y_{\hat J}Z_{\hat K}$ and the level-$\lvl$ block triple $X_I Y_J Z_K$ that contains it, $Z_{\hat K}$ is compatible with $X_I Y_J Z_K$.
\end{claim}
\begin{proof}
  First of all, \cref{item:global:compatibility2} is clearly satisfied, because we zeroed out every $\hat{K}$ with $\split(\hat K, S_{*,*,k}) \ne \splresavg_{Z, *,*,k}$ for any $k$. Next, we show that \cref{item:global:compatibility1} is also satisfied.

  Recall that we zeroed out all $X_{\hat{I}}$ where $\split(\hat I, S_{i,j,k}) \ne \splres_{X, i, j, k}$ for any $i, j, k$. Let $(i, j, k) \in \mathbb{Z}^3_{\ge 0}$ where $i+j+k = 2^\lvl$ and $j = 0$. 
  As $X_{\hat{I}} Y_{\hat{J}} Z_{\hat{K}}$ remains in $\TDoublePrime$, $\split(\hat I, S_{i, j, k}) = \splres_{X, i, j, k}$. Because $j = 0$, $J_t = 0$ for every $t \in S_{i, j, k}$, which implies that $\bigl(\hat{J}_{(t-1) \cdot 2^{\lvl - 1} + 1}, \, \hat{J}_{(t-1) \cdot 2^{\lvl - 1} + 2}, \, \ldots \, , \, \hat{J}_{t \cdot 2^{\lvl - 1}}\bigr) = \vec{0}$. As $\hat{I}_{\hat{t}} + \hat{J}_{\hat{t}} + \hat{K}_{\hat{t}} = 2$ for every $\hat{t}$, we have that 
  \[\bigbk{\hat{K}_{(t-1) \cdot 2^{\lvl - 1} + 1}, \, \hat{K}_{(t-1) \cdot 2^{\lvl - 1} + 2}, \, \ldots \, , \, \hat{K}_{t \cdot 2^{\lvl - 1}}} = \vec{2} - \bigbk{\hat{I}_{(t-1) \cdot 2^{\lvl - 1} + 1}, \, \hat{I}_{(t-1) \cdot 2^{\lvl - 1} + 2}, \, \ldots \, , \, \hat{I}_{t \cdot 2^{\lvl - 1}}}\]
  for every $t \in S_{i, j, k}$.
  Thus, for every $L \in \{0, 1, 2\}^{2^{\lvl - 1}}$, the proportion of $L$ appearing in $\bigl(\hat{I}_{(t-1) \cdot 2^{\lvl - 1} + 1},$ $\hat{I}_{(t-1) \cdot 2^{\lvl - 1} + 2}, \, \ldots \, , \, \hat{I}_{t \cdot 2^{\lvl - 1}}\bigr)$ over $t \in S_{i, j, k}$ is exactly the proportion of $\vec{2} - L$ appearing in $\bigl(\hat{K}_{(t-1) \cdot 2^{\lvl - 1} + 1},$ $\hat{K}_{(t-1) \cdot 2^{\lvl - 1} + 2}, \, \ldots \,,\, \hat{K}_{t \cdot 2^{\lvl - 1}}\bigr)$. In other words,
  $\split(\hat K, S_{i, j, k})(L) = \split(\hat I, S_{i, j, k})(\vec{2} - L) = \splres_{X, i, j, k}(\vec{2} - L)$. By \cref{rmk:assumptions_on_complete_split_dist}, this implies that $\split(\hat K, S_{i, j, k}) = \splres_{Z,i,j,k}$.

  We can show that $\split(\hat K, S_{i, j, k}) = \splres_{Z,i,j,k}$ with $i = 0$ similarly.
\end{proof}

\subsection{Compatibility Zero-Out II: Unique Triple}

In this step, we zero out level-1 $Z$-blocks that are compatible with more than one level-$\lvl$ triples. To do so, we check if each level-1 $Z$-block $Z_{\hat K}$ is compatible with multiple level-$\lvl$ triples. If so, we zero it out and it becomes a ``hole''. Note that after this step, each remaining level-1 $Z$-block $Z_{\hat{K}}\in Z_K$ is compatible with a unique level-$\lvl$ triple $(X_I, Y_J, Z_K)$ containing it.

\subsection{Usefulness Zero-Out}

Next, we further zero out some level-1 $Z$-blocks using the following definition of \emph{usefulness}.

\begin{definition}[Usefulness]
For a level-1 block $Z_{\hat K}$ and a level-$\lvl$ triple $X_I Y_J Z_K$ containing it, if for all $(i,j,k)$ we have $\split(\hat K, S_{i,j,k}) = \splres_{Z,i,j,k}$, then we say that $Z_{\hat K}$ is \emph{useful} for $X_I Y_J Z_K$.
\end{definition}

For each $Z_{\hat K}$, it appears in a unique triple $X_I Y_J Z_K$ by the previous zeroing out. Furthermore, if $Z_{\hat K}$ is not useful for this triple, we zero out $Z_{\hat K}$. We call the current tensor $\TTriplePrime$.

If there is no hole, then the subtensor of the remaining tensor over $X_I Y_J Z_K$ is isomorphic to 
\[ \mathcal{T}^* = \bigotimes_{i+j+k = 2^\lvl} T_{i,j,k}^{\otimes A_1 \cdot \alpha(i, j, k) \cdot n} [\splres_{X, i, j, k}, \splres_{Y, i, j, k}, \splres_{Z, i, j, k}], \]
i.e., it is the level-$\lvl$ interface tensor with parameter list 
\[ \left\{ \bk{ A_1 \cdot \alpha(i, j, k) \cdot n, i, j, k, \splres_{X, i, j, k}, \splres_{Y, i, j, k}, \splres_{Z, i, j, k} } \right\}_{i + j + k = 2^{\lvl}}. \]
More formally:

\begin{claim}
  \label{cl:global:after-unique-triple-zero-out}
  For any level-$\lvl$ block triple $X_I Y_J Z_K$ contained in $\TDoublePrime$ (or equivalently, $\TPrime$), the subtensor of $\TTriplePrime$ restricted to blocks $X_I, Y_J, Z_K$ is a subtensor of $\T^*$, where the missing variables in this subtensor are exactly those in level-1 blocks $Z_{\hat K}$ that are compatible with multiple level-$\lvl$ triples in $\TDoublePrime$.
\end{claim}

\begin{proof}
Initially,
$$\left.\TPrime\right\vert_{X_I Y_J Z_K} \equiv \bigotimes_{i+j+k = 2^\lvl} T_{i,j,k}^{\otimes A_1 \cdot \alpha(i, j, k) \cdot n}. $$

To show $\left.\TTriplePrime\right|_{X_I Y_J Z_K}$ is a subtensor of $\T^*$, it suffices to show that the  level-1 $X$-blocks ($Y$-blocks or $Z$-blocks resp.) remaining in $\left.\TTriplePrime\right|_{X_I Y_J Z_K}$ have the property that $\split(\hat I, S_{i,j,k}) = \splres_{X, i, j, k}$ ($\split(\hat J, S_{i,j,k}) = \splres_{Y, i, j, k}$ or $\split(\hat K, S_{i,j,k}) = \splres_{Z, i, j, k}$ resp.) for every $i, j, k$. This is true because we enforced these constraints on $X$- and $Y$-blocks in the compatibility zeroing-out step, and enforced the constraints on $Z$-blocks by zeroing out $Z_{\hat{K}}$ that is not useful for the unique level-$\lvl$ triple that contains it. Furthermore, these are the only constraints we have on the level-1 $X$- and $Y$-blocks, so the set of $X$- and $Y$-variables in $\left.\TTriplePrime\right|_{X_I Y_J Z_K}$ is the same as that in $\T^*$. It remains to analyze which level-1 $Z$-blocks are missing in $\left.\TTriplePrime\right|_{X_I Y_J Z_K}$.

There are three constraints we enforced on level-1 $Z$-blocks:
\begin{enumerate}
\item In the compatibility zeroing-out, we enforced that for every index $k \in \{0, 1, \ldots, 2^\lvl\}$,  $\split(\hat K, S_{*,*,k}) = \splresavg_{Z,*,*,k}$.
\item In the unique triple zeroing-out, we zeroed out $Z_{\hat{K}}$ that is compatible with multiple level-$\lvl$ triples.
\item In the unique triple zeroing-out, we zeroed out $Z_{\hat{K}}$ that is not useful for the unique level-$\lvl$ triple $X_I Y_J Z_K$ that contains it. Thus, we will have that $\split(\hat K, S_{i,j,k}) = \splres_{Z, i, j, k}$ for every $i, j, k$ if $Z_{\hat{K}}$ remains.
\end{enumerate}
The third constraint implies the first constraint, because if the third constraint holds, then for every $k$, 
\begin{align*}
\Split(\hat K, S_{*,*,k}) &= \frac{1}{\sum_{i, j} \alpha(i, j, k)} \sum_{i+j = 2^\lvl - k} \alpha(i, j, k) \cdot \Split(\hat K, S_{i, j, k}) \\
&= \frac{1}{\sum_{i, j} \alpha(i, j, k)}\sum_{i+j = 2^\lvl - k} \alpha(i, j, k) \cdot \splres_{Z, i, j, k} = \splresavg_{Z, *, *, k}.
\end{align*}
Therefore, we can ignore the first condition. As a result, the set of level-1 $Z$-blocks not in  $\left.\TTriplePrime\right|_{X_I Y_J Z_K}$ but in $\T^*$ are exactly those that are compatible with multiple block triples in $\TDoublePrime$.
\end{proof}

Also, note that for different remaining block triples $X_I Y_J Z_K$, $\TTriplePrime\vert_{X_I Y_J Z_K}$ are level-1-independent, i.e., they do not share the same level-1 blocks. This is because $X_I$ and $Y_J$ are already in unique level-$\lvl$ triples in $\TPrime$; for every level-1 block $Z_{\hat{K}}$, \cref{cl:global:compatible} shows that $Z_{\hat{K}}$ is compatible with every level-$\lvl$ triple $X_I Y_J Z_K$ containing it, and then we zeroed out $Z_{\hat{K}}$ that are compatible with multiple triples. Thus, every remaining $Z_{\hat{K}}$ in $\TTriplePrime$ is contained a unique level-$\lvl$ triple as well. As a result, we can write
\[
  \TTriplePrime = \bigoplus_{X_I Y_J Z_K \; \textup{remaining}} \TTriplePrime \vert_{X_I Y_J Z_K}
  \numberthis \label{eq:T'''_as_direct_sum}
\]
as a direct sum of broken copies of $\T^*$.

\subsection{Fixing Holes}
\label{sec:global:fix-holes}

Next, we analyze the fraction of holes in the broken copies of $\T^*$ contained in $\TTriplePrime$. To do so, we define the following notion of \emph{typicalness}, which will then be used to define the quantity $\pcomp$:

\begin{definition}[Typicalness]
  A level-1 $Z$-block $Z_{\hat{K}}$ in some level-$\lvl$ $Z$-block $Z_K$ is \emph{typical} if $\split(\hat K, S_{*,*,k}) = \splresavg_{Z,*,*,k}$ for every $k$. When $Z_K$ is consistent with $\alphz$, this condition can be equivalently written as $\split(\hat K, [A_1 n]) = \splavg{Z}$, where we recall that $\splavg{Z} = \sum_{i, j, k} \alpha(i, j, k) \cdot \splres_{Z, i, j, k}.$
\end{definition}

\begin{definition}[$\pcomp$]
  For fixed $Z_{\hat{K}}$ and $Z_K$ where $Z_{\hat{K}} \in Z_K$ and $Z_{\hat{K}}$ is typical, $\pcomp$ is the probability that a uniformly random block triple $X_I Y_J Z_K$ consistent with $\alpha$ is compatible with $\hat{K}$.
\end{definition}

By symmetry, this probability is the same for different $Z_{\hat{K}}$ and $Z_K$ where $Z_{\hat{K}} \in Z_K$ and $Z_{\hat{K}}$ is typical, so $\pcomp$ is well-defined. Since holes only arise when some $Z_{\hat K}$ is compatible with multiple triples, the value of $\pcomp$ is closely related to the fraction of holes, and is given by the following claim.

\begin{claim}
  The value of $\pcomp$ is
  $$2^{\left(\lambda_Z - H(\splavg{Z}) + H(\alphz)\right) A_1 \cdot n \pm o(n)}, $$
  where we recall that 
  \[
    \lambda_Z = \sum_{i, j, k: i = 0 \textup{ or } j = 0} \alpha(i, j, k) \cdot H(\splres_{Z, i, j, k}) + \sum_k \alpha(\+, \+, k) \cdot H(\splresavg_{Z, \+, \+, k}),
  \]
  \[ \alpha(\+, \+, k) = \sum_{i, j > 0} \alpha(i, j, k), \quad
    \splresavg_{Z, \+, \+, k} = \frac{1}{\alpha(\+, \+, k)} \sum_{i, j > 0} \alpha(i, j, k) \cdot \splres_{Z, i, j, k},
  \]
  and
  \[
    \splavg{Z} = \sum_{i, j, k} \alpha(i, j, k) \cdot \splres_{Z, i, j, k}.
  \]
\end{claim}

\begin{proof}
By symmetry, it suffices to compute the following two quantities, and $\pcomp$ will be the ratio between them: (1) the number of tuples $(I, J, K, \hat{K})$ where $X_I Y_J Z_K$ is consistent with $\alpha$, $\hat{K} \in K$, $Z_{\hat{K}}$ is typical, and $Z_{\hat{K}}$ is compatible with $X_I Y_J Z_K$; (2) the number of $(I, J, K, \hat{K})$ where $X_I Y_J Z_K$ is consistent with $\alpha$, $\hat{K} \in K$, and $Z_{\hat{K}}$ is typical.

We first compute the second quantity. First, the number of typical $Z_{\hat{K}}$ is $2^{H(\splavg{Z}) \cdot A_1 \cdot n \pm o(n)}$. Each of these $Z_{\hat{K}}$ uniquely determines a level-$\lvl$ block $Z_K$. Also, for  each $Z_K$, the number of block triples $X_I Y_J Z_K$ consistent with $\alpha$ is $\frac{\numalpha}{\numzblock} = 2^{(H(\alpha) - H(\alphz)) \cdot A_1 \cdot n \pm o(n)}$. Therefore, the second quantity is
\begin{equation}
  \label{eq:global:pcomp_denominator}
  2^{(H(\splavg{Z}) + H(\alpha) - H(\alphz)) \cdot A_1 \cdot n \pm o(n)}.
\end{equation}

Next, we compute the first quantity, which is the number of $(I, J, K, \hat{K})$ where $X_I Y_J Z_K$ is consistent with $\alpha$, $\hat{K} \in K$, $Z_{\hat{K}}$ is typical, and $Z_{\hat{K}}$ is compatible with $X_I Y_J Z_K$. By \cref{item:global:compatibility2} in \cref{def:global:compatibility}, if $Z_{\hat{K}}$ is compatible with any level-$\lvl$ block triple, then it is typical. Thus, we can drop the condition that $Z_{\hat{K}}$ is typical, and equivalently count the number of $(I, J, K, \hat{K})$ where $X_I Y_J Z_K$ is consistent with $\alpha$, $\hat{K} \in K$, and $Z_{\hat{K}}$ is compatible with $X_I Y_J Z_K$.

First, the number of block triples $X_I Y_J Z_K$ consistent with $\alpha$ is $\numalpha$. Then, for each such block triple, we count the number of $Z_{\hat{K}} \in Z_K$ that is compatible with it. If we fix some $X_I Y_J Z_K$, then we also have fixed the values of $S_{i, j, k}$ for all $i, j, k$. Then we can rewrite the condition for $Z_{\hat{K}}$ being compatible with $X_I Y_J Z_K$ equivalently as follows:

\begin{definition}[Compatibility']
  \label{def:global:compatibility'}
  For level-$\l$ triple $X_I Y_J Z_K$ consistent with $\alpha$, a level-1 block $Z_{\hat{K}} \in Z_K$ is compatible with $X_IY_J Z_K$ if
  \begin{itemize}
  \item For every $\{(i, j, k) \in \mathbb{Z}_{\ge 0}^3 \mid i + j + k = 2^{\lvl}, \, i = 0 \text{ or } j = 0\}$,  $\split(\hat K, S_{i,j,k}) = \splres_{Z,i,j,k}$. (This is exactly \cref{item:global:compatibility1} in \cref{def:global:compatibility}).
  \item For every $k$, let $S_{\+, \+, k} \defeq \bigcup_{i > 0, j > 0} S_{i, j, k}$. Then $\split(\hat K, S_{\+, \+, k}) = \splresavg_{Z, \+, \+, k}$.
  \end{itemize}
\end{definition}

\cref{item:global:compatibility1} and the second condition above imply the original condition $\split(\hat K, S_{*, *,k}) = \splresavg_{Z, *, *, k}$ in \cref{item:global:compatibility2}, because
\begin{align*}
  \Split(\hat K, S_{*, *,k})
  &= \frac{1}{\sum_{i, j \ge 0} \alpha(i, j, k)} \sum_{i, j \ge 0} \alpha(i, j, k) \cdot \Split(\hat K, S_{i, j, k})\\
  &= \frac{1}{\sum_{i, j \ge 0} \alpha(i, j, k)} \left(\sum_{i, j > 0} \alpha(i, j, k) \cdot \Split(\hat K, S_{i, j, k}) + \sum_{i = 0 \text{ or } j = 0} \alpha(i, j, k) \cdot \Split(\hat K, S_{i, j, k})\right)\\
  &= \frac{1}{\sum_{i, j \ge 0} \alpha(i, j, k)} \left(\sum_{i, j > 0} \alpha(i, j, k) \cdot \Split(\hat{K}, S_{\+, \+, k}) + \sum_{i = 0 \text{ or } j = 0} \alpha(i, j, k) \cdot \splres_{Z,i,j,k}\right)\\
  &= \frac{1}{\sum_{i, j \ge 0} \alpha(i, j, k)} \left(\sum_{i, j > 0} \alpha(i, j, k) \cdot \splresavg_{Z, \+, \+, k} + \sum_{i = 0 \text{ or } j = 0} \alpha(i, j, k) \cdot \splres_{Z,i,j,k}\right)\\
  &= \frac{1}{\sum_{i, j \ge 0} \alpha(i, j, k)}  \sum_{i, j \ge 0} \alpha(i, j, k) \cdot \splres_{Z,i,j,k} \,=\, \splresavg_{Z, *, *, k}.
\end{align*}
Similarly, \cref{item:global:compatibility1} and \cref{item:global:compatibility2} together imply the second condition in \cref{def:global:compatibility'}. Therefore, \cref{def:global:compatibility'} is an equivalent definition of compatibility.

In \cref{def:global:compatibility'}, there are constraints on the complete split distributions of $\hat{K}$ on some disjoint subsets of $[A_1 n]$. Therefore, we can count the number of valid subsequences of $\hat{K}$ for each of these subsets of indices, and multiply them together. For every $(i, j, k) \in \mathbb{Z}_{\ge 0}^3$ where $i + j + k = 2^{\lvl}$ while $i = 0$ or $j = 0$, we require that $\split(\hat{K}, S_{i, j, k}) = \splres_{Z, i, j, k}$, so the number of possibilities of $\hat{K}$ on the subset of indices $S_{i,j, k}$ is $2^{H(\splres_{Z, i, j, k}) \cdot |S_{i, j, k}| \pm o(n)} = 2^{H(\splres_{Z, i, j, k}) \cdot \alpha(i, j, k) \cdot A_1 n \pm o(n)}$. For every $k$, we require that $\split(\hat K, S_{\+, \+,k}) = \splresavg_{Z, \+, \+, k}$, so the number of possibilities of $\hat K$ on $S_{\+, \+,k}$ is $2^{H(\splresavg_{Z, \+, \+, k}) \cdot |S_{\+, \+, k}| \pm o(n)} = 2^{H(\splresavg_{Z, \+, \+, k}) \cdot \alpha(\+, \+, k) \cdot A_1 n \pm o(n)}$. Overall, the number of possible compatible $\hat{K}$, multiplied by the number of block triples $X_I Y_J Z_K$, is
\begin{equation}
  \label{eq:global:pcomp_numerator}
  \numalpha \cdot \prod_{\substack{i, j, k \\ i = 0 \textup{ or } j = 0}} 2^{H(\splresavg_{Z,i, j, k}) \cdot \alpha(i, j, k) \cdot A_1 n \pm o(n)} \cdot \prod_{k} 2^{H(\splresavg_{Z, \+, \+, k}) \cdot \alpha(\+, \+, k) \cdot A_1 n \pm o(n)} = 2^{(H(\alpha) + \lambda_Z) \cdot A_1 n \pm o(n)}.
\end{equation}

Finally, as mentioned, $\pcomp$ is the ratio between \eqref{eq:global:pcomp_numerator} and \eqref{eq:global:pcomp_denominator}, so
\[ \pcomp = \left(2^{(H(\alpha) + \lambda_Z) \cdot A_1 n \pm o(n)}\right) \Big/ \left(2^{(H(\splavg{Z}) + H(\alpha) - H(\alphz)) \cdot A_1 \cdot n \pm o(n)}\right) = 2^{\left(\lambda_Z - H(\splavg{Z}) + H(\alphz)\right) A_1 \cdot n \pm o(n)} \]
as desired.
\end{proof}

\begin{claim}
  \label{cl:global:prob-of-holes}
  For every $b \in B$, every level-$\lvl$ block triple $X_I Y_J Z_K$ consistent with $\alpha$, and for each typical $Z_{\hat{K}} \in Z_K$, the probability that $Z_{\hat{K}}$ is compatible with multiple triples in $\TDoublePrime$ is at most
  \[ \frac{\numalpha \cdot \pcomp}{\numzblock \cdot M_0}, \]
  conditioned on $h_X(I) = h_Y(J) = h_Z(K) = b$.
\end{claim}

\begin{proof}
By the definition of $\pcomp$, the total number of level-$\lvl$ block triples $X_{I'} Y_{J'} Z_{K}$ that is compatible with $Z_{\hat{K}}$ is $\frac{\numalpha}{\numzblock} \cdot \pcomp$. For each $X_{I'} Y_{J'} Z_{K}$ different from  $X_{I} Y_{J} Z_{K}$, the probability that $X_{I'} Y_{J'} Z_{K}$ is mapped to the same bucket $b$ as $X_{I} Y_{J} Z_{K}$ is $\frac{1}{M}$ by \cref{cl:global:hash-function-basic-properties}. Thus, by the union bound, the probability that any of them is mapped to the same bucket as $X_{I} Y_{J} Z_{K}$ is upper bounded by $\frac{\numalpha}{\numzblock} \cdot \pcomp \cdot \frac{1}{M} \le \frac{\numalpha \cdot \pcomp}{\numzblock \cdot M_0}$. Furthermore, if none of them are mapped to the same bucket as $X_{I} Y_{J} Z_{K}$, then $Z_{\hat{K}}$ is compatible with a unique triple $X_{I} Y_{J} Z_{K}$ in $\TDoublePrime$, so the claim follows.
\end{proof}

Recall that we require $M_0$ to be at least $8\cdot \max\bigBK{\frac{\numtriple}{\numxblock}, \frac{\numtriple}{\numyblock}}$. Now, we add another (and final) constraint: $M_0 \ge \frac{\numalpha \cdot \pcomp}{\numzblock} \cdot 80N$.  That is, we will set $M_0$ to be
\begin{align*}
& \max\left\{\frac{8\numtriple}{\numxblock}, \frac{8\numtriple}{\numyblock},  \frac{\numalpha \cdot \pcomp}{\numzblock} \cdot 80N\right\} \\
&= 2^{\max\{H(\alpha) - P_{\alpha} - H(\alphx), \; H(\alpha) - P_{\alpha} - H(\alphy), \; H(\alpha) + \lambda_Z - H(\splavg{Z})\} \cdot A_1 \cdot n \pm o(n)}.
\end{align*}

Now, for every $b \in B$ and every level-$\lvl$ block triple $X_I Y_J Z_K$ that is consistent with $\alpha$ with $h_X(I) = h_Y(J) = h_Z(K) = b$,
\begin{enumerate}
\item by \cref{cl:global:prob-IJK-remain}, it remains in $\TPrime$ with probability $\ge \frac{3}{4}$;
\item by \cref{cl:global:prob-of-holes},  linearity of expectation and Markov's inequality, among  $Z_{\hat{K}} \in Z_K$ that is useful for $X_I Y_J Z_K$ (this implies that $Z_{\hat{K}}$ is typical, so we could apply \cref{cl:global:prob-of-holes}), the fraction of $Z_{\hat{K}}$ that becomes a hole in $\TTriplePrime$ is at most $10/80N = \frac{1}{8N}$ with probability at least $9/10$.
\end{enumerate}

Therefore, by the union bound, with constant probability, the subtensor of $\TTriplePrime$ over $X_I, Y_J, Z_K$ is a copy of $\T^*$ whose fraction of holes does not exceed $1/8N$. The expected number of $X_I Y_J Z_K$ with $h_X(I) = h_Y(J) = h_Z(K) = b$ over all $b \in B$ is $\numalpha \cdot M^{-1-o(1)}$, so overall, $\TTriplePrime$ contains $\numalpha \cdot M^{-1-o(1)}$ copies of $\T^*$ whose fraction of holes is $1/8N$.

By \cref{cor:fix-interface}, we can degenerate them into $\numalpha \cdot M^{-1-o(1)}$ unbroken copies of $\T^*$.

\subsection{Summary}
So far, we have degenerated $\bigbk{\CW_q^{\otimes 2^{\lvl - 1}}}^{\otimes A_1 \cdot n}$ into $\ge \numalpha \cdot M_0^{-1-o(1)}$ copies of a level-$\lvl$ interface tensor $\T^*$ with parameter list
\[ \left\{ \bk{ n \cdot A_1 \cdot \alpha^{(1)}(i, j, k), i, j, k, \splres^{(1)}_{X, i, j, k}, \splres^{(1)}_{Y, i, j, k}, \splres^{(1)}_{Z, i, j, k} } \right\}_{i + j + k = 2^{\lvl}}. \]
By plugging in the bounds of $\numalpha$ and $M_0$, we see that the number of copies we obtained (in the first region) is 
\[
  2^{A_1 n \cdot \min\BK{H(\alphx^{(1)}) - P_\alpha^{(1)}, \, H(\alphy^{(1)}) - P_\alpha^{(1)}, \,  H(\splavg{Z}^{(1)}) - \lambda_Z^{(1)}} - o(n)}.
\]

By symmetry, we can apply the same method to the second and third region, where for the second region we perform asymmetric hashing that shares $Y$-variable blocks, and for the third region we perform asymmetric hashing that shares $X$-blocks. Taking the tensor product of these results returned by our method on the three regions concludes the proof.

\section{Constituent Stage}
\label{sec:constituent-tensor}
\label{sec:constituent}

In the constituent stage for level-$\lvl$ for some $\lvl > 1$, the input is an $s$-term level-$\lvl$ $\eps$-interface tensor with parameters
\[\{(n_t, i_t, j_t, k_t, \splresXt, \splresYt, \splresZt)\}_{t \in [s]}\] that meet the following constraints:
\begin{enumerate}
\item For every $t \in [s]$, if $\hat{i}_1 + \hat{i}_2 + \cdots + \hat{i}_{2^{\lvl-1}} \ne i_t$, then $\splresXt(\hat{i}_1, \hat{i}_2, \ldots, \hat{i}_{2^{\lvl-1}}) = 0$. Similar constraints hold for $\splresYt$ and  $\splresZt$.
\item For every $t \in [s]$ with $j_t = 0$, and every $\hat{i}_1, \hat{i}_2, \ldots, \hat{i}_{2^{\lvl-1}}$, 
  $$
  \splresXt(\hat{i}_1, \hat{i}_2, \ldots, \hat{i}_{2^{\lvl-1}}) = \splresZt(2 - \hat{i}_1, 2 - \hat{i}_2, \ldots, 2 - \hat{i}_{2^{\lvl-1}}).
  $$
  Similar relations hold between $\splresXt$ and $\splresYt$ where $k_t = 0$ and between $\splresYt$ and $\splresZt$ where $i_t = 0$. 
\end{enumerate}
Additionally, we let $n =  \sum_t n_t$ and $N = 2^{\lvl - 1} \cdot n$. 
The goal of this stage is to degenerate the input to the tensor product between a matrix multiplication tensor and multiple independent copies of a level-$(\lvl-1)$ $\eps'$-interface tensor for some $\eps' > 0$.

Before we apply the laser method, let us handle the terms $t \in [s]$ in the level-$\lvl$ $\eps$-interface tensor where $i_t = 0$, $j_t = 0$ or $k_t = 0$, which are already matrix multiplication tensors. The proof idea of the following theorem is similar to the proof idea of a result in \cite{virgi12}, who showed the version of the following theorem without complete split distributions. 

\begin{theorem}
  \label{thm:consituent_MM_terms}
  If $k_t=0$, then
  \[ T_{i_t,j_t,k_t}^{\otimes n_t}[\splresXt, \splresYt, \splresZt, \eps] \equiv \angbk{1, M, 1}, \]
  where
  \[ M = 2^{n_t (H(\splresXt) \pm o_{1/\eps}(1)) \pm o(n)} \cdot q^{n_t\sum_{(\hat{i}_1, \hat{i}_2, \ldots, \hat{i}_{2^{\lvl-1}})} \splresXt(\hat{i}_1, \hat{i}_2, \ldots, \hat{i}_{2^{\lvl-1}}) \sum_{p=1}^{2^{\lvl-1}} [\hat{i}_p = 1]}. \]
  Similar results hold when $i_t = 0$ or $j_t = 0$.
\end{theorem}
\begin{proof}
  As $k_t = 0$, there is only one $Z$-variable $z_0$ in the given tensor. Also, for each fixed $X$-variable $x$, there is a unique $Y$-variable $y$ so that $xyz_0$ is a term in the given tensor (this is because it is a subtensor of $\CW_q^{\otimes N}$), and vice versa. Thus, the given tensor is isomorphic to  an inner product tensor $\angbk{1, M, 1}$ for some $M \ge 0$. It remains to calculate the number of $X$-variables in the given tensor. The $X$-variables are distributed among several level-1 $X$-blocks. Fixing a complete split distribution $\xiXt$ whose $L_{\infty}$ distance to $\splresXt$ is within $\eps$, the number of level-1 $X$-blocks in $T_{i_t, j_t, k_t}^{\otimes n_t}$ that conform with $\xiXt$ is
  \[
    2^{n_t H(\xiXt) \pm o(n)} \;=\; 2^{n_t (H(\splresXt) \pm \oeps(1)) \pm o(n)}.
    \numberthis \label{eq:num_of_level_1_x_block}
  \]
  In each of these level-1 blocks, say $X_{\hat I}$, the number of $X$-variables is
  \begin{align*}
    q^{\sum_{p=1}^{n_t \cdot 2^{\lvl - 1}} [\hat I_p = 1]} \;&= \; q^{n_t \sum_{(\hat i_1, \ldots, \hat i_{2^{\lvl - 1}})} \xiXt(\hat i_1, \ldots, \hat i_{2^{\lvl - 1}}) \sum_{p=1}^{2^{\lvl-1}} [\hat i_p = 1]} \\
                                                     &= \; q^{n_t \sum_{(\hat i_1, \ldots, \hat i_{2^{\lvl - 1}})} (\splresXt(\hat i_1, \ldots, \hat i_{2^{\lvl - 1}}) \pm \oeps(1)) \sum_{p=1}^{2^{\lvl-1}} [\hat i_p = 1]}.
                                                       \numberthis \label{eq:num_x_var_in_each_block}
  \end{align*}
  The product of \eqref{eq:num_of_level_1_x_block} and \eqref{eq:num_x_var_in_each_block} gives the number of $X$-variables belonging to level-1 $X$-blocks that are consistent with a certain $\xiXt$; taking summation over all $\xiXt$ (there are $\poly(n)$ of which) proves the lemma.
\end{proof}

Next, we assume that we already used \cref{thm:consituent_MM_terms} to handle terms with $i_t = 0$, $j_t = 0$ or $k_t = 0$, and assume without loss of generality that we are left with the first $s'$ terms for some $s' \le s$.

For a triple of level-$\lvl$ complete split distributions $(\splresX, \splresY, \splresZ)$ associated with the tensor power of the constituent tensor $T_{i_t,j_t,k_t}$, we define a distribution $\splonelevelX$ on $\{0, \ldots, 2^{\lvl-1}\}^2$ as follows:
\[
  \splonelevelX(l_X, r_X) \defeq \sum_{\substack{(\hat{i}_1, \hat{i}_2, \ldots, \hat{i}_{2^{\lvl - 1}}):\\  \hat{i}_1 + \cdots + \hat{i}_{2^{\lvl - 2}} = l_X, \\ \hat{i}_{2^{\lvl-2}+1} + \cdots + \hat{i}_{2^{\lvl - 1}} = r_X}} \splresX(\hat{i}_1, \hat{i}_2, \ldots, \hat{i}_{2^{\lvl-1}}).
\]
It describes how every level-$\lvl$ index $i_t$ splits into two level-$(\lvl - 1)$ indices. We similarly define $\splonelevelY$ and $\splonelevelZ$.

Let $\alpha$ be a distribution on possible combinations of $(l_X, l_Y, l_Z)$ such that the marginals of $\alpha$ are consistent with $\splonelevelX(l_X, i-l_X)$, $\splonelevelY(l_Y, j-l_Y)$, $\splonelevelZ(l_Z, k-l_Z)$. Moreover, let $\splresXt[i', j', k']$, $\splresYt[i', j', k']$, $\splresZt[i', j', k']$ be level-$(\lvl-1)$ complete split distributions. We then define the following quantities:
\begin{itemize}
\item $D$ is the set of distributions whose marginal distributions on the three dimensions are consistent with $\splonelevelX(l_X, i-l_X)$, $\splonelevelY(l_Y, j-l_Y)$, $\splonelevelZ(l_Z, k-l_Z)$ respectively, and let the penalty term $P_\alpha \defeq \max_{\alpha' \in D} H(\alpha') - H(\alpha) \ge 0$.
\item For every $k'$,
  $\alpha(\+, \+, k') \defeq \sum_{i' > 0, j' > 0} \alpha(i', j', k')$; for every $j'$, $\alpha(\+, j', \+) \defeq \sum_{i' > 0, k' > 0} \alpha(i', j', k')$; and for every $i'$,
  $\alpha(i', \+, \+) \defeq \sum_{j' > 0, k' > 0} \alpha(i', j', k')$.
\item For every $k'$,
  $\alpha(\<, \<, k') \defeq \sum_{i' < i_t , j' < j_t} \alpha(i', j', k')$; for every $j$, $\alpha(\<, j', \<) \defeq \sum_{i' < i_t, k' < k_t} \alpha(i', j', k')$; and for every $i'$,
  $\alpha(i', \<, \<) \defeq \sum_{j' < j_t, k' < k_t} \alpha(i', j', k')$.
\item For every $k'$, $\splresavg_{Z, \+, \+, k'} \defeq \frac{1}{\alpha(\+, \+, k')} \sum_{i' > 0, j' > 0} \alpha(i', j', k') \cdot \splres_{Z, i', j', k'}$, while $\splresavg_{Y, \+, j', \+}$ and $\splresavg_{X, i', \+, \+}$ are defined similarly.
\item
  $ \displaystyle \lambda_Z \defeq \sum_{i', j', k': i' = 0 \text{ or } j' = 0} \bigbk{\alpha(i', j', k')+\alpha(i_t-i', j_t-j', k_t-k')} \cdot H(\splres_{Z, i', j', k'}) \\
  + \sum_{k'} \bigbk{\alpha(\+, \+, k') + \alpha(\<, \<, k_t-k')} \cdot H(\splresavg_{Z, \+, \+, k_t-k'}), $
  while $\lambda_X$ and $\lambda_Y$ are defined similarly.
\end{itemize}

In the following proposition, we will use the above definitions for different $t \in [s']$ and $r \in [3]$. We will use $t$ in the subscripts and $(r)$ in the superscripts on variables to denote that they are computed using values of $\alpha_t^{(r)}, \splres_{X, t}^{(r)}, \splres_{Y, t}^{(r)}, \splres_{Z, t}^{(r)}, \{\splresXt[t, i', j', k']^{(r)}\}_{i',j',k'}, \{\splresYt[t, i', j', k']^{(r)}\}_{i',j',k'}, \{\splresZt[t, i', j', k']^{(r)}\}_{i',j',k'}$. 

\begin{prop}
  \label{prop:constituent-stage-no-eps}
  An $s'$-term level-$\lvl$ $\eps$-interface tensor with parameters
  \[\{(n_t, i_t, j_t, k_t, \splresXt, \splresYt, \splresZt)\}_{t \in [s']}\]
  for $\eps > 0, \; i_t, j_t, k_t > 0 \; \forall \; t \in [s']$ can be degenerated into
  \[ 2^{(E_1 + E_2 + E_3) - o(n) - \oeps(n)} \]
  independent copies of a level-$(\lvl-1)$ interface tensor with parameter list
  \[ \left\{ \bk{ n_t \cdot A_{t,r} \cdot \bigbk{\alpha_{t}^{(r)}(i', j', k')+\alpha_{t}^{(r)}(i_t-i', j_t-j', k_t-k')}, i', j', k', \splresXt[t, i', j', k']^{(r)}, \splresYt[t, i', j', k']^{(r)}, \splresZt[t, i', j', k']^{(r)} }\right\}\]
  for $t \in [s']$, $r \in [3]$, $i' + j' + k' = 2^{\lvl-1}$, $0 \le i' \le i_t$, $0 \le j' \le j_t$, $0 \le k' \le k_t$, where
  \begin{itemize}
  \item $0 \le A_{t,1}, A_{t,2}, A_{t,3} \le 1$ and $A_{t,1} + A_{t,2} + A_{t,3} = 1$ for every $t \in [s']$;
  \item For every $t$, and for every $W \in \{X, Y, Z\}$, $A_{t,1} \splres_{W, t}^{(1)} + A_{t,2} \splres_{W, t}^{(2)}+ A_{t,3} \splres_{W, t}^{(3)} = \splres_{W, t}$ ($\splres_{W, t}^{(r)}$ are intermediate variables that will be used later); 
  \item For every $W \in \{X, Y, Z\}$,  $r \in [3]$ and $i' + j' + k' = 2^{\lvl - 1}$, $\splres_{W, t, i', j', k'}^{(r)}$ is a level-$(\lvl-1)$ complete split distribution;
  \item For every $W \in \{X, Y, Z\}$, $t \in [s']$ and $r \in [3]$,
    \[ \splres_{W, t}^{(r)} = \sum_{i', j', k'} \alpha_{t}^{(r)}(i', j', k') \cdot \left(\splres_{W, t, i', j', k'}^{(r)} \times \splres_{W, t,  i_t - i', j_t - j', k_t - k'}^{(r)}\right); \]
  \item $\displaystyle
    \begin{aligned}[t]
      E_1 \defeq \min\Bigg\{
      \sum_{t \in [s']} A_{t,1} \cdot n_t \cdot \left( H(\splonelevelXt[t]^{(1)}) - P_{\alpha, t}^{(1)}\right), \;
      & \sum_{t \in [s']} A_{t,1} \cdot n_t \cdot \left( H(\splonelevelYt[t]^{(1)}) - P_{\alpha, t}^{(1)}\right), \\
      & \sum_{t \in [s']} A_{t,1} \cdot n_t \cdot \left(H(\splres_{Z, t}^{(1)}) - \lambda_{Z, t}^{(1)}\right) \Bigg\}, \\
      E_2 \defeq \min\Bigg\{
      \sum_{t \in [s']} A_{t,2} \cdot n_t \cdot \left( H(\splonelevelXt[t]^{(2)}) - P_{\alpha, t}^{(2)}\right), \;
      & \sum_{t \in [s']} A_{t,2} \cdot n_t \cdot \left( H(\splonelevelZt[t]^{(2)}) - P_{\alpha, t}^{(2)}\right), \\
      & \sum_{t \in [s']} A_{t,2} \cdot n_t \cdot \left(H(\splres_{Y, t}^{(2)}) - \lambda_{Y, t}^{(2)}\right) \Bigg\}, \\
      E_3 \defeq \min\Bigg\{
      \sum_{t \in [s']} A_{t,3} \cdot n_t \cdot \left( H(\splonelevelYt[t]^{(3)}) - P_{\alpha, t}^{(3)}\right), \;
      & \sum_{t \in [s']} A_{t,3} \cdot n_t \cdot \left( H(\splonelevelZt[t]^{(3)}) - P_{\alpha, t}^{(3)}\right), \\
      & \sum_{t \in [s']} A_{t,3} \cdot n_t \cdot \left(H(\splres_{X, t}^{(3)}) - \lambda_{X, t}^{(3)}\right) \Bigg\}.
    \end{aligned}
    $
  \end{itemize}
\end{prop}

Given \cref{prop:constituent-stage-no-eps}, we obtain the following theorem, whose proof is essentially the same as that of \cref{thm:global-stage-with-eps}.

\begin{theorem}
  \label{thm:constituent-stage-with-eps}
  $2^{o(n)}$ independent copies of $s'$-term level-$\lvl$ $3\eps$-interface tensor with parameters
  \[\{(n_t, i_t, j_t, k_t, \splresXt, \splresYt, \splresZt)\}_{t \in [s']}\]
  for $\eps > 0, i_t, j_t, k_t > 0\ \forall\ t \in [s']$ can be degenerated into
  \[2^{(E_1 + E_2 + E_3) - o(n) - \oeps(n)}\]
  independent copies of a level-$(\lvl-1)$ $\eps$-interface tensor with parameter list
  \[ \left\{ \bk{ n_t \cdot A_{t,r} \cdot \bigbk{\alpha_{t}^{(r)}(i', j', k')+\alpha_{t}^{(r)}(i_t-i', j_t-j', k_t-k')}, i', j', k', \splresXt[t, i', j', k']^{(r)}, \splresYt[t, i', j', k']^{(r)}, \splresZt[t, i', j', k']^{(r)} } \right\}\]
  for $t \in [s']$, $r \in [3]$, $i' + j' + k' = 2^{\lvl-1}$, $0 \le i' \le i_t$, $0 \le j' \le j_t$, $0 \le k' \le k_t$, where the constraints are the same as those in \cref{prop:constituent-stage-no-eps}.
\end{theorem}

\begin{proof}
  Similar to \cref{thm:global-stage-with-eps}, for every set of complete split distributions $\midBK{\xiWt[t, i', j', k']^{(r)}}_{W, t, r, i', j', k'}$ that is at most $\eps$ away in $L_\infty$ distance from $\midBK{\splresWt[t, i', j', k']^{(r)}}_{W, t, r, i', j', k'}$ , we take an independent copy of the input interface tensor, and degenerate it to independent copies of the output interface tensor with the specified complete split distributions. Let
  \[
    \xi_{W, t}^{(r)} = \sum_{i', j', k'} \alpha_{t}^{(r)}(i', j', k') \cdot \left(\xi_{W, t, i', j', k'}^{(r)} \times \xi_{W, t,  i_t - i', j_t - j', k_t - k'}^{(r)}\right) \quad (\forall \, W \in \{X, Y, Z\}, \, r \in [3], \, t \in [s'])
    \numberthis \label{eq:xi_computed_from_target}
  \]
  and $\xi_{W, t} = A_{t, 1} \xi_{W, t}^{(1)} + A_{t, 2} \xi_{W, t}^{(2)} + A_{t, 3} \xi_{W, t}^{(3)}$ be determined by the considered complete split distributions $\midBK{\xi_{W, t, i', j', k'}}$. According to \cref{prop:constituent-stage-no-eps}, an $\eps$-interface tensor $\T$ with parameter list $\{(n_t, i_t, j_t, k_t, \xiXt, \xiYt, \xiZt)\}_{t \in [s']}$ can degenerate to $2^{E_1 + E_2 + E_3 - o(n) - \oeps(n)}$ copies of the target interface tensor. Summing up a copy of the outcome tensor for each $\midBK{\xiWt[t, i', j', k']^{(r)}}_{W, t, r, i', j', k'}$ will give the output $\eps$-interface tensor, so we can get $2^{E_1 + E_2 + E_3 - o(n) - \oeps(n)}$ independent copies of the output tensor in total.

  It remains to show that $\T$ is a subtensor of the input interface tensor, i.e., a $3\eps$-interface tensor with parameters $\midBK{(n_t, i_t, j_t, k_t, \splresXt, \splresYt, \splresZt)}_{t \in [s']}$. On the right-hand side of \eqref{eq:xi_computed_from_target}, the two complete split distributions have at most $\eps$ distance from $\splres_{W, t, i', j', k'}^{(r)}$ and $\splres_{W, t, i_t - i', j_t - j', k_t - k'}^{(r)}$, so their product has $\le 2\eps$ 
  distance\footnote{The distance is at most $2\eps$ for the following reason: first, we change $\splres_{W, t, i', j', k'}^{(r)}$ to $\xi_{W, t, i', j', k'}^{(r)}$, which introduces an additive $\eps$ error (as the right hand side in \cref{eq:xi_computed_from_target} is a weighted average of the entries of $\xi_{W, t, i', j', k'}^{(r)}$); then we change $\splres_{W, t, i_t - i', j_t - j', k_t - k'}^{(r)}$ to $\xi_{W, t, i_t - i', j_t - j', k_t - k'}^{(r)}$, which introduces another additive $\eps$ error. } from $\splres_{W, t, i', j', k'}^{(r)} \times \splres_{W, t, i_t - i', j_t - j', k_t - k'}^{(r)}$; since the coefficients $\alpha_t^{(r)}(i', j', k')$ sum up to 1, we know that the left-hand side $\xiWt^{(r)}$ has at most $2\eps$ distance from $\splresWt^{(r)}$ as well, and the same holds between $\xiWt$ and $\splresWt$. Thus the $\eps$-interface tensor with complete split distributions $\midBK{\xiWt}_{W, t}$ is contained in the $3\eps$-interface tensor with $\midBK{\splresWt}_{W, t}$ as a subtensor. Then we conclude the proof.
\end{proof}

The remainder of this section aims to prove \cref{prop:constituent-stage-no-eps}.

\subsection{Dividing into Regions}

For each of the $s'$ terms, say the $t$-th term, we pick three real numbers $A_{t,1}, A_{t,2}, A_{t,3} \ge 0$ where $A_{t,1} + A_{t,2} + A_{t,3} = 1$, that aims to divide the $t$-th term in the input level-$\lvl$ $\eps$-interface tensor into three regions of sizes $A_{t,1} n_t, A_{t,2} n_t$ and $A_{t,3} n_t$ respectively.
We also pick three different complete split distributions $\splresXt^{(1)}, \splresXt^{(2)}, \splresXt^{(3)}$, with the constraint
\begin{equation}
\splresXt^{(1)} A_{t,1} + \splresXt^{(2)} A_{t,2} + \splresXt^{(3)} A_{t,3} = \splresXt.
\end{equation}
We also pick $\splresYt^{(r)}$ and $\splresZt^{(r)}$ for $r \in [3]$ with similar constraints. Similar to \cref{rmk:assumptions_on_complete_split_dist},  we assume without loss of generality that, for every $t, r$ and every $L \in \{0, 1, 2\}^{2^{\lvl-1}}$, 
\[
  \splresXt^{(r)}(L) = \splresZt^{(r)}(\vec{2}-L) \text{ if } j_t = 0, \quad   
  \splresZt^{(r)}(L) = \splresYt^{(r)}(\vec{2}-L) \text{ if } i_t = 0, \quad   
  \splresYt^{(r)}(L) = \splresXt^{(r)}(\vec{2}-L) \text{ if } k_t = 0   
\]
where $\vec{2}$ denotes the length-$(2^{\lvl-1})$ vector whose coordinates are all $2$, 
and 
\[
  \splresXt^{(r)}(L) = 0 \text{ if } \sum_{t} L_t \ne i_t, \quad   
  \splresYt^{(r)}(L) = 0 \text{ if } \sum_{t} L_t \ne j_t, \quad   
  \splresZt^{(r)}(L) = 0 \text{ if } \sum_{t} L_t \ne k_t.   
\]

For any level-$1$ $X$-block, if the portion of it in the $r$-th region of the $t$-th term is not $\eps$-approximate consistent with $\splresXt^{(r)}$, we zero it out. We similarly handle level-1 $Y$-blocks and $Z$-blocks. It is not hard to see the following.

\begin{claim}
  \label{cl:dividing_into_region_zero_out}
  After the previous zeroing-out, we obtain a tensor that is isomorphic to
  \[
    \bigotimes_{r=1}^3 \bigotimes_{t = 1}^{s'} T_{i_t, j_t, k_t}^{\otimes A_{t,r} n_t}[\splresXt^{(r)}, \splresYt^{(r)}, \splresZt^{(r)}, \eps].
  \]
\end{claim}

\begin{proof}
  We only need to show that for a fixed $t$,
  \[
    T_{i_t, j_t, k_t}^{\otimes n_t}[\splresXt, \splresYt, \splresZt, \eps] \;\degen\; \bigotimes_{r=1}^3 T_{i_t, j_t, k_t}^{\otimes A_{t,r} n_t} [\splresXt^{(r)}, \splresYt^{(r)}, \splresZt^{(r)}, \eps]
    \numberthis \label{eq:dividing/sub_zero_out}
  \]
  by performing the above zeroing-out rule, i.e., zeroing out every level-1 $X$-block whose portion in the $r$-th region is not $\eps$-approximate consistent with $\splresXt^{(r)}$, and doing similarly for $Y$- and $Z$-blocks. Suppose some level-1 $X$-block belongs to the right-hand side and has complete split distributions $\xiXt^{(1)}, \xiXt^{(2)}, \xiXt^{(3)}$ in three regions respectively, each of which is at most $\eps$ away from $\splresXt^{(1)}, \splresXt^{(2)}, \splresXt^{(3)}$ in $L_{\infty}$ distance. Then, its average complete split distribution $\xiXt \defeq A_1 \xiXt^{(1)} + A_2 \xiXt^{(2)} + A_3 \xiXt^{(3)}$ has at most $\eps$ distance from $\splresXt$, which means that the considered level-1 $X$-block also belong to the left-hand side. It is the same for $Y$- and $Z$-blocks, so the right-hand side of \eqref{eq:dividing/sub_zero_out} is a subtensor of the left-hand side, i.e., Eq.~\eqref{eq:dividing/sub_zero_out} holds, which further implies the claim.
\end{proof}

In the following, we will focus on the first region $r=1$, in which we will apply asymmetric hashing that allows the sharing of $Z$-blocks. Let
\[ \T^{(1)} \defeq \bigotimes_{t = 1}^{s'} T_{i_t, j_t, k_t}^{\otimes A_{t,1} n_t}[\splresXt^{(1)}, \splresYt^{(1)}, \splresZt^{(1)}, \eps]. \]
We will omit the superscript $(1)$ on all variables for conciseness.

\subsection{Asymmetric Hashing}

Next, we apply hashing similarly to the global stage. For every $t \in [s']$, recall that $\alpha_t$ is a distribution on $\{(i', j', k') \in \mathbb{Z}_{\ge 0}^3: i' + j' + k' = 2^{\lvl-1}\}$. Additionally, the marginal distributions of $\alpha_t(i', j', k')$ on the three dimensions are the same as $\splonelevelXt(i', i_t - i')$, $\splonelevelYt(j', j_t - j')$, $\splonelevelZt(k', k_t - k')$, respectively.

Each level-$(\lvl-1)$ index sequence is partitioned into $s'$ parts, where each part corresponds to one term in $\T$. The $t$-th part is a length-$(2n_t)$ $\{0, \ldots, 2^{\lvl-1}\}$-sequence, which can also be viewed as a length-$(n_t)$ $\{0, \ldots, 2^{\lvl-1}\}^2$-sequence by combining pairs of adjacent numbers. If the $t$-th part of a level-$(\lvl-1)$ $X$-index sequence is not consistent with the distribution $\splonelevelXt$ for any $t$, we zero out the corresponding level-$(\lvl-1)$ $X$-block. We similarly handle the $Y$- and $Z$-blocks.

Let $\numxblock$ be the number of remaining level-$(\lvl-1)$ $X$-blocks, and it is not difficult to see that
\begin{equation}
  \label{eq:constituent:NBX}
  \numxblock = 2^{\sum_{t} H(\splonelevelXt) \cdot A_{t,1} n_t \pm o(n)}.
\end{equation}
Similarly, let $\numyblock$ and $\numzblock$ be the number of remaining $Y$- and $Z$-blocks, and we have 
\begin{equation}
\label{eq:constituent:NBY-NBZ}
    \numyblock = 2^{\sum_{t} H(\splonelevelYt) \cdot A_{t,1} n_t \pm o(n)}, \quad \numzblock = 2^{\sum_{t} H(\splonelevelZt) \cdot A_{t,1} n_t \pm o(n)}. 
\end{equation}
Let $\numalpha$ be the number of remaining block triples that are consistent with $\{\alpha_t\}_{t \in [s']}$. We have 
\begin{equation}
\label{eq:constituent:numalpha}
    \numalpha = 2^{\sum_t H(\alpha_t) \cdot A_{t,1} n_t \pm o(n)}.
\end{equation}
Finally, let $\numtriple$ be the number of remaining level-$(\lvl-1)$ block triples $X_IY_JZ_K$.
\begin{claim}
  \label{cl:constituent:numtriple-bound}
  $\numtriple = 2^{\sum_t (H(\alpha_t) + P_{\alpha, t}) \cdot A_{t,1} n_t \pm o(n)}$, where we recall that $P_{\alpha, t} \defeq \max_{\alpha'_t \in D_t} H(\alpha'_t) - H(\alpha_t)$ in which $D_t$ is the set of distributions sharing the same marginals as $\alpha_t$.
\end{claim}
\begin{proof}
  Fixing a series of distributions $\alpha'_t \in D_t$ ($t = 1, 2, \ldots, s'$), the number of level-$(\lvl - 1)$ block triples consistent with $\midBK{\alpha'_t}_{t \in [s']}$ equals
  \[
    2^{\sum_t H(\alpha'_t) \cdot A_{t,1} n_t \pm o(n)} \;\le\; 2^{\sum_t \max_{\alpha''_t \in D_t} H(\alpha''_t) \cdot A_{t,1} n_t \pm o(n)} \;\le\; 2^{\sum_t (H(\alpha_t) + P_{\alpha, t}) \cdot A_{t,1} n_t \pm o(n)}.
  \]
  Taking summation over all $\poly(n) = 2^{o(n)}$ series of distributions $\midBK{\alpha'_t}_{t \in [s']}$ will prove the claim.
\end{proof}

Let $M \in [M_0, 2M_0]$ be a prime number for some integer $M_0$. Similar as before, the value of $M_0$ is yet to be fixed, but we first require that 
\begin{equation}
\label{eq:constituent:M_0_bounds_1}
M_0 \ge 8\cdot \max\left\{\frac{\numtriple}{\numxblock}, \frac{\numtriple}{\numyblock}\right\}.
\end{equation}

We independently pick uniformly random elements $b_0, \{w_p\}_{p=0}^{2n} \in \{0, \ldots, M-1\}$, and define the following hash functions $h_X, h_Y, h_Z : \{0, \ldots, 2^{\lvl-1}\}^n \rightarrow \{0, \ldots, M - 1\}$:
\begin{align*}
    h_X(I) &= b_0 + \left(\sum_{p=1}^{2n} w_p \cdot I_p \right) \bmod M,\\
    h_Y(J) &= b_0 + \left(w_0 + \sum_{p=1}^{2n} w_p \cdot J_p \right) \bmod M,\\
    h_Z(K) &= b_0 + \frac{1}{2}\left(w_0+\sum_{p=1}^{2n} w_p \cdot (2^{\lvl-1} - K_p) \right) \bmod M.
\end{align*}

Next, for a Salem-Spencer subset $B$ of $\{0, \ldots, M - 1\}$ that has size $M^{1-o(1)}$, we zero out all level-$(\lvl-1)$ blocks $X_I$ with $h_X(I) \notin B$, $Y_J$ with $h_Y(J) \notin B$, and $Z_K$ with $h_Z(K) \notin B$. Then all remaining block triples are contained in a bucket $b$ for some $b \in B$.

For every bucket $b$, if it contains two level-$(\lvl-1)$ triples $X_I Y_J Z_K$ and $X_I Y_{J'}Z_{K'}$ that share the same $X$-block, then we zero out $X_I$. We similarly handle $Y$-blocks. We repeatedly perform the previous zeroing-outs so that all remaining triples do not share $X$- or $Y$-blocks. For every level-$(\lvl-1)$ block  $X_I$ (or $Y_J$), we check whether the unique triple containing it is consistent with  $\{\alpha_t\}_{t \in [s']}$; if not, we zero out $X_I$ (or $Y_J$). We call the tensor after this step $\TPrime$. 

The following claims, which are analogous to the claims in \cref{sec:global}, still hold, and we omit their proofs to conciseness. 

\begin{claim}[Implicit in \cite{cw90}, see also \cite{duan2023}]
  \label{cl:constituent:hash-function-basic-properties}
  For a level-$(\lvl-1)$ block triple $X_I Y_J Z_K \in \T$, and for every $b \in \{0, \ldots, M - 1\}$,
  \[\Pr\Bk{\tall h_X(I) = h_Y(J) = h_Z(K) = b} = \frac{1}{M^2}.\]
  Furthermore, for two different block triples $X_I Y_J Z_K,  X_I Y_{J'} Z_{K'} \in \T$ that share the same $X$-block, and for every $b \in \{0, \ldots, M - 1\}$,
  \[\Pr\Bk{\tall h_X(I) = h_Y(J') = h_Z(K') = b \;\middle\vert\; h_X(I) = h_Y(J) = h_Z(K) = b} = \frac{1}{M}. \]
  This also holds analogously for different block triples that share the same $Y$-block or $Z$-block.
\end{claim}

\begin{claim}
  \label{cl:constituent:prob-IJK-remain}
  For every $b \in B$ and for every level-$(\lvl-1)$ block triple $X_I Y_J Z_K \in \T$ that is consistent with $\{\alpha_t\}_{t \in [s']}$, the probability that $X_I Y_J Z_K$ remains in $\TPrime$ conditioned on $h_X(I) = h_Y(J) = h_Z(K) = b$ is $\ge \frac{3}{4}$. 
\end{claim}

\begin{claim}
  \label{cl:constituent:expected-IJK-remain}
  The expected number of level-$(\lvl-1)$ block triples in $\TPrime$ is at least $\numalpha \cdot M_0^{-1-o(1)}$.
\end{claim}

\subsection{Compatibility Zero-Out I}

Recall that for every $W \in \{X, Y, Z\}$ and $i' + j' + k' = 2^{\lvl - 1}$, $\splres_{W, t, i', j', k'}$ is a level-$(\lvl-1)$ complete split distribution, and they satisfy
\begin{equation}
\label{eq:constituent:complete-split-dist-constraint}
\splres_{W, t} = \sum_{i', j', k'} \alpha_{t}(i', j', k') \cdot \left(\splres_{W, t, i', j', k'} \times \splres_{W, t,  i_t - i', j_t - j', k_t - k'}\right).
\end{equation}
Let
\[ S^{(I, J, K)}_{t, i', j', k'} \defeq \{p \textup{ is in the $t$-th term} \mid I_p = i', J_p = j', K_p = k'\}, \]
and
\[S^{(K)}_{t, *, *, k'} := \{p \textup{ is in the $t$-th term} \mid K_p = k'\}, \quad S_{t, *, *, *} := \{p \textup{ is in the $t$-th term}\}.  \]
If clear from the context, we will drop the superscript $(I, J, K)$ or $(K)$.

Recall that in $\TPrime$, every level-$(\lvl-1)$ block $X_I$ is in a unique block triple $X_I Y_J Z_K$. For every level-1 block $X_{\hat{I}} \in X_I$, we will zero out $X_{\hat{I}}$ if $\split(\hat I, S_{t, i',j',k'}) \ne \splres_{X, t, i', j', k'}$ for any $t, i', j', k'$. Similarly, every level-$\lvl$ block $Y_J$ is in a unique block triple, and we zero out every $Y_{\hat{J}} \in Y_J$ where $\split(\hat J, S_{t, i',j',k'}) \ne \splres_{Y, t, i', j', k'}$ for any $t, i', j', k'$.

For every level-$1$ block $Z_{\hat K} \in Z_K$, we zero out $Z_{\hat K}$ if $\split(\hat K, S_{t, *,*,k'}) \ne \splresavg_{Z, t, *,*,k'}$ for any $t, k'$,  where
\[
  \splresavg_{Z, t, *,*,k'} \defeq \frac{\sum_{i'+j' = 2^{\lvl-1} - k'} \bigbk{\alpha(i', j', k') + \alpha(i_t-i', j_t-j', k_t-k')} \cdot \splres_{Z,i',j',k'}}{\sum_{i'+j' = 2^{\lvl-1} - k'} \bigbk{\alpha(i', j', k')+ \alpha(i_t-i', j_t-j', k_t-k')}}.
\]

We call the tensor after the previous zeroing-outs $\TDoublePrime$.

Next, we define the notion of compatibility. 

\begin{definition}[Compatibility]
  \label{def:constituent:compatibility}
  For some $I, J, K$, a level-$1$ block $Z_{\hat{K}} \in Z_K$ is compatible with a level-$(\lvl-1)$ triple $X_IY_J Z_K$ if 
  \begin{enumerate}
  \item 
    \label{item:constituent:compatibility1} For every $t$ and every $(i', j', k') \in \mathbb{Z}_{\ge 0}^3 \cap [0, i_t] \times [0, j_t] \times [0, k_t]$ with $ i' + j' + k' = 2^{\lvl-1}$, $i' = 0 \text{ or } j' = 0$, there is $\split(\hat K, S_{t, i',j',k'}) = \splres_{Z,t, i',j',k'}$.
  \item
    \label{item:constituent:compatibility2}
    For every $t$ and every index $k' \in \{0, 1, \ldots, \min\{2^{\lvl-1}, k_t\}\}$,  $\split(\hat K, S_{t, *,*,k'}) = \splresavg_{Z,t,  *,*,k'}$.
  \end{enumerate}
\end{definition}

\begin{claim}
  \label{cl:constituent:compatible}
  In $\TDoublePrime$, for every remaining level-1 block triple $X_{\hat I}Y_{\hat J}Z_{\hat K}$ and the level-$(\lvl-1)$ block triple $X_I Y_J Z_K$ that contains it, $Z_{\hat K}$ is compatible with $X_I Y_J Z_K$.
\end{claim}

The proof of this claim is the same as \cref{cl:global:compatible}. 

\subsection{Compatibility Zero-Out II: Unique Triple}

In this step, we zero out all level-$1$ $Z$-block $Z_{\hat K}$ that are compatible with more than one level-$(\lvl-1)$ triples and they become holes. After this step, each remaining level-$1$ $Z$-block $Z_{\hat{K}}\in Z_K$ is compatible with a unique level-$(\lvl-1)$ triple $X_I Y_J Z_K$ containing it.

\subsection{Usefulness Zero-Out}

Next, we further zero out some level-1 $Z$-blocks using the following definition of \emph{usefulness}.

\begin{definition}[Usefulness]
For a level-1 block $Z_{\hat K}$ and a level-$(\lvl-1)$ triple $X_I Y_J Z_K$ containing it, if for all $t, i', j', k'$ we have $\split(\hat K, S_{t, i',j',k'}) = \splres_{Z,t, i',j',k'}$, then we say that $Z_{\hat K}$ is \emph{useful} for $X_I Y_J Z_K$.
\end{definition}

For each $Z_{\hat K}$, it appears in a unique triple $X_I Y_J Z_K$ by the previous zeroing out. Furthermore, if $Z_{\hat K}$ is not useful for this triple, we zero out $Z_{\hat K}$. We call the current tensor $\TTriplePrime$.

Ideally, we want the subtensor of $\TTriplePrime$ over each triple $X_I Y_J Z_K$ to be isomorphic to 
\[ \mathcal{T}^* = \bigotimes_{t \in [s']} \; \bigotimes_{i'+j'+k' = 2^{\lvl-1}} T_{i',j',k'}^{\otimes A_{t,1} \cdot (\alpha_t(i', j', k')+\alpha_t(i_t-i', j_t-j', k_t-k')) \cdot n_t} [\splres_{X, t, i', j', k'}, \splres_{Y, t, i', j', k'}, \splres_{Z, t, i', j', k'}]. \]
However, there will be two types of holes. The first type of holes is caused by the fact that some level-$1$ subtensors are already missing in the input tensor because we enforced complete split distributions $\splresXt, \splresYt, \splresZt$ on it; the second type of holes is caused by zeroing out $Z_{\hat{K}}$ that are compatible with multiple level-$(\lvl-1)$ triples. In the next section, we will analyze and fix these two types of holes. 

\subsection{Fixing Holes}

First, we analyze the fraction of holes that are caused by the complete split distributions enforced in the input. To do so, we focus on a fixed triple $X_I Y_J Z_K$ and the subtensor $\T^*$ we desire. Then we take a random level-$1$ block that is not zeroed out in $\T^*$, and upper bound the probability that this level-$1$ block is zeroed out in the input level-$\lvl$ $\eps$-interface tensor. By symmetry, it suffices to focus on $X$-blocks. 

Fix any $(\hat{i}_1, \ldots, \hat{i}_{2^{\lvl-1}})$, let us analyze the fraction of its occurrences in a random level-$1$ $X$-block in $\T^*$. For every $t \in [s']$, and for every $i', j', k'$, we first focus on the level-$\lvl$ positions in the $t$-th term where $(i_t, j_t, k_t)$ is split into $(i', j', k')$ and $(i_t - i', j_t - j', k_t - k')$ (thus, there are $A_{t,1} \cdot \alpha_t(i', j', k') \cdot n_t$ such positions). Among these positions, we want to analyze the number of positions that correspond to the level-$1$ chunk $(\hat{i}_1, \ldots, \hat{i}_{2^{\lvl-1}})$. Therefore, the first half-chunk, which corresponds to $(i', j', k')$, should be $(\hat{i}_1, \ldots, \hat{i}_{2^{\lvl-2}})$, and the second half-chunk, which corresponds to $(i_t - i', j_t - j', k_t - k')$, should be $(\hat{i}_{2^{\lvl-2}+1}, \ldots, \hat{i}_{2^{\lvl-1}})$. 

There are $A_{t,1} \cdot (\alpha_t(i', j', k')+\alpha_t(i_t-i', j_t-j', k_t-k')) \cdot n_t$ level-$(\lvl-1)$ positions corresponding to $(i', j', k')$, and among them, $A_{t,1} \cdot \alpha_t(i', j', k') \cdot n_t$ are in odd positions. By definition of $\T^*$, the fraction of $(\hat{i}_1, \ldots, \hat{i}_{2^{\lvl-2}})$ in these $A_{t,1} \cdot (\alpha_t(i', j', k')+\alpha_t(i_t-i', j_t-j', k_t-k')) \cdot n_t$ positions is $\splres_{X, t, i', j', k'} (\hat{i}_1, \ldots, \hat{i}_{2^{\lvl-2}})$, and if we take a random level-$1$ $X$-block in $\T^*$, the fraction of $(\hat{i}_1, \ldots, \hat{i}_{2^{\lvl-2}})$ among the odd positions corresponding to $(i', j', k')$ is $\splres_{X, t, i', j', k'} (\hat{i}_1, \ldots, \hat{i}_{2^{\lvl-2}}) \pm o(1)$ with $1-1/\poly(n)$ probability, by concentration bounds. Furthermore, the subset of positions in these $A_{t,1} \cdot \alpha_t(i', j', k') \cdot n_t$ positions is also random. Similarly, with $1-1/\poly(n)$ probability, the fraction of  $(\hat{i}_{2^{\lvl-2}+1}, \ldots, \hat{i}_{2^{\lvl-1}})$ in the even positions corresponding to $(i_t - i', j_t - j', k_t - k')$ is  $\splres_{X, t, i', j', k'} (\hat{i}_{2^{\lvl-2}+1}, \ldots, \hat{i}_{2^{\lvl-1}}) \pm o(1)$, and the positions are also random. Applying concentration bounds again, we get that the fraction of level-$\lvl$ positions corresponding to $(\hat{i}_{1}, \ldots, \hat{i}_{2^{\lvl-1}})$ among positions that split into $(i', j', k')$ and $(i_t - i', j_t - j', k_t - k')$
is $$\splres_{X, t, i', j', k'} (\hat{i}_{1}, \ldots, \hat{i}_{2^{\lvl-2}})\cdot \splres_{X, t, i_t-i', j_t-j', k_t-k'} (\hat{i}_{2^{\lvl-2}+1}, \ldots, \hat{i}_{2^{\lvl-1}}) \pm o(1).$$

Summing over all $i', j', k'$, we get that with probability $1-1/\poly(n)$, the fraction of level-$\lvl$ positions with $(\hat{i}_{1}, \ldots, \hat{i}_{2^{\lvl-1}})$ is 
\begin{align*}
& \sum_{i', j', k'} \alpha_t(i', j', k') \cdot \splres_{X, t, i', j', k'} (\hat{i}_{1}, \ldots, \hat{i}_{2^{\lvl-2}})\cdot \splres_{X, t, i_t-i', j_t-j', k_t-k'} (\hat{i}_{2^{\lvl-2}+1}, \ldots, \hat{i}_{2^{\lvl-1}}) \pm o(1)\\
=\;& \splres_{X, t}(\hat{i}_{1}, \ldots, \hat{i}_{2^{\lvl-1}}) \pm o(1). \tag{by \cref{eq:constituent:complete-split-dist-constraint}}
\end{align*}
The $o(1)$ term can become less than $\eps$, and the $1-1/\poly(n)$ probability can be bounded by $1-1/n^2$
for sufficiently large $n$. Therefore, a random level-$1$ $X$-block appears in $\T$ with probability at least $1-1/n^2$. This means that the fraction of holes caused by the complete split distributions enforced in the input is $1-1/n^2$ for the $X$-dimension. By symmetry, the same also holds for the $Y$- and $Z$-dimensions.

\bigskip

Next, we focus on holes caused by zeroing out $Z_{\hat{K}}$ that are compatible with multiple level-$(\lvl-1)$ triples. The analysis will be similar to \cref{sec:global:fix-holes}.

First, notice that for every level-1 $Z$-block $Z_{\hat K}$ that appears in the input of the constituent stage, its complete split distribution $\xiZt$ in the $t$-th term must be within $\eps$ $L_\infty$-distance to the given parameter $\splresZt$. Then we define $\pcomp$ as follows:

\begin{definition}[$\pcomp$]
  For fixed $Z_{\hat{K}}$ and $Z_K$ where $Z_{\hat{K}} \in Z_K$ and $\hat K$ has level-$\l$ complete split distributions $\midBK{\xiZt}_{t \in [s']}$, we define $\pcomp^*(\midBK{\xiZt}_{t \in [s']})$ as the probability that a uniformly random block triple $X_I Y_J Z_K$ consistent with $\{\alpha_t\}_{t \in [s']}$ is compatible with $Z_{\hat{K}}$. We further define $\displaystyle \pcomp \defeq \max_{\substack{\midBK{\xiZt}_{t \in [s']} \,: \\ \norm{\xiZt - \splresZt}_\infty \le \eps \; \forall t}} \pcomp^*(\midBK{\xiZt}_{t \in [s']})$.
\end{definition}

By symmetry between level-$\lvl$ positions, this probability $\pcomp^*(\midBK{\xiZt}_{t \in [s']})$ is the same for different $\hat K$ that have the same complete split distributions, so $\pcomp^*$ and $\pcomp$ is well-defined.

\begin{claim}
  The value of $\pcomp^*(\midBK{\xiZt}_{t \in [s']})$ is at most
  \[
    2^{\sum_{t \in [s']} \left( \lambda_{Z,t} - H(\xiZt) + H(\gamma_{Z, t}) \right) A_{t,1} \cdot n_t \pm o(n)},
  \]
  where we recall that
  \begin{align*}
    \lambda_{Z,t} = & \sum_{i', j', k' \,:\, i' = 0 \textup{ or } j' = 0} \bigbk{\alpha_t(i', j', k') + \alpha_t(i_t-i', j_t-j', k_t-k')} \cdot H(\splres_{Z, t, i', j', k'}) \\
    & + \sum_{k'} \bigbk{\alpha_t(\+, \+, k') + \alpha_t(\<, \<, k_t-k')} \cdot H(\splresavg_{Z, \+, \+, k_t-k'}),
  \end{align*}
  and 
  \[
    \alpha_t(\+, \+, k') = \sum_{i' > 0, \, j' > 0} \alpha_t(i', j', k'), \quad \alpha_t(\<, \<, k') = \sum_{i' < i_t, \, j' < j_t} \alpha_t(i', j', k').
  \]
  Furthermore,
  \[
    \pcomp \,\le\, 2^{\sum_{t \in [s']} \left( \lambda_{Z,t} - H(\splresZt) + H(\gamma_{Z, t}) + \oeps(1) \right) A_{t,1} \cdot n_t + o(n)}.
    \numberthis \label{eq:constituent:pcomp_upper_bound}
  \]
\end{claim}

\begin{proof}
  Similar to before, it suffices to compute the following two quantities, and $\pcomp^*(\midBK{\xiZt}_t)$ will be the ratio between them:
  \begin{enumerate}[label=(\arabic*)]
  \item the number of tuples $(I, J, K, \hat{K})$ where $X_I Y_J Z_K$ is consistent with $\{\alpha_t\}_{t \in [s']}$, $\hat{K} \in K$, $\hat K$ has complete split distributions $\midBK{\xiZt}_t$, and $Z_{\hat{K}}$ is compatible with $X_I Y_J Z_K$;
  \item the number of $(I, J, K, \hat{K})$ where $X_I Y_J Z_K$ is consistent with $\{\alpha_t\}_{t \in [s']}$, $\hat{K} \in K$, and $\hat K$ has complete split distributions $\midBK{\xiZt}_t$.
  \end{enumerate}

  We first compute the second quantity. First, the number of $Z_{\hat{K}}$ with the desired complete split distributions $\midBK{\xiZt}_t$ is $2^{\sum_t H(\xiZt) \cdot A_{t, 1} \cdot n_t \pm o(n)}$. Each of these $Z_{\hat{K}}$ uniquely determines a level-$(\lvl-1)$ block $Z_K$. Also, for each $Z_K$, the number of block triples $X_I Y_J Z_K$ consistent with $\{\alpha_t\}_{t \in [s']}$ is $\frac{\numalpha}{\numzblock} = 2^{\sum_t (H(\alpha_t) - H(\gamma_{Z, t})) \cdot A_{t,1} \cdot n_t \pm o(n)}$. Therefore, the second quantity is
  \begin{equation}
    \label{eq:constituent:pcomp_denominator}
    2^{\sum_t (H(\xiZt) + H(\alpha_t) - H(\gamma_{Z, t})) \cdot A_{t,1} \cdot n_t \pm o(n)}.
  \end{equation}

  Then, we compute the first quantity, which does not exceed the number of $(I, J, K, \hat{K})$ where $X_I Y_J Z_K$ is consistent with $\{\alpha_t\}_{t \in [s']}$, $\hat{K} \in K$, and $Z_{\hat K}$ is compatible with $X_I Y_J Z_K$. (We dropped the condition of having correct level-$\l$ complete split distributions $\midBK{\xiZt}_t$ and got an overestimation.)

  First, the number of block triples $X_I Y_J Z_K$ consistent with $\{\alpha_t\}_{t \in [s']}$ is $\numalpha$. Then, for each such block triple, we count the number of $Z_{\hat{K}} \in Z_K$ that is compatible with it. If we fix some $X_I Y_J Z_K$, then we also have fixed the values of $S_{t, i, j, k}$ for all $t, i, j, k$. Then it is not difficult to see that the following condition is equivalent to the condition for $Z_{\hat{K}}$ being compatible with $X_I Y_J Z_K$:

\begin{definition}[Compatibility']
  \label{def:constituent:compatibility'}
  For level-$(\l - 1)$ triple $X_I Y_J Z_K$ consistent with $\midBK{\alpha_t}_{t \in [s']}$, a level-1 block $Z_{\hat{K}} \in Z_K$ is compatible with $X_IY_J Z_K$ if
  \begin{itemize}
  \item For every $t$ and every $(i', j', k') \in \mathbb{Z}_{\ge 0}^3 \cap [0, i_t] \times [0, j_t] \times [0, k_t]$ with $i' + j' + k' = 2^{\lvl-1}$, $i' = 0 \text{ or } j' = 0$, there is $\split(\hat K, S_{t, i',j',k'}) = \splresZt[t,i',j',k']$. (This is exactly \cref{item:constituent:compatibility1} in \cref{def:constituent:compatibility}).
  \item For every $t, k$, let $S_{t, \+, \+, k} := \bigcup_{i > 0, j > 0} S_{t, i, j, k}$. Then $\split(\hat K, S_{t, \+, \+, k}) = \splresavg_{Z, t, \+, \+, k}$.
  \end{itemize}
\end{definition}

In this definition, there are constraints on the complete split distributions of $\hat{K}$ on some disjoint subsets of level-$(\l - 1)$ positions, i.e., subsets of $\bigBk{2 \sum_t A_{t,1} n_t}$. Therefore, we can count the number of valid subsequences of $\hat{K}$ for each of these subsets of indices, and multiply them together to get the number of valid $\hat K$. For every $t$ and every $(i', j', k') \in \mathbb{Z}_{\ge 0}^3 \cap [0, i_t] \times [0, j_t] \times [0, k_t]$ where $i' + j' + k' = 2^{\lvl-1}$ with $i' = 0$ or $j' = 0$, we require that $\split(\hat{K}, S_{t, i', j', k'}) = \splres_{Z, t, i', j', k'}$, so the number of possibilities of $\hat{K}$ on the subset of indices $S_{t, i',j', k'}$ is
\[ 2^{H(\splres_{Z, t, i', j', k'}) \cdot |S_{t, i', j', k'}| \pm o(n)} = 2^{H(\splres_{Z, t, i', j', k'}) \cdot (\alpha_t(i', j', k') + \alpha_t(i_t - i', j_t - j', k_t - k')) \cdot A_{t,1} n_t \pm o(n)}. \]
For every $t, k$, we require that $\split(\hat K, S_{t, \+, \+,k'}) = \splresavg_{Z, t, \+, \+, k'}$, so the number of possibilities of $\hat K$ on $S_{t, \+, \+,k'}$ is
\[ 2^{H(\splresavg_{Z, t, \+, \+, k'}) \cdot |S_{t, \+, \+, k'}| \pm o(n)} = 2^{H(\splresavg_{Z, t, \+, \+, k'}) \cdot (\alpha_t(\+, \+, k') + \alpha_t(\<, \<, k_t - k')) \cdot A_{t,1} n_t \pm o(n)}. \]
Overall, the number of possible compatible $\hat{K}$, multiplied by the number of block triples $X_I Y_J Z_K$, is
\begin{align}
  \label{eq:constituent:pcomp_numerator}
  \begin{split}
    &\numalpha \cdot \prod_{\substack{t, i', j', k' \\ i' = 0 \textup{ or } j' = 0}} 2^{H(\splresavg_{Z, t, i', j', k'}) \cdot (\alpha_t(i', j', k') + \alpha_t(i_t- i', j_t-j', k_t-k') )\cdot A_{t,1} n_t \pm o(n)} \\
    & \qquad \cdot \prod_{t, k'} 2^{H(\splresavg_{Z, t, \+, \+, k'}) \cdot (\alpha_t(\+, \+, k')+\alpha_t(\<, \<, k_t-k')) \cdot A_{t,1} n_t \pm o(n)}\\ 
    ={} & 2^{\sum_t (H(\alpha_t) + \lambda_{Z,t}) \cdot A_{t,1} n_t \pm o(n)}.
  \end{split}
\end{align}

Finally, as mentioned, $\pcomp^*(\midBK{\xiZt}_{t})$ is the ratio between \eqref{eq:constituent:pcomp_numerator} and \eqref{eq:constituent:pcomp_denominator}, so
\[
  \pcomp^*(\midBK{\xiZt}_{t \in [s']}) \,\le\, 2^{\sum_t \left(\lambda_{Z,t} - H(\xiZt) + H(\gamma_{Z, t})\right) A_{t,r} \cdot n_t + o(n)}
\]
as desired. The bound \eqref{eq:constituent:pcomp_upper_bound} on $\pcomp$ follows as the $L_\infty$ distance between $\midBK{\xiZt}_t$ and $\midBK{\splresZt}_t$ is at most $\eps$.
\end{proof}

The proof of the following claim is essentially the same as that of \cref{cl:global:prob-of-holes}.
\begin{claim}
  \label{cl:constituent:prob-of-holes}
  For every $b \in B$, every level-$(\lvl-1)$ block triple $X_I Y_J Z_K$ consistent with $\{\alpha_t\}_{t \in [s']}$, and for each typical $Z_{\hat{K}} \in Z_K$, the probability that $Z_{\hat{K}}$ is compatible with multiple triples in $\TDoublePrime$ is at most
  \[ \frac{\numalpha \cdot \pcomp}{\numzblock \cdot M_0}, \]
  conditioned on $h_X(I) = h_Y(J) = h_Z(K) = b$.
\end{claim}

Recall that we require $M_0$ to be at least $8 \cdot \max \BigBK{\frac{\numtriple}{\numxblock}, \frac{\numtriple}{\numyblock}}$. Now, we add another (and final) constraint: $M_0 \ge \frac{\numalpha \cdot \pcomp}{\numzblock} \cdot n^2$.  That is, we will set $M_0$ to be
\begin{gather*}
 \max\left\{\frac{8\numtriple}{\numxblock}, \; \frac{8\numtriple}{\numyblock}, \; \frac{\numalpha \cdot \pcomp}{\numzblock} \cdot n^2\right\} \\
\le 2^{\max\{\sum_t (H(\alpha_t) - P_{\alpha, t} - H(\gamma_{X, t})) A_{t,1} \cdot n_t, \;\; \sum_t (H(\alpha_t) - P_{\alpha, t} - H(\gamma_{Y, t})) A_{t,1} \cdot n_t, \;\; \sum_t (H(\alpha_t) + \lambda_{Z,t} - H(\splresZt)) A_{t,1} \cdot n_t\} + o(n)}.
\end{gather*}

Similar to before, for every $b \in B$ and every level-$(\lvl-1)$ block triple $X_I Y_J Z_K$ that is consistent with $\{\alpha_t\}_{t \in [s']}$ and $h_X(I) = h_Y(J) = h_Z(K) = b$, with constant probability, it remains in $\TPrime$ and the fraction of holes caused by enforcing that each $Z_{\hat{K}}$ is compatible with a unique triple is $1/n^2$. Additionally, as discussed earlier, the fraction of holes caused by the input complete split distribution constraints are also $1/n^2$. Overall, we expect to get $\numalpha \cdot M^{-1-o(1)}$ copies of $\T^*$ whose fraction of holes is $O(1/n^2)$.

By \cref{cor:fix-interface}, we can degenerate them into $\numalpha \cdot M^{-1-o(1)}$ unbroken copies of $\T^*$ because $O(1/n^2) \le \frac{1}{8N}$ for sufficiently large $n$.

\subsection{Summary}

In the analysis, we have degenerated $\bigotimes_{t = 1}^{s'} T_{i_t, j_t, k_t}^{\otimes A_{t,1} n_t}\bigBk{\splresXt^{(1)}, \splresYt^{(1)}, \splresZt^{(1)}, \eps}$ into $\ge \numalpha \cdot M_0^{-1-o(1)}$ copies of a level-$(\lvl-1)$ interface tensor $\T^*$ with parameter list
\[
  \begin{aligned}
    \Big\{ \Big(
    A_{t,1} \cdot n_t \cdot \bigbk{\alpha^{(1)}_t(i', j', k')+\alpha^{(1)}_t(i_t-i', j_t-j', k_t-k')}, & \\
    i', \; j', \; k', \; \splres^{(1)}_{X, t, i', j', k'}, \; \splres^{(1)}_{Y, t, i', j', k'}, \; \splres^{(1)}_{Z, t, i', j', k'} & \Big) \Big\} _{t \in [s'], i' + j' + k' = 2^{\lvl-1}}.
  \end{aligned}
\]
By plugging in the bounds of $\numalpha$ and $M_0$, we see that the number of copies we obtained (in the first region) is 
\[
  2^{\min\left\{\sum_{t \in [s']} A_{t,1} \cdot n_t \cdot \left( H(\splonelevelXt[t]^{(1)}) - P_{\alpha, t}^{(1)}\right), \;\; \sum_{t \in [s']} A_{t,1} \cdot n_t \cdot \left( H(\splonelevelYt[t]^{(1)}) - P_{\alpha, t}^{(1)}\right), \;\; \sum_{t \in [s']} A_{t,1} \cdot n_t \cdot \left(H(\splres_{Z, t}^{(1)}) - \lambda_{Z, t}^{(1)}\right)\right\} - \oeps(n) - o(n)}.
\]
We conclude the proof by applying the same method to the second and third region, where for the second region we perform asymmetric hashing that shares $Y$-variable blocks, and for the third region we perform asymmetric hashing that shares $X$-blocks, and taking the tensor product of these returned results.

\section{Fixing Holes}
\label{sec:fixing_holes}

In this section, we show (by generalizing a result by Duan \cite{duanpersonal}) that we can degenerate a direct sum of some broken copies of an interface tensor into an unbroken copy of the same tensor as long as we only have a small fraction of holes in the $X$-, $Y$-, $Z$-dimensions. Since our result of fixing holes in all $X$-, $Y$-, $Z$-variables might be of independent interest, we present our result in a more general setting.

Let us first describe the setup of this section. We consider a partitioned tensor $T$ on variable sets $X = \{x_1,\dots, x_{N_X}\}$, $Y = \{y_1,\dots, y_{N_Y}\}$, $Z = \{z_1,\dots, z_{N_Z}\}$ of size $|X| = N_X$, $|Y| = N_Y$, $|Z| = N_Z$ with partitions $X = \bigsqcup_{i = 1}^{M_X} X_i$, $Y = \bigsqcup_{j = 1}^{M_Y} Y_k$, $Z = \bigsqcup_{k = 1}^{M_Z} Z_k$ into equal-size parts $|X_i| = m_X$ for all $i\in [M_X]$, $|Y_j| = m_Y$ for all $j\in [M_Y]$, and $|Z_k| = m_Z$ for all $k\in [M_K]$. (We use the notation $X_i$ to represent both the part itself and the set of elements in this part.) Let $\calP_X = \{X_i \mid i \in [M_X]\}$ denote the set of parts in the partition of $X$, and similarly let $\calP_Y, \calP_Z$ denote the set of parts in the partition of $Y$ and $Z$ respectively. Note that by definition $N_X = M_X \cdot m_X$, $N_Y = M_Y \cdot m_Y$, $N_Z = M_Z \cdot m_Z$ and $|\calP_X| = M_X$, $|\calP_Y| = M_Y$, $|\calP_Z| = M_Z$.

We consider the broken copies of $T$ where some of the $X$-, $Y$- and $Z$-parts are missing which we call the \emph{holes}. (Equivalently, the variables in a part are either all present or all missing.) More specifically, we say that $T_\hole$ is \emph{a broken copy of $T$ with holes $P^{(0)}_X \subseteq M_X$, $P^{(0)}_Y \subseteq M_Y$, $P^{(0)}_Z \subseteq M_Z$} when
\[
  T_\hole \;=\; T \big\vert_{X \setminus \bigsqcup_{X_t \in P^{(0)}_X} X_t, \; Y \setminus \bigsqcup_{Y_t \in P^{(0)}_Y} Y_t, \; Z \setminus \bigsqcup_{Z_t \in P^{(0)}_Z} Z_t}
  \numberthis \label{eq:naive_definition_Thole}
\]
is obtained from $T$ via zeroing out the variables in the parts $P^{(0)}_X \subseteq \calP_X$, $P^{(0)}_Y \subseteq \calP_Y$, $P^{(0)}_Z \subseteq \calP_Z$. For simplicity, we define the notation
\[
  T \Vert_{P_X, \, P_Y, \, P_Z} \;\defeq\; T \big\vert_{\bigsqcup_{X_t \in P_X} X_t, \; \bigsqcup_{Y_t \in P_Y} Y_t, \; \bigsqcup_{Z_t \in P_Z} Z_t}
\]
to represent the subtensor of $T$ over the set of parts $P_X, P_Y, P_Z$. With this notation, \cref{eq:naive_definition_Thole} can be rewritten as
\[
  \Thole \;=\; T \big\Vert_{\calP_X \setminus P^{(0)}_X, \; \calP_Y \setminus P^{(0)}_Y, \; \calP_Z \setminus P^{(0)}_Z}.
\]
We call the ratios $\bigabs{P^{(0)}_X} \big/ M_X$, $\bigabs{P^{(0)}_Y} \big/ M_Y$, $\bigabs{P^{(0)}_Z} \big/ M_Z$ the \emph{fraction of holes} in the $X$-, $Y$-, $Z$-dimension respectively.

We will show that we can degenerate sub-polynomially many broken copies of $T$ with small fraction of holes in all three dimensions into an unbroken copy of $T$ if $T$ satisfies the following property.

\begin{property}\label{property:fixable}
  There exists a subset $\calG \subseteq \calS_{N_X} \times \calS_{N_Y} \times \calS_{N_Z}$ of permutations over the variables of $T$ where $\calS_N$ denotes the symmetric group on $[N]$, such that $\calG$ satisfies the following:
  \begin{enumerate}
  \item Every $(\pi_X,\pi_Y,\pi_Z) \in \calG$ \emph{preserves the partitions}.
    Specifically, it permutes any part into some entire part, i.e., for every part $X_t\in \calP_X$, there exists $X_{t'}\in \calP_X$ such that $\pi_X(X_t) \defeq \midBK{\pi_X(x) \mid x \in X_t} = X_{t'}$. Similar conditions hold for $Y$- and $Z$-parts. Hence, $\pi_X, \pi_Y, \pi_Z$ also induce permutations on $\calP_X, \calP_Y, \calP_Z$, respectively. \label{item:parts-to-parts}
  \item Every $(\pi_X,\pi_Y,\pi_Z) \in \calG$ \emph{preserves the tensor structure of $T$}. Formally, the coefficient of $x_i \cdot y_j \cdot z_k$ in $T$ equals the coefficient of $\pi_X(x_i) \cdot \pi_Y(y_j) \cdot \pi_Z(z_k)$ in $T$, for all variables $x_i, y_j, z_k$. \label{item:preserve-structure}
  \item \label{item:uniform} A uniformly random element $(\pi_X, \pi_Y, \pi_Z) \in \calG$ \emph{permutes any given part to a uniform random part}. Formally, for any fixed $X_t, X_{t'}\in \calP_X$, $Y_t, Y_{t'}\in \calP_Y$, $Z_t, Z_{t'}\in \calP_Z$, and for a uniformly random $(\pi_X,\pi_Y,\pi_Z)\in \calG$, we have
    \begin{align*}
      \Pr[\pi_X(X_t) = X_{t'}] &= \frac{1}{M_X}, \\
      \Pr[\pi_Y(Y_t) = Y_{t'}] &= \frac{1}{M_Y}, \\
      \Pr[\pi_Z(Z_t) = Z_{t'}] &= \frac{1}{M_Z}.
    \end{align*}
  \end{enumerate}
\end{property}

We show the following. 

\begin{theorem}\label{thm:fix-holes-general}
  Let $T$ be a partitioned tensor defined above; let $T_1,\dots, T_r$ be broken copies of $T$, where in each $T_i$ for $i \in [r]$, at most $\frac{1}{4 \log M_X}$, $\frac{1}{4 \log M_Y}$, and $\frac{1}{4 \log M_Z}$ fraction of $X$-, $Y$-, and $Z$-parts are holes, respectively. If $T$ satisfies \cref{property:fixable} with a set of permutations $\calG$, then there exists a constant $C_0$ such that for $r \ge C_0\cdot M^{\frac{3}{\log \log N}}$ where $M = \max\{M_X, M_Y, M_Z\}$, we have
  \[
    \bigoplus_{i = 1}^r T_i \;\degen\; T.
  \]
  In particular, $M^{o(1)}$ broken copies of $T$ with fraction of holes $O\bigbk{\frac{1}{\log M}}$ can degenerate into an unbroken copy of $T$.
\end{theorem}

Before proving \cref{thm:fix-holes-general}, we first show the following \cref{lem:shufflable} that will explain why we need \cref{item:uniform} in \cref{property:fixable}. The lemma essentially states that if $T$ satisfies \cref{property:fixable}, then we can find a set of permutations $\pi_X, \pi_Y,\pi_Z$ on the partitions of $X$-, $Y$-, $Z$-variables such that any set of parts can be permuted away from any set of positions that we specify. Specifically, one should think under the context of degenerating a broken copy of $T$ with holes into some subtensor $T\vert_{X', Y', Z'}$, the lemma states that we can find a set of permutations preserving the tensor structure of $T$ on the variable sets such that the holes are away from the terms in $T\vert_{X', Y', Z'}$. Then applying the permutation on the broken copy would give the subtensor $T\vert_{X', Y', Z'}$ without holes or with fewer amount of holes.

\begin{lemma}\label{lem:shufflable}
  Let $T$ be a tensor satisfying the assumptions of \cref{thm:fix-holes-general} with $\calG$. Then there exists $(\pi_X, \pi_Y, \pi_Z)\in \calG$ such that for any sets of parts $P_X, P_X'\subseteq \calP_X$, $P_Y, P_Y'\subseteq \calP_Y$, $P_Z, P_Z' \subseteq \calP_Z$ we have
  \begin{align}
    \begin{split}\label{eq:shufflable}
      \abs{P_X \cap \pi_X(P_X')} &\le \frac{4|P_X|\cdot |P_X'|}{|\calP_X|},\\ 
      \abs{P_Y \cap \pi_Y(P_Y')} &\le \frac{4|P_Y|\cdot |P_Y'|}{|\calP_Y|},\\
      \abs{P_Z \cap \pi_Z(P_Z')} &\le \frac{4|P_Z|\cdot |P_Z'|}{|\calP_Z|}.
    \end{split}
  \end{align}
\end{lemma}

\begin{proof}
  We prove the lemma using a probabilistic argument. Consider a uniformly random element $(\pi_X, \pi_Y, \pi_Z)\in \calG$. By \cref{item:uniform} in \cref{property:fixable}, for any $X_t \in \calP_X$, we have
  \[\Pr\Bk{\pi_X(X_t)\in P_X} = \frac{|P_X|}{|\calP_X|}.\]
  By linearity of expectation
  \[\E\Bk{\abs{P_X \cap \pi_X(P_X')}} = \sum_{X_t \in P_X'}\Pr\Bk{\pi_X(X_t) \in P_X} = \frac{|P_X|\cdot |P_X'|}{|\calP_X|}.\]
  Thus by Markov's inequality, have
  \[\Pr\Bk{\abs{\pi_X(X_t)\in P_X} > \frac{4|P_X|\cdot |P_X'|}{\calP_X}} \le \frac{1}{4}.\]
  The argument works similarly for $Y$ and $Z$, so by union bound over $X, Y, Z$, with probability $\ge \frac{1}{4}$, a random $(\pi_X, \pi_Y, \pi_Z)\in \calG$ satisfies \cref{eq:shufflable}. Therefore, we can conclude that there exists such a $(\pi_X, \pi_Y, \pi_Z)\in \calG$ satisfying \cref{eq:shufflable}.
\end{proof}

We now proceed to prove \cref{thm:fix-holes-general}. The main idea is to first take a broken copy of $T$ that covers most of the terms in the tensor, then write the missing terms as a sum of 7 smaller subtensors which we treat as 7 subproblems, and finally recurse on each of the subproblems with smaller sizes.

\begin{proof}[Proof of \cref{thm:fix-holes-general}]
  Assume $P_X, P_Y, P_Z$ are sets of $h_X, h_Y, h_Z$ parts of $X$-, $Y$-, $Z$-dimension respectively, and assume that we need to produce $T \Vert_{P_X, P_Y, P_Z}$. The number of broken copies of $T$ required for this purpose is denoted as $f(h_X, h_Y, h_Z)$. Clearly, $f(h_X, h_Y, h_Z) = 0$ when one of $h_X, h_Y, h_Z$ equals zero (because $T \Vert_{P_X, P_Y, P_Z}$ would be an empty tensor), and we need to upper bound $f(M_X, M_Y, M_Z)$, which is the number of broken copies required to produce a complete copy of $T$.

  Take a broken copy of $T$, namely $\Thole = T \Vert_{\calP_X \setminus P^{(0)}_X, \; \calP_Y \setminus P^{(0)}_Y, \; \calP_Z \setminus P^{(0)}_Z}$, where $P^{(0)}_X$, $P^{(0)}_Y$, $P^{(0)}_Z$ are the set of holes. Then, applying \cref{lem:shufflable} on $P_X, P^{(0)}_X, P_Y, P^{(0)}_Y, P_Z, P^{(0)}_Z$ gives $(\pi_X, \pi_Y, \pi_Z) \in \calS_{N_X} \times \calS_{N_Y} \times \calS_{N_Z}$ such that

  \begin{align}\label{eq:subproblem-size}
    \begin{split}
      \abs{P_X^{(0)'}} &\defeq \abs{P_X \cap \pi_X\bigbk{P_X^{(0)}}} \le 4 \cdot \frac{\abs{P_X} \cdot \bigabs{P_X^{(0)}}}{M_X} \le \frac{1}{\log M_X} \cdot \abs{P_X}, \\
      \abs{P_Y^{(0)'}} &\defeq \abs{P_Y \cap \pi_Y\bigbk{P_Y^{(0)}}} \le \frac{1}{\log M_Y} \cdot \abs{P_Y}, \\
      \abs{P_Z^{(0)'}} &\defeq \abs{P_Z \cap \pi_X\bigbk{P_Z^{(0)}}} \le \frac{1}{\log M_Z} \cdot \abs{P_Z}.
    \end{split}
  \end{align}
  We relabel the variables in $\Thole$ according to the permutations $\pi_X, \pi_Y, \pi_Z$, obtaining another broken copy of $T$ with sets of holes $\pi_X\bigbk{P_X^{(0)}}$, $\pi_Y\bigbk{P_Y^{(0)}}$, $\pi_Z\bigbk{P_Z^{(0)}}$. We then zero out all parts outside $P_X, P_Y, P_Z$. The obtained tensor, denoted by $\Thole'$, is a subtensor of the target tensor $T \Vert_{P_X, P_Y, P_Z}$:
  \[
    \Thole' \;=\; T \Vert_{P_X \setminus P_X^{(0)'}, \; P_Y \setminus P_Y^{(0)'}, \; P_Z \setminus P_Z^{(0)'}} \;\defeq\; T \Vert_{P_X^{(1)'}, \, P_Y^{(1)'}, \, P_Z^{(1)'}},
  \]
  where $P_X^{(0)'} \defeq P_X \cap \pi_X\bigbk{P_X^{(0)}}$ is the set of holes in $X$-parts, and $P_X^{(1)'} \defeq P_X \setminus P_X^{(0)'}$; similar for $Y$- and $Z$-dimension.

  Next, we write $T \Vert_{P_X, P_Y, P_Z}$ as a sum of 8 subtensors:
  \[
    T \Vert_{P_X, P_Y, P_Z} \;=\; T \Vert_{P_X^{(1)'}, \, P_Y^{(1)'}, \, P_Z^{(1)'}} + \sum_{\substack{a, b, c \in \midBK{0, 1} \\ 0 \in \midBK{a, b, c}}} T \Vert_{P_X^{(a)'}, \, P_Y^{(b)'}, \, P_Z^{(c)'}}.
  \]
  Notice that the first term $T \Vert_{P_X^{(1)'}, \, P_Y^{(1)'}, \, P_Z^{(1)'}}$ equals $\Thole'$ (which we already obtained by consuming one broken copy $\Thole$), and the other seven subtensors are significantly smaller than $T \Vert_{P_X, P_Y, P_Z}$, so we can obtain them recursively. The fact that $\bigabs{P_X^{(1)'}} \le |P_X| = h_X$, $\bigabs{P_Y^{(1)'}} \le |P_Y| = h_Y$, $\bigabs{P_Z^{(1)'}} \le |P_Z| = h_Z$ together with \cref{eq:subproblem-size} gives us the following recursion:
  \begin{align*}
    f(h_X, h_Y, h_Z) \le & \, 1 + f\lpr{\frac{h_X}{\log M_X}, h_Y, h_Z} + f\lpr{h_X, \frac{h_Y}{\log M_Y}, h_Z} + f\lpr{h_X, h_Y, \frac{h_Z}{\log M_Z}}\\
                        &+f\lpr{\frac{h_X}{\log M_X}, \frac{h_Y}{\log M_Y}, h_Z} + f\lpr{\frac{h_X}{\log M_X}, h_Y, \frac{h_Z}{\log M_Z}} \\
                        &+ f\lpr{h_X,\frac{h_Y}{\log M_Y}, \frac{h_Z}{\log M_Z}} + f\lpr{\frac{h_X}{\log M_X}, \frac{h_Y}{\log M_Y}, \frac{h_Z}{\log M_Z}}.
  \end{align*}
  Since $\abs{\calP_X} = M_X$, $\abs{\calP_Y} = M_Y$, $\abs{\calP_Z} = M_Z$, we can solve the recursion for $f(M_X, M_Y, M_Z)$ and get 
  \begin{align*}
    f(M_X, M_Y, M_Z) &\le 7^{1 + \ceil{\log_{\log M_X} M_X} + \ceil{\log_{\log M_Y} M_Y} + \ceil{\log_{\log M_Z}M_Z}} \\
                     &\le C_0 \cdot M^{\frac{3}{\log \log M}}
  \end{align*}
  where $C_0$ is a sufficiently large constant and $M =\max \{M_X, M_Y, M_Z\}$ since the function $M^{1/\log \log M}$ is monotonic increasing for sufficiently large $M$.
\end{proof}

We remark that \cref{thm:fix-holes-general} also works for non-partitioned tensors satisfying \cref{property:fixable} when considering on $X$-, $Y$-, $Z$-variables as partitioned into size-$1$ parts where each part consists of a single variable.

Now let us return our attention to the context of fast matrix multiplication and show that we can fix the holes in the interface tensors with holes obtained in our algorithm.

\fixinterface*

\begin{proof}
  Consider the level-$1$ partition of the $X$-, $Y$-, $Z$-variables in $T$ into level-$1$ blocks indexed by sequences in $\{0,1,2\}^N$ with length exactly $N = 2^{\ell-1}\cdot \sum_{t\in [s]}n_t$ as defined in the statement. By definition, the level-$1$ blocks remaining in $T$ are consistent with the distributions $\splresXt, \splresYt, \splresZt$ over each term $t\in [s]$ in $T$, which means that every level-$1$ block $X_{\hat{I}}$ with index sequence $\hat{I}\in \{0,1,2\}^N$ has the same number of $0$'s, $1$'s, and $2$'s. This implies that each level-$1$ $X$-variable block contains the same number of variables and the number of level-$1$ blocks can be bounded by $3^N$. Similarly, there are $\le 3^N$ level-$1$ $Y$- and $Z$-variable blocks and the partitions of $Y$- and $Z$-variables into level-$1$ blocks are partitions into equal-sized parts.

  We let the partitions of $X$-, $Y$-, $Z$-variables into level-1 blocks be the partitions used for \cref{thm:fix-holes-general}, and therefore the number of blocks $M_X, M_Y, M_Z \le 3^N$. Then suppose we can find an appropriate $\calG\subseteq \calS_{|X|}\times \calS_{|Y|} \times \calS_{|Z|}$ satisfying \cref{property:fixable} for $T$, then by \cref{thm:fix-holes-general}, as the fraction of holes in every broken copy $T_i$ is at most $\frac{1}{8N} \le \frac{1}{4 \log 3^N} \le \min\bigBK{\frac{1}{4 \log M_X}, \frac{1}{4 \log M_Y}, \frac{1}{4 \log M_Z}}$ in all three dimensions, a direct sum of $\lpr{3^N}^{\frac{3}{\log \log 3^N}} = 2^{C_1\cdot \frac{N}{\log N}}$ broken copies (with sufficiently large constant $C_1 > 0$) of $T$ can degenerate into an unbroken copy of $T$.

  Thus it suffices to construct a set of permutations $\calG\subseteq \calS_{|X|}\times \calS_{|Y|} \times \calS_{|Z|}$ that together with $T$ satisfies \cref{property:fixable}. Note that every $X$-, $Y$-, or $Z$-variable in $T$ is indexed by a sequence in $\midBK{0, 1, \ldots, q + 1}^{N} = \bigbk{\midBK{0, 1, \ldots, q + 1}^{2^{\l - 1}}}^{n}$, we call every $2^{\l - 1}$ consecutive indices a \emph{chunk} and randomly permute chunks within the same term in $T$. Specifically, consider the set $\calH = \calS_{n_1}\times \dots \times \calS_{n_s}$. For each $\sigma = (\sigma_1,\dots, \sigma_s)\in \calH$, consider that $\sigma_t$ permutes the $n_t$ length-$2^{\ell-1}$ chunks in the $t$-th term for $t\in [s]$. $\sigma$ can be regarded as a permutation over $[n]$, indicating the destinations of all $n$ chunks. It also induces a permutation $\sigma' \in \calS_{N}$ over $N$ level-1 indices. Formally, the $j$-th index in the $i$-th chunk is permuted to the $j$-th index in the $\sigma(i)$-th chunk, i.e., $\sigma'((i - 1) \cdot 2^{\l - 1} + j) = (\sigma(i) - 1) \cdot 2^{\l - 1} + j$ for all $i \in [n]$ and $j \in [2^{\l - 1}]$. Further, $\sigma'$ induces a permutation $\pi_X$ over all $X$-variables, given by
  \[
    \pi_X\bigbk{x_{(\hat i_1, \hat i_2, \ldots, \hat i_N)}} \,\defeq\, x_{(\hat i_{\sigma'(1)}, \hat i_{\sigma'(2)}, \ldots, \hat i_{\sigma'(N)})},
  \]
  where $x_{(\hat i_1, \hat i_2, \ldots, \hat i_N)}$ represents the $X$-variable indexed by $(\hat i_1, \hat i_2, \ldots, \hat i_N) \in \midBK{0, 1, \ldots, q + 1}^{N}$. The permutations $\pi_Y, \pi_Z$ over $Y$- and $Z$-variables are defined similarly. Finally, $\calG$ is defined as all permutations generated in the above way, i.e., $\calG = \midBK{(\pi_X, \pi_Y, \pi_Z) \textup{ induced from } \sigma \in \calH}$.

  Note that $\calG$ is well-defined, since for any element $(\pi_X, \pi_Y, \pi_Z)\in \calG$ and any level-$1$ index sequence $\hat{I}$ in $T$ satisfying the complete split distributions $\midBK{\splresXt}_{t \in [s]}$, $\pi_X(X_{\hat{I}})$ must also satisfy the complete split distributions $\midBK{\splresXt}_{t \in [s]}$, because the permutation acts on each term individually. Now we check that $\calG$ satisfies \cref{property:fixable}. It is easy to see by definition that the set $\calG$ satisfies \cref{item:parts-to-parts} and \cref{item:preserve-structure} since variables in one level-1 variable block all get permuted to the same level-1 variable block. \cref{item:uniform} holds due to the symmetry of the chunks within the same term.
\end{proof}

\newcommand{\numgrow}{V}
\newcommand{\numgrowl}[1][\l]{\numgrow_{#1}}
\newcommand{\interl}[1][\l]{\T_{#1}}

\section{Numerical Result}
\label{sec:numerical}

Let $\l^* > 0$ be an integer and let $N = 2^{\l^* - 1} \cdot n$. Our upper bound of $\omega(1, \kappa, 1)$ is formed by successively applying \cref{thm:global-stage-with-eps,thm:consituent_MM_terms,thm:constituent-stage-with-eps} to degenerate $2^{o(n)}$ independent copies of $\CW_q^{\otimes N} \equiv \bigbk{\CW_q^{\otimes 2^{\l^* - 1}}}^{\otimes n}$ into independent matrix multiplication tensors of the form $\angbk{a, a^{\kappa}, a}$, shown in \cref{alg:framework}.

\begin{figure}[ht]
  \centering
  \begin{tcolorbox}
    \captionof{algocf}{Procedure of Degeneration}{\label{alg:framework}}
    Let $\eps > 0$ be a fixed constant and $\l^* > 0$ be an integer.
    \begin{enumerate}
    \item Degenerate $2^{o(n)}$ independent copies of $\bigbk{\CW_q^{\otimes 2^{\l^* - 1}}}^{\otimes n}$ into $\numgrowl[\l^*]$ (independent) copies of a level-$\l^*$ $(\eps \cdot 3^{\l^*})$-interface tensor $\interl[\l^*]$, where the number of copies $\numgrowl[\l^*]$ and the parameter list of $\interl[\l^*]$ are given in \cref{thm:global-stage-with-eps} and \cref{prop:global-stage-no-eps}.
    \item For each $\l = \l^*, \ldots, 2$:
      \begin{itemize}
      \item Degenerate \emph{every} $2^{o(n)}$ copies of the level-$\l$ $(\eps \cdot 3^{\l})$-interface tensor $\interl[\l]$ into $\numgrowl[\l - 1]$ independent copies of the tensor product of a level-$(\l - 1)$ $(\eps \cdot 3^{\l - 1})$-interface tensor $\interl[\l - 1]$ and some matrix multiplication tensor $\angbk{a_\l, b_\l, c_\l}$. Here, the number of copies $\numgrowl[\l - 1]$, the parameter list of $\interl[\l - 1]$ and the matrix multiplication size $\angbk{a_\l, b_\l, c_\l}$ are all given in \cref{thm:constituent-stage-with-eps} and \cref{prop:constituent-stage-no-eps}.
      \end{itemize}
    \item The level-$1$ $3\eps$-interface tensor $\interl[1]$ can degenerate into a matrix multiplication tensor, written $\angbk{a_1, b_1, c_1}$, according to \cref{thm:consituent_MM_terms}.
    \item So far, we have obtained $\numgrow \defeq \prod_{\l = 1}^{\l^*} \numgrowl[\l]$ copies of $\angbk{A, B, C} \equiv \bigotimes_{\l = 1}^{\l^*} \angbk{a_\l, b_\l, c_\l}$.
    \end{enumerate}
    We first let $n \to \infty$ and apply Sch{\"o}nhage's asymptotic sum inequality 
    (\cref{thm:schonhage-ineq-rect}) on the above degeneration, obtaining a bound on $\omega(1, \kappa, 1)$ which might depend on $\eps$; then, we let $\eps \to 0$, obtaining the bound $\omega(1, \kappa, 1) \le \omega'$ as long as
    \[
      \lim_{\eps \to 0} \lim_{n \to \infty} \numgrow^{1/n} \cdot \min\bigBK{A, B^{1/\kappa}, C}^{\omega'/n} \ge (q + 2)^{2^{\l^* - 1}}.
      \numberthis \label{eq:asymptotic_sum_inequality_in_algorithm}
    \]
  \end{tcolorbox}
  \vspace{-1.5em}
\end{figure}

\bigskip

Every degeneration step in \cref{alg:framework} requires a set of parameters, including the distribution $\alpha$ over constituent tensors, the proportions of tensor powers $A_1, A_2, A_3$ assigned to three regions, and others. If we are given an assignment to the parameters, we can precisely calculate
\[
  \lim_{\eps \to 0} \lim_{n \to \infty} \numgrowl^{1/n}, \quad \lim_{\eps \to 0} \lim_{n \to \infty} a_\l^{1/n}, \quad \lim_{\eps \to 0} \lim_{n \to \infty} b_\l^{1/n}, \quad \lim_{\eps \to 0} \lim_{n \to \infty} c_\l^{1/n}
\]
according to \cref{thm:global-stage-with-eps,thm:consituent_MM_terms,thm:constituent-stage-with-eps}. Plugging them into \eqref{eq:asymptotic_sum_inequality_in_algorithm} would verify the correctness of the claimed bound on $\omega(1, \kappa, 1)$.

\paragraph{Optimization strategy.}

Finding a set of parameters that lead to the best bound of $\omega(1, \kappa, 1)$ can be modeled as a constrained optimization problem:
\[
  \begin{array}{cl}
    \textup{minimize} & \qquad \omega' \\
    \textup{subject to} & \textup{all constraints in \cref{thm:global-stage-with-eps,thm:consituent_MM_terms,thm:constituent-stage-with-eps}} \\
                      & \displaystyle \lim_{\eps \to 0} \lim_{n \to \infty} \numgrow^{1/n} \cdot \min\bigBK{A, B^{1/\kappa}, C}^{\omega'/n} \ge (q + 2)^{2^{\l^* - 1}}.
  \end{array}
  \numberthis \label{eq:optimization_problem}
\]
We used \emph{sequential quadratic programming (SQP)} to solve this optimization problem, which is a well-known iterative approach for solving nonlinear constrained optimization. The software package SNOPT~\cite{SNOPT} is used for performing SQP. Like all other optimization methods for nonlinear optimization, SQP does not guarantee finding the global optimum or a specific convergence rate; the quality of the solution and the time performance both rely on the \emph{initial point} of the iterative process, which could be provided by the user.

For $\kappa = 1$, we take the parameters from \cite{LeGall32power} which Le Gall used to analyze $\CW_q^{\otimes 2^{\l^* - 1}}$ for square matrix multiplication, and transform it into a feasible solution to the optimization problem \eqref{eq:optimization_problem}, which we set as the initial point. Specifically, Le Gall's parameters consist of a distribution $\alpha$ over level-$\l^*$ constituent tensors (for the global stage) together with a split distribution $\alpha_{i,j,k}$ for every constituent tensor $T_{i,j,k}$ (for the constituent stages). We specify our parameters as follows:
\begin{itemize}
\item For every constituent tensor $T_{i,j,k}$ that appears in our interface tensors, we directly set $\alpha_{i,j,k}$ as its split distribution in every region, and let $A_1 = A_2 = A_3 = 1/3$, which means that all three regions are symmetric to each other.
\item The distribution used in our global stage is set to $\alpha$ as well. Other parameters are uniquely determined by these specified ones.
\item For every constituent tensor $T_{i,j,k}$ that contains a zero, say $i = 0$, we choose its complete split distributions $\splresX, \splresY, \splresZ$ that maximizes its size as an inner product tensor, i.e., maximizes $H(\splresY)$.
\item Other parameters are uniquely determined by the specified ones.
\end{itemize}
It is easy to see that these parameters form a feasible solution. Furthermore, these parameters actually lead to the same upper bound on $\omega$ as Le Gall's analysis. We start from this feasible solution and perform SQP to obtain an upper bound for $\omega = \omega(1, 1, 1)$.

For $\kappa \ne 1$, our strategy is to start with a solution for another $\kappa$ nearby. For example, it is natural to believe that a good solution for $\omega(1, 0.95, 1)$ is similar to that for $\omega(1, 1, 1)$. Therefore, we use our parameters for $\omega(1, 1, 1)$ as the initial point for optimizing the bound of $\omega(1, 0.95, 1)$, and proceed with SQP to obtain the bound for $\omega(1, 0.95, 1)$. Then, we can further start with our parameters for $\omega(1, 0.95, 1)$ to obtain parameters for $\omega(1, 0.90, 1)$, and so on.

\paragraph{Lagrange multipliers.} In \cref{thm:global-stage-with-eps}, we need to calculate $P_\alpha = \max_{\alpha' \in D} H(\alpha') - H(\alpha)$ where $D$ represents the set of distributions that share marginals with $\alpha$. Although this definition of $P_\alpha$ is not a closed form in terms of $\alpha$, we can let the max-entropy distribution $\alpha_{\max} \defeq \argmax_{\alpha' \in D} H(\alpha')$ be an optimizable variable, and use the method of Lagrange multipliers to ensure that $\alpha'$ has the largest entropy among $D$.

Formally, we first add linear constraints to force $\alpha_{\max}$ and $\alpha$ to have the same marginals:
\begin{align*}
  \sum_{j + k = 2^{\l^*} - i} \bigbk{\alpha_{\max}(i,j,k) - \alpha(i,j,k)} &= 0, \qquad \forall i = 0, 1, \ldots, 2^{\l^*},
                                                  \numberthis \label{eq:same_marginal_constraint_x} \\
  \sum_{i + k = 2^{\l^*} - j} \bigbk{\alpha_{\max}(i,j,k) - \alpha(i,j,k)} &= 0, \qquad \forall j = 0, 1, \ldots, 2^{\l^*},
                                                  \numberthis \label{eq:same_marginal_constraint_y} \\
  \sum_{i + j = 2^{\l^*} - k} \bigbk{\alpha_{\max}(i,j,k) - \alpha(i,j,k)} &= 0, \qquad \forall k = 0, 1, \ldots, 2^{\l^*},
                                                  \numberthis \label{eq:same_marginal_constraint_z} \\
  \sum_{i,j,k} \alpha_{\max}(i,j,k) &= 1, \numberthis \label{eq:alpha_max_sum=1} \\
  \alpha_{\max}(i,j,k) &\ge 0, \qquad \forall i+j+k = 2^{\lvl^*}. \numberthis \label{eq:alpha_max_nonnegative}
\end{align*}
Let $\lambda_X(i), \lambda_Y(j), \lambda_Z(k), \lambda_S$ ($0 \le i, j, k \le 2^{\l^*}$) be Lagrange multipliers for \eqref{eq:same_marginal_constraint_x}, \eqref{eq:same_marginal_constraint_y}, \eqref{eq:same_marginal_constraint_z}, \eqref{eq:alpha_max_sum=1} respectively, which we also treat as optimizable variables. Then the first-order optimality of $H(\alpha_{\max})$ can be written as
\[
  \lambda_X(i) + \lambda_Y(j) + \lambda_Z(k) + \lambda_S = \ln \alpha_{\max}(i,j,k) + 1, \qquad \forall i + j + k = 2^{\l^*}. \numberthis \label{eq:lagrange_constraint}
\]
(Note that any $\alpha_{\max}$ satisfying \eqref{eq:lagrange_constraint} will also satisfy strict inequalities in \eqref{eq:alpha_max_nonnegative}, thus we do not need to create Lagrange multipliers for \eqref{eq:alpha_max_nonnegative}.)
Since the entropy function $H(\cdot)$ is strictly concave, any $\alpha_{\max}$ satisfying these constraints is guaranteed to have maximum entropy. (Conversely, the true max-entropy distribution $\alpha_{\max}$ will satisfy all these requirements.) We include these Lagrange multiplier constraints \eqref{eq:lagrange_constraint} in our optimization problem \eqref{eq:optimization_problem}.\footnote{In the program, we use the exponential form of \eqref{eq:optimization_problem}: $\exp(\lambda_X(i) + \lambda_Y(j) + \lambda_Z(k) + \lambda_S - 1) = \alpha_{\max}(i,j,k)$, in order to avoid numerical issues like $\ln 0$.} Similarly, in \cref{thm:constituent-stage-with-eps}, we also introduce Lagrange multiplier constraints when we need to ensure that some distribution has maximum entropy given its marginals.

\paragraph{Smooth the landscape.} In \cref{thm:global-stage-with-eps,thm:constituent-stage-with-eps}, the intermediate variables named $E_1, E_2, E_3$ are minimums of three terms. If we calculate them according to the definition, it would create a ``spike'' (non-differentiable point) in the landscape, which is unfriendly for many optimizable methods including SQP. (SQP requires all objective and constraint functions to be twice continuously differentiable.) To address this issue, we treat $E_1, E_2, E_3$ as optimizable variables and transform the minimum into linear inequality constraints:
\[
  E = \min(x, y, z) \qquad \Rightarrow \qquad E \le x, \; E \le y, \; E \le z.
\]
Since $E$ (any of $E_1, E_2, E_3$) is positively correlated with the number of matrix multiplication tensors we produce, we do not need to worry that $E$ takes on a value smaller than $\min(x, y, z)$. The newly introduced constraints are linear and thus have smooth landscapes. We include these auxiliary optimizable parameters and constraints in the optimization problem \eqref{eq:optimization_problem}. In practice, we also observe that SQP would not work well without this type of smoothing.

\paragraph{Numerical results.}
\label{subsec:numerical}

We wrote a MATLAB~\cite{MATLAB2022} program to solve the optimization problem \eqref{eq:optimization_problem}, with the help of SNOPT~\cite{SNOPT}, a software package for solving large-scale optimization problems. By running the program for different $\kappa$, we obtained various upper bounds of $\omega(1, \kappa, 1)$, as shown in \cref{table:result}. All bounds are obtained by analyzing the fourth power\footnote{Our analysis also works for the eighth power, but it was too slow to solve the optimization problem due to the large number of parameters.} of the CW tensor with $q = 5$. Specifically, we obtained the important bounds $\omega \le 2.371552$, $\alpha \ge 0.321334$, and $\mu \le 0.527661$. The code and parameters are available at \url{https://osf.io/7wgh2/?view_only=ce1a6a66d9fc432d8f6da39a6ea4b6e4}.

\bibliographystyle{alpha}
\bibliography{ref}
\end{document}